\documentclass{article}

\usepackage{arxiv}

\usepackage[utf8]{inputenc}
\usepackage[T1]{fontenc}
\usepackage{hyperref}
\usepackage{url}
\usepackage{booktabs}
\usepackage{amsfonts}
\usepackage{amsmath}
\usepackage{amssymb}
\usepackage{nicefrac}
\usepackage{microtype}
\usepackage{graphicx}
\usepackage{subcaption}
\usepackage{xcolor}
\usepackage{algorithm}
\usepackage{algorithmic}
\usepackage{multirow}
\usepackage{makecell}
\usepackage{stmaryrd}
\usepackage{amsthm}
\usepackage{enumitem}
\usepackage{float}
\usepackage[section]{placeins}  % \FloatBarrier at every \section: floats stay in their section
\usepackage{tikz}
\usetikzlibrary{arrows.meta,positioning,shapes.geometric,fit,decorations.pathreplacing}
\usepackage{pgfplots}
\usepgfplotslibrary{fillbetween,groupplots}
\pgfplotsset{compat=1.12}
\usepackage[round,authoryear]{natbib}
\usepackage{listings}


\newtheorem{theorem}{Theorem}
\newtheorem{proposition}[theorem]{Proposition}

\newtheorem{definition}[theorem]{Definition}

\newcommand{\sem}[1]{\llbracket #1 \rrbracket}

\renewcommand{\shorttitle}{Compile Once, Differentiate Everywhere}

\title{\textbf{Compile Once, Differentiate Everywhere:}\\
\textbf{A Differentiable Meta-Circular Interpreter}}

\author{
  \textbf{Lucas Sheneman} \\
  Institute for Interdisciplinary Data Sciences \\
  University of Idaho \\
  \texttt{sheneman@uidaho.edu}
}

\date{}

\begin{document}

% arxiv.sty sets \flushbottom, which stretches vertical glue (e.g. float separation)
% on underfull pages and produces large gaps after figures. Use natural bottom spacing.
\raggedbottom

\maketitle

\vspace{-0.3in}%
\begin{abstract}
\setlength{\parskip}{2pt}%
\vspace{0.28\baselineskip}%
The boundary between program execution and gradient-based optimization has long limited the use of code itself as a learnable scientific model. We present a compiler that translates a substantial, self-hosting subset of Scheme into differentiable computation graphs for standard autograd backends. Because the subset is expressive enough to compile its own evaluator, this yields \emph{differentiable meta-circular interpretation} (DMCI): a compiled Scheme interpreter executes programs supplied as data, while reverse-mode autodiff propagates gradients through the interpreter to continuous constants embedded in those programs. The interpreter is compiled once, so new programs inherit differentiability without recompilation or custom gradient machinery, while retaining ordinary language features such as closures, recursion, and data structures. We prove that gradients through the compiled interpreter are correct almost everywhere and show that they match direct compilation to numerical precision across 171 recursive and higher-order program-seed pairs. We then use DMCI for \emph{program-and-parameter co-search}, where a large language model proposes Scheme programs and exact gradients calibrate their continuous parameters through a single frozen interpreter. This enables a hybrid workflow that can be embedded inside OpenEvolve-style program search: the outer loop proposes discrete program structures, while DMCI supplies exact gradient-based calibration of each candidate’s continuous parameters. On battery capacity-fade data, the search recovers a knee-like degradation structure and improves held-out extrapolation over hand-crafted baselines on the harder early-extrapolation split (matching them on the later split); on a high-dimensional El~Ni\~no inverse problem, DMCI optimizes an interpreted Kalman-filter likelihood where gradient-free search fails. These results extend symbolic regression and neurosymbolic search from closed-form expressions to executable, stateful programs, making model-generated code directly optimizable against data.
\end{abstract}

% ============================================================================
\section{Introduction}
\label{sec:intro}
% ============================================================================

Programs are more than equations: a program describes a \emph{process} that may recurse, branch, allocate data structures, and apply functions received as arguments. The Neural Compiler~\citep{sheneman2026neural} compiled first-order arithmetic expressions into differentiable PyTorch modules with exact gradients, but its source language was deliberately restricted, with no closures, recursion, data structures, or higher-order functions. This paper removes that restriction, extending the compiler to a substantial, self-hosting Scheme subset~\citep{r7rs} with first-class functions, heap-allocated pairs and lists, general recursion, and proper tail calls (Figure~\ref{fig:architecture}). The resulting system compiles its own evaluator, a Scheme interpreter written in Scheme and realized as a differentiable computation graph. When it processes a program containing learnable constants, gradients flow from the output through the interpreter's dispatch, heap operations, and arithmetic back to those constants.

\subsection{Differentiable Meta-Circular Interpretation}

We call this capability \emph{differentiable meta-circular interpretation} (DMCI): the classical Lisp idea of a language implementing its own evaluator~\citep{mccarthy1960recursive,abelson1996structure}, realized in a differentiable setting. The key architectural insight is that \emph{the model is not differentiable; the interpreter is}. A scientific model expressed as a Scheme program is not itself a differentiable computation graph; it is input data to the compiled interpreter, which \emph{is} a differentiable computation graph. Because the interpreter is compiled once and verified once, any program it evaluates automatically inherits differentiability. This inverts the usual approach to differentiable scientific modeling, where each model must be individually reimplemented in a differentiable framework~\citep{aboelyazeed2023,jiang2025canveg}.

DMCI is distinct from prior differentiable interpreters~\citep{gaunt2017terpret,bosnjak2017forth,feser2016differentiable} in three ways: the interpreter is a standard Scheme program \textbf{compiled} automatically rather than a hand-engineered neural architecture; closures, recursion, and data structures are \textbf{inherited} from the source language; and the result is a \textbf{standard autograd} module composable with arbitrary training code.

\subsection{Contributions}

\begin{enumerate}[leftmargin=2em,topsep=2pt,itemsep=1pt]
    \item A compiler for a substantial, self-hosting Scheme subset (closures, \texttt{letrec}, heap-allocated pairs/lists, first-class functions, proper tail calls) targeting differentiable autograd backends (Appendix~\ref{app:compiler}).
    \item \emph{Differentiable meta-circular interpretation}: a compiled meta-circular evaluator that makes runtime-supplied programs differentiable, matching direct compilation to zero relative gradient error and within $7\times10^{-7}$ loss across 171 (program, seed) pairs.
    \item Theory establishing gradient correctness almost everywhere on trace-constant regions of parameter space (Section~\ref{sec:theory}).
    \item Empirical validation of gradient fidelity (Section~\ref{sec:exp_a}), end-to-end optimization of LLM-generated differentiable programs (Section~\ref{sec:exp_b}), recursive scientific models (Section~\ref{sec:exp_c}), and batched scalability yielding $875\times$ throughput and a $3{,}848\times$ population-optimization speedup (Section~\ref{sec:exp_h}); seven further experiments appear in the appendices (Appendices~\ref{app:exp_d}--\ref{app:exp_g},~\ref{app:exp_i}--\ref{app:exp_fluzoo_ext}).
    \item A real-data case study showing that an interpreted Kalman-filter likelihood for a linear inverse model of ENSO can be optimized by exact gradients through the interpreter, recovering a 58--611-parameter fit where gradient-free search fails (Section~\ref{sec:exp_c}, Appendix~\ref{app:exp_lim_enso_ext}).
    \item A second real-data flagship demonstrating LLM-and-DMCI \emph{co-search}: a language model proposes scientific model structures as Scheme programs that one frozen interpreter calibrates by exact gradients, closing discrete structure search and continuous parameter fitting into one loop. On mechanism-labeled battery capacity-fade data the co-search \emph{recovers the correct degradation structure} from a smooth seed, and on real Severson lithium-ion cells~\citep{severson2019} it forecasts held-out fade; a parallel influenza study (FluZoo) surfaces the \emph{fitness-fidelity pitfall} of such co-search and the rigorous-rescoring protocol that detects it (Section~\ref{sec:exp_battery}, Appendices~\ref{app:exp_battery_ext},~\ref{app:exp_fluzoo_ext}).
\end{enumerate}

\subsection{Relationship to Prior Work}
\label{sec:relwork_intro}

This paper directly addresses Limitation~1 of the Neural Compiler~\citep{sheneman2026neural}: the restriction to first-order expressions. Section~\ref{sec:theory} formalizes the extended guarantees: compilation correctness (Theorem~\ref{thm:correctness}), gradient correctness almost everywhere (Theorem~\ref{thm:gradient}), and preservation under composition (Theorem~\ref{thm:composition}). Gradients are correct throughout the regions of parameter space where the program's discrete execution trace remains unchanged. At branch boundaries, where execution paths switch, the source program is itself non-differentiable; the compiled program preserves this behavior rather than introducing new discontinuities.

DMCI is best understood by where it locates differentiation (Section~\ref{sec:related} gives the full comparison). Systems such as Zygote~\citep{innes2019differentiable}, JAX~\citep{jax2018github}, and Pearlmutter and Siskind's VLAD/Stalin-$\nabla$~\citep{pearlmutter2008lambda} perform \emph{language-level AD}, where the differentiated object is host-language code and AD is a transformation applied to that code. A second line moves the locus to the evaluator itself: closest to ours, \citet{maleki2021adding} make Scheme programs differentiable by \emph{hand}-differentiating a meta-circular interpreter with dual numbers for nested forward-mode AD. A third line pushes gradients \emph{through} execution but for a different aim: differentiable interpreters such as TerpreT~\citep{gaunt2017terpret}, $\partial$4~\citep{bosnjak2017forth}, and \citet{feser2016differentiable} are hand-engineered neural architectures for restricted languages that target program \emph{synthesis}, where gradient descent over discrete structure is notoriously hard.

DMCI unifies the interpreter-as-locus idea with a compiled, self-hosting, reverse-mode framing none of these lines reaches. The differentiated object is a standard Scheme evaluator compiled automatically (not host code, not a hand-built neural architecture); the program it runs is supplied as first-class data carrying inherited language semantics (closures, recursion, data structures, self-hosting); and lowering to standard autograd backends inherits reverse-mode AD, GPU acceleration, and batching for free. Its exact gradients optimize continuous parameters within a given structure, sidestepping the pathological landscapes that limit gradient-based synthesis; composed with an outer language-model search over program structure they yield the program-and-parameter co-search of Section~\ref{sec:exp_battery}, while a Gumbel-Softmax relaxation of the interpreter's dispatch (Experiment~E, Appendix~\ref{app:exp_e}) makes the discrete choice itself differentiable.

\begin{figure*}[t]
\centering
\begin{tikzpicture}[
    box/.style={draw, rounded corners=3pt, minimum height=0.7cm, font=\small, align=center, fill=#1},
    box/.default=white,
    arr/.style={-{Stealth[length=5pt]}, thick},
    label/.style={font=\scriptsize, text=black!70},
]

% ===== (a) Direct compilation =====
\node[box=white, minimum width=1.8cm] (src_a) at (0, 1.6) {Scheme\\[-1pt]{\scriptsize source}};
\node[box=blue!8, minimum width=1.7cm] (parse_a) at (2.2, 1.6) {Parse};
\node[box=blue!8, minimum width=1.7cm] (anf_a) at (4.1, 1.6) {ANF\\[-1pt]{\scriptsize transform}};
\node[box=blue!8, minimum width=1.7cm] (graph_a) at (6.0, 1.6) {Compute\\[-1pt]{\scriptsize Graph}};
\node[box=teal!15, minimum width=1.8cm] (mod_a) at (8.1, 1.6) {PyTorch\\[-1pt]{\scriptsize module}};
\node[box=white, minimum width=1.5cm] (fwd_a) at (10.0, 1.6) {Forward\\[-1pt]{\scriptsize $f_\theta(x)$}};

\node[draw, dashed, rounded corners=4pt, gray, fit=(parse_a)(anf_a)(graph_a),
      inner xsep=5pt, inner ysep=5pt,
      label={[font=\scriptsize\bfseries, gray]above:Compiler}] {};

\node[font=\scriptsize\bfseries, text=teal!50!black, anchor=south west, align=left]
    at ([yshift=1pt]src_a.north west) {(a) Direct\\[-1pt]compilation};

\draw[arr] (src_a) -- (parse_a);
\draw[arr] (parse_a) -- (anf_a);
\draw[arr] (anf_a) -- (graph_a);
\draw[arr] (graph_a) -- (mod_a);
\draw[arr] (mod_a) -- (fwd_a);

% ===== (b) DMCI =====
\node[box=white, minimum width=1.8cm] (src_b) at (0, -0.8) {Interpreter\\[-1pt]{\scriptsize source}};
\node[box=blue!8, minimum width=1.7cm] (parse_b) at (2.2, -0.8) {Parse};
\node[box=blue!8, minimum width=1.7cm] (anf_b) at (4.1, -0.8) {ANF\\[-1pt]{\scriptsize transform}};
\node[box=blue!8, minimum width=1.7cm] (graph_b) at (6.0, -0.8) {Compute\\[-1pt]{\scriptsize Graph}};
\node[box=orange!15, minimum width=1.8cm] (mod_b) at (8.1, -0.8) {Compiled\\[-1pt]{\scriptsize interpreter}};
\node[box=white, minimum width=1.5cm] (fwd_b) at (10.0, -0.8) {Forward\\[-1pt]{\scriptsize $\widehat{f}_\theta(x)$}};

\node[draw, dashed, rounded corners=4pt, gray, fit=(parse_b)(anf_b)(graph_b),
      inner xsep=5pt, inner ysep=5pt,
      label={[font=\scriptsize\bfseries, gray]above:Compiler}] {};

\node[font=\scriptsize\bfseries, text=orange!50!black, anchor=south west]
    at ([yshift=1pt]src_b.north west) {(b) DMCI};

\draw[arr] (src_b) -- (parse_b);
\draw[arr] (parse_b) -- (anf_b);
\draw[arr] (anf_b) -- (graph_b);
\draw[arr] (graph_b) -- (mod_b);
\draw[arr] (mod_b) -- (fwd_b);

% Program P(θ) feeding into compiled evaluator
\node[box=white, minimum width=1.8cm] (prog) at (8.1, -2.3) {Program $P(\theta)$\\[-1pt]{\scriptsize (symbolic program)}};
\draw[arr] (prog) -- (mod_b) node[midway, right, label, xshift=2pt] {data};

% ===== (a) Loss and subtle gradient =====
\node[box=red!10, minimum width=1.2cm] (loss_a) at (11.8, 1.6) {$\mathcal{L}$};
\draw[arr] (fwd_a) -- (loss_a);
\draw[-{Stealth[length=5pt]}, red!70!black, thick]
    (loss_a.north) -- ++(0,0.4) -| (mod_a.north);

% ===== (b) Loss and gradient flow (dominant) =====
\node[box=red!10, minimum width=1.2cm] (loss_b) at (11.8, -0.8) {$\mathcal{L}$};
\draw[arr] (fwd_b) -- (loss_b);

\draw[-{Stealth[length=5pt]}, red!70!black, thick]
    (loss_b.south) -- ++(0,-1.8) -| (prog.south);

\end{tikzpicture}
\caption{Two execution modes of the compiler, sharing one pipeline (Parse $\to$ ANF $\to$ Compute Graph) but applying autograd at different levels. \textbf{(a)~Direct compilation}: each Scheme program is compiled to its own differentiable PyTorch module, as in the original Neural Compiler~\citep{sheneman2026neural}, and autograd differentiates through that graph (red). \textbf{(b)~DMCI}: the self-hosted interpreter is compiled \emph{once}; programs are then supplied as symbolic data and executed at runtime, and gradients flow from $\mathcal{L}$ back through the interpreter's dispatch, environment, and heap operations to the parameters $\theta$ in $P(\theta)$ (red). (a)~differentiates a compiled \emph{program}; (b)~differentiates a compiled \emph{interpreter executing} a program, so any new program inherits differentiability without recompilation.}
\label{fig:architecture}
\end{figure*}
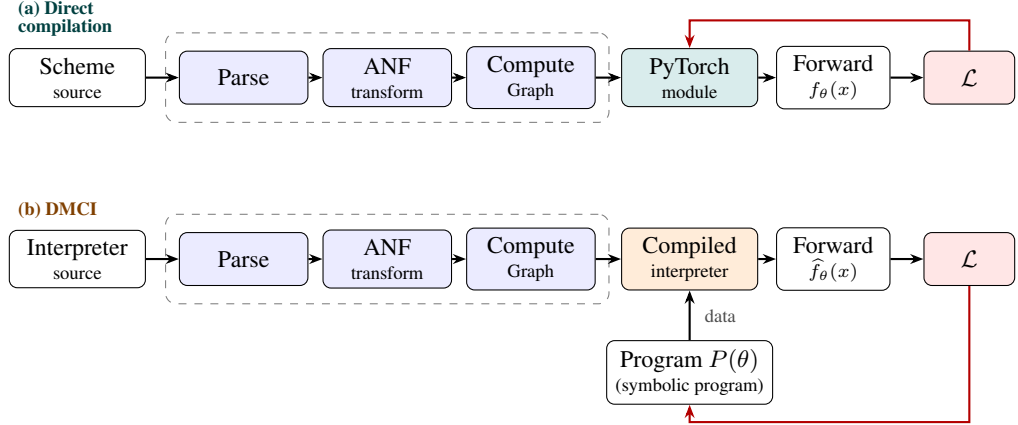

\subsection{Why Scheme?}
\label{sec:why_scheme}

Scheme was chosen because it is one of the smallest practical languages that simultaneously supports self-hosting interpretation, higher-order functional programming, and programs-as-data. These properties make it an unusually clean setting in which to study differentiable meta-circular interpretation.

First, Scheme has a long tradition of compact meta-circular evaluators~\citep{mccarthy1960recursive,abelson1996structure}, allowing the interpreter itself to be expressed in the same language it evaluates. This self-hosting property is central to DMCI: the differentiable object is not an individual scientific model, but a compiled evaluator capable of executing arbitrary runtime-supplied programs.

Second, Scheme provides first-class functions, lexical closures, recursion, and dynamic data structures within a comparatively small semantic core. As a result, the compiler can support language features commonly required by scientific models while remaining tractable to formal analysis. The gradient-correctness results of Section~\ref{sec:theory} are therefore established over a language expressive enough to implement its own evaluator rather than over a minimal arithmetic calculus.

Third, Scheme's homoiconic representation allows executable programs to be represented directly as symbolic data structures. This aligns naturally with DMCI's architecture, in which object programs are supplied to the compiled evaluator as data while continuous parameters are bound separately as differentiable tensors.

Importantly, DMCI is not fundamentally tied to Scheme. The core idea is interpreter-level differentiation: compile a program evaluator once and allow every subsequently supplied program to inherit differentiability. Any language with a suitable self-hosting interpreter and a compilation path into the \texttt{ComputeGraph} intermediate representation could, in principle, support the same approach. Scheme serves as a particularly compact and well-understood vehicle for demonstrating the concept, not as a limitation of the underlying method.

% ============================================================================
\section{Method}
\label{sec:method}
% ============================================================================

\subsection{Compiler Pipeline}

Scheme source is parsed into an AST, transformed to A-normal form (ANF), then compiled to a graph-based intermediate representation (\texttt{ComputeGraph}) of typed nodes (constants, inputs, primitive operations, conditionals, loops, function calls) whose edges carry tagged-value tensors (the 14-dimensional representation of the next subsection). The same graph lowers to four backends that share these semantics: PyTorch and JAX (both differentiable) and NumPy and CuPy (inference-only). Their evaluators walk the graph and execute recursive and iterative control flow dynamically, producing autograd-tracked tensors on the differentiable backends (Appendix~\ref{app:backends}). One pipeline compiles both user programs (direct compilation) and the self-hosted evaluator (the DMCI path), the modes (a) and (b) of Figure~\ref{fig:architecture}.

\paragraph{Backend scope and automatic differentiation.} The shared \texttt{ComputeGraph} lowers to a scalar, straight-line direct-compile path on all four backends, but the differentiable meta-circular \emph{interpreter} (with its dictionary heap, tagged values, and variable-length trampolined \texttt{recur} loop) runs differentiably only on PyTorch. The reason is that DMCI requires gradients to flow through \emph{data-dependent} control flow: which eval-apply clause fires and when the loop terminates are decided by reading structural integers off tagged values via \texttt{.item()}, a deliberate non-differentiable boundary separating discrete, program-determined control from the differentiable tensor computations recorded by autograd. PyTorch's define-by-run autograd records the operations that \emph{actually} execute, so it differentiates the realized trajectory, trampolined loop included, with no restructuring. JAX traces to a functional \texttt{jaxpr} that forbids data-dependent Python branching and whose \texttt{lax.while\_loop} is not reverse-mode differentiable, so a JAX interpreter would require re-architecting the trampoline as a fixed-length \texttt{lax.scan} with masking, future work whose payoff is \texttt{jit} fusion and \texttt{vmap} over the population path (Section~\ref{sec:exp_h}, Appendix~\ref{app:backends}). NumPy is retained as a forward-only reference oracle against which every DMCI result is validated.

\subsection{Tagged Values and Differentiable Heap}

All runtime values are 14-dimensional vectors: a 10-dimensional one-hot type tag (nil, boolean, integer, float, character, symbol, pair, string, closure, vector) concatenated with a 4-dimensional payload. Gradients flow through the payload slots while type tags act as discrete routing signals for dispatch; construction is linear, and for constant learning the program structure (not the learnable constants) determines control flow. Because control flow is program-determined and values are fixed-width tensors, a leading batch dimension broadcasts through the entire interpreter graph unchanged, enabling the batched evaluation of Section~\ref{sec:exp_h}. Pairs and closures are heap-allocated in a dictionary mapping integer addresses to tagged-value tensors: \texttt{cons} allocates fresh keys storing tensors directly and \texttt{car}/\texttt{cdr} return the original reference. Because the language is pure, heap slots are write-once, so PyTorch's autograd graph stays valid through arbitrary allocation sequences, unlike a pre-allocated buffer, whose in-place mutation breaks autograd after a few writes (Appendix~\ref{app:exp_a_ext}).

\subsection{The Differentiable Meta-Circular Evaluator}

The compiler supports closures via heap-allocated environments, general recursion via lazy branch evaluation (a recursion-bearing conditional expands only the branch its arguments select, so \texttt{(if base?\ v (f $\dots$))} unfolds exactly as deep as the recursion runs, rather than diverging), and proper tail calls via a four-level strategy culminating in a runtime trampoline (Appendix~\ref{app:compiler}). The meta-circular evaluator (\texttt{bootstrap/compiler.scm}, 288 lines; Appendix~\ref{app:source_eval}) uses environment lookup via association lists, recursive subexpression evaluation, heap-allocated closures, and symbol-based dispatch. Every construct it uses lies within the compilable subset $\mathcal{L}_{\textrm{DMCI}}$, so the compiler accepts its own implementation: the self-hosting property that closes the DMCI loop.

Compiled, the evaluator is a \emph{fixed, reusable} differentiable module: given a program $P(\theta)$ supplied as data, it computes $\text{eval}(P(\theta), x)$ and reverse-mode autograd computes $\nabla_\theta\,\ell(\text{eval}(P(\theta), x))$ (Theorem~\ref{thm:gradient}, Figure~\ref{fig:architecture}). ``Supplied as data'' means the program is \emph{quoted} into a fixed nested-list datum, while its learnable constants enter \emph{separately} as parameter tensors bound by name in the evaluator's environment (each \texttt{(cons 'name $\theta_i$)}). Because the quoted program structure fixes the interpreter's dispatch and branch trace, reverse-mode gradients propagate through the realized program execution but are taken only with respect to the bound numeric parameter tensors, not the symbolic program structure itself. These gradients traverse tagged-value wrapping, heap storage, variable lookup, and arithmetic dispatch, each a differentiable PyTorch operation. Any subsequently supplied program reuses the same compiled module without recompilation (Appendix~\ref{app:exp_j}).

% ============================================================================
\section{Theoretical Guarantees}
\label{sec:theory}
% ============================================================================

The contribution is not a new automatic-differentiation theorem but a proof that DMCI's specific mechanisms (tagged-value wrap/unwrap, a detached-address heap, environment lookup, closure application, dispatch, and recursion under a fixed trace) preserve the hypotheses for almost-everywhere correctness; full proofs are in Appendix~\ref{app:proofs}. We write $\mathcal{L}_{\textrm{DMCI}}$ for the compiled Scheme subset (full grammar in Appendix~\ref{app:proofs}): variables, constants (including learnable parameters), $\lambda$, application, \texttt{if}, \texttt{letrec}, \texttt{cons}/\texttt{car}/\texttt{cdr}, and arithmetic, comparison, and transcendental primitives. Throughout, $\rho$ is the runtime environment (a variable-to-heap-address map), $H$ the differentiable heap (address-to-tagged-value), and $v_{\textrm{in}}$ the distinguished evaluator input supplied alongside the quoted program (formal domains in Appendix~\ref{app:proofs}, Definition~\ref{def:domains}).

\begin{definition}[Trace-constant region]
For a program $e(\theta)$ with parameters $\theta \in \mathbb{R}^n$, the \emph{trace-constant region} $\Theta_{\textrm{tc}} \subseteq \mathbb{R}^n$ is the set of parameter values where the discrete execution trace (branch decisions and tag dispatches) is locally constant. Because every primitive in $\mathcal{L}_{\textrm{DMCI}}$ is piecewise-analytic and control flow is decided by such (PAP) predicates, the denotation of $\mathcal{L}_{\textrm{DMCI}}$ is an $\omega$-PAP map in the sense of \citet{huot2023omegapap}; consequently each branch boundary is a measure-zero analytic hypersurface, so $\Theta_{\textrm{tc}}$ is open with full Lebesgue measure (Appendix~\ref{app:proofs}). In the meta-circular evaluator the tag dispatches are fixed by the quoted program's symbols and so are constant in $\theta$ \emph{everywhere}; only numeric branch conditions (e.g.\ \texttt{(< x $\alpha$)}) can depend on $\theta$, and these are exactly the measure-zero boundaries of $\Theta_{\textrm{tc}}$.
\end{definition}

\begin{theorem}[Compilation correctness]
\label{thm:correctness}
For every $e \in \mathcal{L}_{\textrm{DMCI}}$, environment $\rho$, heap $H$, and input $v_{\textrm{in}}$: if $\langle e, \rho, H \rangle \Downarrow \langle v, H' \rangle$ in the source semantics, then $\textrm{eval}_{\textrm{compiled}}(\sem{e}, \rho, H, v_{\textrm{in}})$ yields the same tagged value $v$, up to backend floating-point roundoff, whenever source and backend use the same IEEE-754 primitives and evaluation order (made precise in Appendix~\ref{app:proofs}; the two paths are empirically bit-identical, Appendix~\ref{app:exp_c_ext}, Fig.~\ref{fig:c08_loss}).
\end{theorem}

\begin{theorem}[Gradient correctness almost everywhere]
\label{thm:gradient}
Let $e(\theta) \in \mathcal{L}_{\textrm{DMCI}}$ with learnable parameters $\theta$ and terminating evaluation, let $\pi$ project a numeric (float- or integer-tagged) result onto its payload value $p_1$ (stacking to $\mathbb{R}^m$ for a program returning $m$ such results), and let $L(\theta) = \ell(\pi(\textrm{eval}(e(\theta), v_{\textrm{in}})))$ for differentiable $\ell$. Then the gradient computed by reverse-mode autodiff through the compiled interpreter, $\nabla_\theta L_{\textrm{compiled}}(\theta)$, equals the true gradient $\nabla_\theta L(\theta)$ of the denoted real-valued function for almost every $\theta \in \mathbb{R}^n$. Intuitively, the two paths denote the same function (Theorem~\ref{thm:correctness}), and on $\Theta_{\textrm{tc}}$ the DMCI path's extra operations (tag dispatch, environment lookup, heap allocation) are gradient-transparent in $\theta$ (tags are program-determined constants, heap addresses are detached integers retrieved by reference (Proposition~\ref{lem:heap}), and tagged-value wrap/unwrap is linear on the payload), so reverse-mode autodiff computes the same payload Jacobian along both (proof in Appendix~\ref{app:proofs}).
\end{theorem}

\begin{theorem}[Composition preservation]
\label{thm:composition}
Let $e_1 : \mathbb{R}^n \to \mathbb{R}^m$ and $e_2 : \mathbb{R}^m \to \mathbb{R}^k$ denote functions that are gradient-correct almost everywhere, with $e_2$ applied to $e_1$'s output. Then the composite $e_2 \circ e_1$ is gradient-correct almost everywhere: its exceptional set $B(e_1) \cup e_1^{-1}(B(e_2))$ is null by $\omega$-PAP closure under composition~\citep{huot2023omegapap} (Appendix~\ref{app:proofs}).
\end{theorem}

\noindent\textbf{Scoping.} On the open, full-measure set $\Theta_{\textrm{tc}}$ gradients are correct; at the measure-zero branch boundaries the compiled program inherits the source's non-differentiability rather than introducing new discontinuities. A parameter receives no gradient through a discrete comparison itself but can still receive gradient through differentiable uses on the selected branch. Only when a parameter feeds \emph{solely} a branch condition (as $\alpha$ does in \texttt{(if (< x $\alpha$) ($*$ $\beta$ x) ($*$ $\gamma$ x))}) is it unlearnable by descent, which the S3 experiment confirms (Appendix~\ref{app:exp_a_ext}).

% ============================================================================
\section{Experiments}
\label{sec:experiments}
% ============================================================================

We organize the experiments around five questions, each carrying part of the central claim: does interpretation preserve gradients (\S\ref{sec:exp_a})? can generated programs be optimized without per-program implementation (\S\ref{sec:exp_b})? does the method handle recursive and real-data scientific programs (\S\ref{sec:exp_c})? does LLM-and-DMCI co-search discover correct scientific structure on real data (\S\ref{sec:exp_battery})? and can batching amortize interpretation overhead (\S\ref{sec:exp_h})? Seven further experiments (supporting ablations, boundary-case tests, and extended capability demonstrations) are reported in the appendices: structural-search cost, Gumbel-Softmax operator recovery, LLM-in-the-loop discovery, runtime program composition, composite-ecosystem calibration, program-space calibration against compile-each workflows, and influenza program-and-parameter co-search (FluZoo).

% ----------------------------------------------------------------------------
\subsection{Does Interpretation Preserve Gradients? (Experiment A)}
\label{sec:exp_a}
% ----------------------------------------------------------------------------

\paragraph{Setup.} We define six program families of increasing complexity (Table~\ref{tab:programs}). Each family specifies a parameterized Scheme program with one or more learnable constants and target values for those constants; the training targets are the noiseless outputs of the \emph{same} program instantiated with the target constants. For each program, the constants are initialized by seeded Gaussian perturbation away from their targets and optimized with Adam on a fixed 8-point input grid over $x \in [0.5, 3.0]$ (no held-out split), for up to 3000 epochs or until the best training loss falls below $10^{-3}$ (full protocol in Appendix~\ref{app:exp_a_ext}).

\begin{table}[t]
\centering
\caption{Program families for Experiment~A: learnable constants (initialized away from targets) to be recovered by gradient descent, with complexity increasing from scalar arithmetic (P1) to nested function composition (P6). The full Scheme source of each family is in Appendix~\ref{app:exp_a_ext}, Table~\ref{tab:programs_full}.}
\label{tab:programs}
\small
\begin{tabular}{@{}llll@{}}
\toprule
ID & Capability tested & Constants & Targets \\
\midrule
P1 & Scalar arithmetic & $\alpha$ & 0.5 \\
P2 & Independent parameters & $a, b$ & 3.0, 0.5 \\
P3 & Fixed-depth recursion & $\alpha$ & 2.0 \\
P4 & Higher-order closure application & $\alpha$ & 1.5 \\
P5 & Multiple function scopes & $a, b$ & 2.0, 1.5 \\
P6 & Nested composition & $a, b, c$ & 2.0, 0.5, 1.0 \\
\bottomrule
\end{tabular}
\end{table}

The families progress from direct arithmetic (P1--P2) through fixed-depth recursion (P3), higher-order closure application (P4), and separate function scopes (P5) to three-parameter nested composition (P6). Five optimization methods are compared: direct compilation, DMCI, and a hand-coded PyTorch interpreter (all autograd), plus finite differences and a $(\mu,\lambda)$ evolution strategy, with ten seeds per (method, program) pair (300 runs).

\paragraph{Trajectory equivalence.} Across all 60 (program, seed) pairs the three autograd methods match on convergence epoch and final loss to within $7\times10^{-7}$, their full trajectories overlapping to numerical precision (Figure~\ref{fig:exp_a_convergence}): tagged values, heap indirection, environment lookup, and dispatch do not measurably alter optimization. At 50 random parameter settings per program (300 total), the relative gradient error $\|\nabla_{\text{DMCI}} - \nabla_{\text{direct}}\|_2 / \|\nabla_{\text{direct}}\|_2$ is zero to numerical precision and cosine similarity is $1.0$ within $3\times10^{-7}$.

\pgfplotsset{
    conv panel/.style={
        width=6.0cm, height=3.0cm, scale only axis,
        tick label style={font=\scriptsize},
        label style={font=\small},
        title style={font=\normalsize},
        grid=major, grid style={gray!20},
        every axis plot/.append style={mark=none},
    }
}
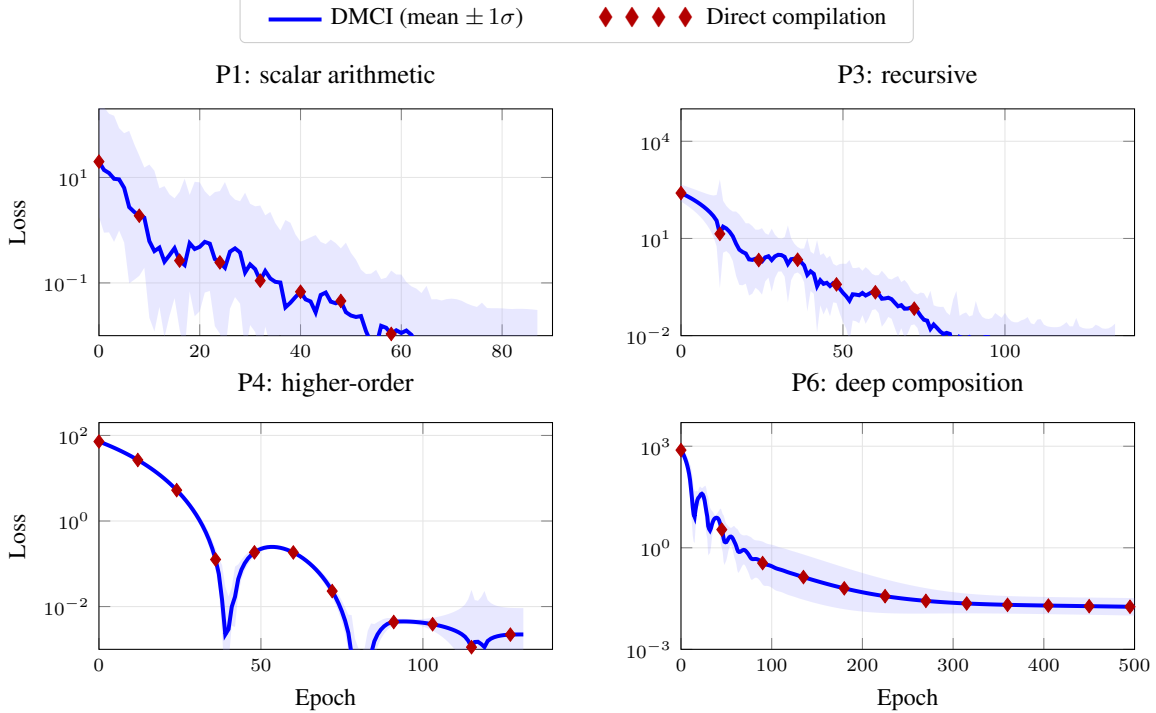
\begin{figure}[tbp]
\centering
% --- Shared legend (hand-built for exact control) ---
\begin{tikzpicture}[baseline=(box.center)]
\node[draw=gray!40, rounded corners=2pt, fill=white, inner xsep=10pt, inner ysep=5pt] (box) {%
    \raisebox{-0.5pt}{\tikz[baseline=-0.5ex]{\draw[blue, line width=1.5pt] (0,0) -- (0.7cm,0);}}%
    \;\small DMCI (mean $\pm\,1\sigma$)%
    \hspace{28pt}%
    {\small\textcolor{red!70!black}{$\blacklozenge$\;\,$\blacklozenge$\;\,$\blacklozenge$\;\,$\blacklozenge$}}%
    \;\;\small Direct compilation%
};
\end{tikzpicture}
\vspace{4pt}

% --- 2x2 grid: groupplot guarantees aligned panels and tight, symmetric spacing ---
% The +/-1 sigma uncertainty band is in LOG10 space (see make_exp_a_data.py): convergence
% curves span several orders of magnitude, so log-space statistics give an interpretable
% *multiplicative* band.  It is drawn as a closed polygon (<tag>_band.dat) via \closedcycle
% rather than the pgfplots `fill between' library, which corrupts groupplot cell layout in
% this pgfplots version.  Band is drawn first so it sits behind the mean line and markers.
\begin{tikzpicture}
\begin{groupplot}[
    group style={group size=2 by 2, horizontal sep=1.7cm, vertical sep=1.15cm},
    conv panel, ymode=log,
]
% --- P1: scalar arithmetic (top-left) ---
\nextgroupplot[title={P1: scalar arithmetic}, ylabel={Loss}, xmin=0, xmax=90, ymin=1e-2, ymax=200]
\addplot[fill=blue!40, opacity=0.22, draw=none, forget plot] table[x=x, y=y] {figures/exp_a_p1_single_const_compiled_interp_band.dat} \closedcycle;
\addplot[blue, line width=1.5pt] table[x=epoch, y=mean] {figures/exp_a_p1_single_const_compiled_interp.dat};
\addplot[red!70!black, only marks, mark=diamond*, mark size=2.5pt, mark repeat=8] table[x=epoch, y=mean] {figures/exp_a_p1_single_const_direct.dat};
% --- P3: recursive (top-right) ---
\nextgroupplot[title={P3: recursive}, xmin=0, xmax=140, ymin=1e-2, ymax=1e5]
\addplot[fill=blue!40, opacity=0.22, draw=none, forget plot] table[x=x, y=y] {figures/exp_a_p3_recursive_compiled_interp_band.dat} \closedcycle;
\addplot[blue, line width=1.5pt] table[x=epoch, y=mean] {figures/exp_a_p3_recursive_compiled_interp.dat};
\addplot[red!70!black, only marks, mark=diamond*, mark size=2.5pt, mark repeat=12] table[x=epoch, y=mean] {figures/exp_a_p3_recursive_direct.dat};
% --- P4: higher-order (bottom-left) ---
\nextgroupplot[title={P4: higher-order}, ylabel={Loss}, xlabel={Epoch}, xmin=0, xmax=140, ymin=1e-3, ymax=200]
\addplot[fill=blue!40, opacity=0.22, draw=none, forget plot] table[x=x, y=y] {figures/exp_a_p4_higher_order_compiled_interp_band.dat} \closedcycle;
\addplot[blue, line width=1.5pt] table[x=epoch, y=mean] {figures/exp_a_p4_higher_order_compiled_interp.dat};
\addplot[red!70!black, only marks, mark=diamond*, mark size=2.5pt, mark repeat=12] table[x=epoch, y=mean] {figures/exp_a_p4_higher_order_direct.dat};
% --- P6: deep composition (bottom-right) ---
\nextgroupplot[title={P6: deep composition}, xlabel={Epoch}, xmin=0, xmax=500, ymin=1e-3, ymax=5000]
\addplot[fill=blue!40, opacity=0.22, draw=none, forget plot] table[x=x, y=y] {figures/exp_a_p6_composed_compiled_interp_band.dat} \closedcycle;
\addplot[blue, line width=1.5pt] table[x=epoch, y=mean] {figures/exp_a_p6_composed_compiled_interp.dat};
\addplot[red!70!black, only marks, mark=diamond*, mark size=2.5pt, mark repeat=45] table[x=epoch, y=mean] {figures/exp_a_p6_composed_direct.dat};
\end{groupplot}
\end{tikzpicture}
\caption{Convergence trajectories for four program families (scalar P1, recursive P3, higher-order P4, deep composition P6): mean (solid blue) $\pm\,1\sigma$ over 10 seeds (shaded), log-scaled loss. DMCI and direct compilation produce near-identical dynamics, so tagged-value dispatch, heap indirection, and evaluator recursion preserve gradient fidelity; P6's longer convergence reflects landscape complexity, not gradient degradation (its panel is truncated at 500 epochs, descending to ${\sim}0.0008$ by epoch ${\sim}2{,}017$). Full five-method comparison in Appendix~\ref{app:exp_a_full}, Figure~\ref{fig:exp_a_convergence_full}.}
\label{fig:exp_a_convergence}
\end{figure}

Each of the three autograd methods converges on all 60 program-seed pairs (100\%; Wilson 95\% lower bound 94.0\%), whereas finite differences converge 50/60 (83\%) at $12.5$--$534\times$ the wall time and a $(\mu,\lambda)$ evolution strategy only 27/60 (45\%), failing on the multi-parameter programs P2 and P6 (0/10 each) and P5 (1/10). Convergence cost scales with landscape complexity rather than gradient degradation (P1 ${\sim}20$ epochs through P6 ${\sim}2{,}017$, all reaching 100\%). Per-program DMCI results (Table~\ref{tab:exp_a_results}), per-method wall-clock (Fig.~\ref{fig:exp_a_wallclock}), and the full five-method baselines (Table~\ref{tab:exp_a_full}) are in Appendix~\ref{app:exp_a_full}.

\paragraph{The payoff: programs as differentiable data.} Sequential DMCI is ${\sim}14\times$ slower per evaluation than direct compilation (decomposed in Appendix~\ref{app:exp_a_ext}, Table~\ref{tab:overhead_decomp}). What it buys is not better optimization (the trajectories coincide to numerical precision) but \emph{programs as manipulable differentiable data}: changing the object program requires no regeneration, retracing, or re-engineering, since it remains an S-expression consumed by one compiled interpreter. Programs can thus be generated by an LLM (Experiment~B), swapped at runtime, or composed programmatically while reusing the same evaluator, with gradient correctness following from the interpreter's correctness almost everywhere (Theorem~\ref{thm:gradient}) rather than per-program verification. This trajectory equivalence holds across all 171 autograd (program, seed) pairs in Experiments~A--C,\footnote{The 171 autograd (program, seed) pairs comprise 60 from Experiment~A ($6$ programs $\times\,10$ seeds), 75 from Experiment~B ($15$ programs $\times\,5$ seeds), and 36 from Experiment~C (six models $\times\,5$ seeds plus two models with three completed seeds).} so the flexibility costs nothing measurable in optimization quality.

% ----------------------------------------------------------------------------
\subsection{Can Generated Programs Be Optimized Without Per-Program Implementation? (Experiment B)}
\label{sec:exp_b}
% ----------------------------------------------------------------------------

The compiler's interface (string in, differentiable module out) creates a natural integration point for large language models. After iterative refinement of a shared, compiler-constrained system prompt (binary-only operators, explicit \texttt{letrec} recursion, numeric-only control flow), a locally-hosted LLM (Qwen3.6-35B) produced compilable Scheme for all 15 selected model descriptions on the final prompt (Appendix~\ref{app:exp_b_ext}); this is constrained, prompt-refined generation, not open-ended autonomous synthesis. Models span electrostatics, spectroscopy, enzyme kinetics, oscillatory dynamics, population dynamics, signal processing, numerical methods, and polynomial evaluation, with 1--4 learnable parameters each (full list in Appendix~\ref{app:exp_b_ext}, Table~\ref{tab:exp_b_models}). Four use recursion or iteration: Euler ODE integration (10-step loop), Taylor series (9-term summation), the recursive filter (8 steps), and Newton's method (5 iterations); one more (M13, composed transforms) uses closures and higher-order composition.

\paragraph{Results.} We compare four methods: DMCI, direct compilation, a hand-coded PyTorch implementation of the same model, and a pure two-hidden-layer MLP, all trained with Adam over five seeds per model (full protocol in Appendix~\ref{app:exp_b_ext}). DMCI and direct compilation match on convergence epoch for \emph{all} 75/75 (model, seed) pairs with zero final-loss difference, and the three physics-informed methods converge on every run (100\%), whereas the MLP converges on only 65/75 (87\%), failing entirely on Coulomb's law ($1/r^2$ singularity). Where both converge, the MLP's mean loss is up to ${\sim}4{,}100\times$ worse on models whose compiled form encodes the exact functional structure but comparable on smooth, monotonic targets, illustrating the value of structural inductive bias rather than universal superiority. Per-model losses, the loss comparison (Figure~\ref{fig:exp_b_loss}), and the ${\sim}73\times$ DMCI overhead are in Appendix~\ref{app:exp_b_ext}.

% ----------------------------------------------------------------------------
\subsection{Recursive Scientific Models and a Real-Data Case Study (Experiment C, LIM-ENSO)}
\label{sec:exp_c}
% ----------------------------------------------------------------------------

Many scientific models are inherently algorithmic: recursion, iteration, and coupled state updates awkward to flatten into first-order arithmetic. We test eight such models (coupled ODEs Lotka--Volterra, SIR, and decay chain at $20$--$30$ steps; iterative maps; recursive filters) under the Experiment~B protocol, yielding $36$ completed (model, seed) pairs per method (full descriptions in Appendix~\ref{app:exp_c_ext}). DMCI and direct compilation again match on convergence epoch for all $36/36$ pairs and all $36$ DMCI runs converge, extending gradient fidelity to recursive programs with up to $30$ state-threading steps. Recursive structure provides strong inductive bias when the compiled program matches the target dynamics (cascaded EMA: $1{,}102\times$ mean loss gap vs.\ MLP, ${\sim}2\times$ by median; decay chain $42\times$), though two smooth-target models (continued fraction, damped pendulum) favor the MLP. Wall-clock cost is governed by optimization difficulty rather than recursion depth: C02/C03 use the most steps yet finish in under an hour, while C01/C06 require ${\sim}1{,}400$--$1{,}800$ epochs (${\sim}7$--$8$\,h).

\paragraph{Real-data case study: a Kalman-filter MLE for ENSO.} To answer the synthetic-only objection
that Experiments~A--C recover constants from self-generated data, we fit a real, multivariate
dynamical-systems maximum-likelihood problem: a Linear Inverse Model (LIM) of the El~Ni\~no--Southern
Oscillation, with the \emph{entire} Kalman-filter negative log-likelihood for ERSSTv5~\citep{huang2017extended}
tropical sea-surface-temperature anomalies written once as a Scheme program
(Appendix~\ref{app:exp_lim_enso_ext}) and folded through the compiled interpreter. The transition operator
$F$ and covariances $Q,R$ enter as bound parameter tensors and reverse-mode autograd yields \emph{exact}
gradients of the likelihood. This is genuinely high-dimensional ($k=D^2+D(D{+}1)/2+1$ free parameters,
from $58$ at $D{=}6$ to $\mathbf{611}$ at $D{=}20$), and exact DMCI gradients make it tractable where
gradient-free search does not: both gradient optimizers (Adam, multi-start L-BFGS) converge reliably and
recover stable operators ($\rho(F)<1$), while batched differential evolution never minimizes the
likelihood and erodes with dimension (Table~\ref{tab:lim_main}). The fitted operator carries the canonical
ENSO signature and attains held-out Nino-3.4 anomaly correlation $0.90/0.65/0.50/0.51$ at the
$3/6/9/12$-month leads, beating persistence at every lead and matching a hand-built Green-function
reference. The value here is not better gradients than a hand-written PyTorch Kalman filter but
\emph{zero per-structure implementation}: a language model proposes alternative $F$ structures (dense,
diagonal, low-rank, AR(2), symmetric/antisymmetric) as one-line edits to the same $17$-line program,
which the \emph{same} frozen interpreter calibrates and a held-out criterion selects among. Full
methodology, exact train/test split, fairness of the Green-function and gradient-free baselines, and the
optimizer/dimension sweep are in Appendix~\ref{app:exp_lim_enso_ext}.

\begin{table}[t]
\centering
\caption{Solver comparison for the dense LIM operator: all solvers optimize the \emph{same} compiled
Kalman-filter NLL objective, so only the optimization method varies. Train NLL is shown at $D{=}6$ and
$D{=}20$ (lower better); ACC@3, spectral radius $\rho(F)$, stability ($\rho(F)<1$), and seed-robustness
are at the headline $D{=}10$. The two exact-gradient methods (Adam, multi-start L-BFGS) recover stable
operators where gradient-free differential evolution does not. Full $D\in\{6,10,15,20\}$ results in
Appendix~\ref{app:exp_lim_enso_ext}, Table~\ref{tab:lim_enso_solver}.}
\label{tab:lim_main}
\small
\begin{tabular}{@{}lccccc@{}}
\toprule
 & \multicolumn{2}{c}{Train NLL $\downarrow$} & \multicolumn{3}{c}{At headline $D{=}10$} \\
\cmidrule(lr){2-3}\cmidrule(lr){4-6}
Solver & $D{=}6$ & $D{=}20$ & ACC@3 $\uparrow$ & $\rho(F)$ & Stable / Robust \\
\midrule
Exact-grad Adam (DMCI)           & $-2060$ & $-4494$  & $0.90$ & $0.96$ & Yes / $3/3$ \\
Multi-start L-BFGS (DMCI)        & $-2065$ & $-4551$  & $0.90$ & $0.96$ & Yes / $3/3$ \\
Differential evolution (gradient-free) & $+636$  & $+11383$ & $0.87$ & $1.08$ & No / $0/3$ \\
\addlinespace
Green-function $G(\tau)$ (reference)   & n/a     & n/a      & $0.91$ & $0.96$ & Yes / n/a \\
\bottomrule
\end{tabular}
\end{table}

% ----------------------------------------------------------------------------
\subsection{Does Co-Search Discover Better Scientific Structure? (Experiment L, Battery Capacity Fade)}
\label{sec:exp_battery}
% ----------------------------------------------------------------------------
% Synthetic-recovery numbers are FINAL (rigorous iters=300 re-score, results/battery_rescore.json).
% Real-Severson forecast + ablation numbers are FINAL (results/battery_rescore_real_ks*.json,
% after the three real-data islands + the iters=300 re-score complete.

The LIM-ENSO study (\S\ref{sec:exp_c}) calibrated a handful of hand-seeded structural variants by hand; the natural next step is to let a language model propose those variants at scale, with exact DMCI gradients calibrating each. This is a \emph{co-search} of a hybrid space (the model searches the discrete space of \emph{programs} while reverse-mode gradients search each program's continuous parameters), and it closes into one loop only because every runtime-supplied program inherits differentiability from the single compiled interpreter, with no per-program reimplementation or recompilation. Structurally this is a \emph{bilevel} optimization: a non-differentiable outer search over program structure wrapping an exact-gradient inner calibration of each structure's constants, the canonical form of joint structure-and-parameter search~\citep{liu2019darts,cranmer2023pysr,chaudhuri2021neurosymbolic}, and distinct from the gradient-through-dispatch relaxation of Experiment~E (Appendix~\ref{app:exp_e}). Gradients search parameters, the language model searches structure, and DMCI's contribution is that the inner calibration is exact and reimplementation-free for \emph{any} structure the outer loop proposes. Such a search raises a question of \emph{fidelity}: when it prefers a structure, is that structure \emph{correct}, or merely better-scoring under a finite-budget proxy? Real forecasting data cannot tell the two apart, because the generating mechanism is unknown. We therefore study battery capacity fade, which admits \emph{both} a mechanism-labeled synthetic ground truth and a real held-out forecast, and run the same co-search on each: a language model, through OpenEvolve~\citep{openevolve} (an open implementation of AlphaEvolve~\citep{alphaevolve}), edits a per-cycle degradation program (the capacity rollout $Q(k)$ written as a Scheme loop) while DMCI gradients calibrate its rate parameters and held-out forecast skill selects.

\paragraph{Synthetic recovery (ground truth).} We generate mechanism-labeled cells from a two-reservoir bottleneck, $Q=q_0\min(1-a\sqrt{k},\,1-ck)$, whose $\min$ of a fast (square-root, SEI-like) and a slow (linear) loss produces a sharp, non-smooth \emph{knee}; the search is seeded with the smooth $Q=q_0-B\sqrt{k}$, which cannot express a knee. Across three independent islands (MAP-Elites quality-diversity over structural complexity $\times$ knee-capability) the co-search evaluates $600$ programs spanning $119$ structurally distinct rollouts (canonical-AST count, Appendix~\ref{app:exp_battery_ext}) and rediscovers the knee mechanism from the smooth seed: two of the three islands converge to programs that are algebraically the generating $\min(q_0-B_1\sqrt{k},\,q_0-B_2 k)$ family (a knee the smooth seed provably cannot express) and, notably, write their own encoding of it rather than copying the hand-written reference. Under a rigorous re-calibration ($300$ Adam steps, versus the $60$ used as the cheap search-time fitness) the recovered structure forecasts the held-out cycle tail at RMSE $0.0111$, which is $2.7\times$ better than the smooth seed ($0.0301$) and below even the hand-written two-reservoir ($0.0143$) and sigmoidal ($0.0149$) references (Table~\ref{tab:battery}, top). Because this is a recovery of the correct \emph{structure} (the two winners are algebraically the generating $\min(\cdot,\cdot)$ family up to parameter renaming, not merely a lower error margin), it is robust to the fitness-fidelity effect isolated below.

\paragraph{Real data (held-out forecast skill).} We then replace the synthetic target with $117$ commercial LFP/graphite cells from the \citet{severson2019} fast-charging dataset (per-cycle discharge capacity, normalized to state of health $Q/Q_0$). These cells fade negligibly over their first ${\sim}100$ cycles and accelerate only near end of life, so each cell's trajectory is resampled onto a common life-fraction grid spanning its full observed cycle life (Figure~\ref{fig:battery_real}a); an absolute early-cycle window is a structure-blind negative control that separates no families (Figure~\ref{fig:battery_real}b; Appendix~\ref{app:exp_battery_ext}). On the held-out forecast, where we fit each cell's early life and predict its accelerating tail, the co-search's selected program (a discovered three-reservoir bottleneck) attains tail RMSE $0.051$ (late split)\,/\,$0.034$ (early split), beating the best hand-crafted model (a sigmoidal knee, $0.051$/$0.057$) by $1.7\times$ on the hard early-extrapolation split and matching it on the late split, and beating the smooth structure families (square-root SEI and power-law) and the naive persistence and linear-extrapolation floors at both (Table~\ref{tab:battery}, bottom). The search explored $160$ structurally distinct rollouts on the real cells (Appendix~\ref{app:exp_battery_ext}); a more flexible discovered structure reaches the best late-split skill ($0.027$) but extrapolates poorly on the early split, a concrete instance of the fitness-fidelity caution that motivates selecting on the harder split. On real cells we claim forecast skill, \emph{not} mechanism recovery: the real rollover is soft, and its physical driver (lithium plating following negative-electrode active-material loss~\citep{attia2022knees}) is latent and not identifiable from a capacity curve alone.

\begin{table}[t]
\centering
\caption{Battery co-search results: held-out forecast RMSE after rigorous 300-step DMCI calibration of each candidate structure.
\textbf{Top:} on mechanism-labeled synthetic data, the search starts from a smooth square-root fade law and recovers a knee-capable structure from the generating family, reducing held-out RMSE relative to the smooth seed and hand-written knee baselines.
\textbf{Bottom:} on real Severson cells, where the true mechanism is unknown, the selected co-search program improves early-tail extrapolation over smooth-family, hand-knee, and naive baselines while matching the best hand-knee model on the late split.
The table separates the two claims: synthetic recovery shows that the search can rediscover correct structure when ground truth is available, while the real-data result shows that runtime-supplied programs calibrated through one frozen DMCI interpreter can improve held-out forecast skill.}
\label{tab:battery}
\small
\begin{tabular}{@{}lcc@{}}
\toprule
\multicolumn{3}{@{}l}{\emph{Synthetic (two-reservoir knee target): held-out RMSE}} \\
Structure & Knee-capable & RMSE $\downarrow$ \\
\midrule
Smooth $\sqrt{t}$ seed                & no  & $0.0301$ \\
Hand two-reservoir (true family)      & yes & $0.0143$ \\
Hand sigmoidal knee                   & yes & $0.0149$ \\
\textbf{Co-search (recovered)}        & yes & $\mathbf{0.0111}$ \\
\midrule
\multicolumn{3}{@{}l}{\emph{Real Severson cells: held-out tail RMSE}} \\
Method & Late split & Early split \\
\midrule
Persistence (naive)                   & $0.084$ & $0.089$ \\
Linear extrapolation (naive)          & $0.080$ & $0.078$ \\
Best smooth family (DMCI)             & $0.083$ & $0.086$ \\
Best hand knee (sigmoidal, DMCI)      & $0.051$ & $0.057$ \\
\textbf{Co-search (selected, DMCI)}   & $0.051$ & $\mathbf{0.034}$ \\
\bottomrule
\end{tabular}
\end{table}

\begin{figure}[t]
\centering
% (a) life-fraction grid -- the real forecast setup
\begin{tikzpicture}
\begin{axis}[
  width=6.2cm, height=5.0cm,
  title={(a) life-fraction grid}, title style={font=\small},
  xlabel={Normalized life index $k$}, ylabel={State of health $Q/Q_0$},
  label style={font=\small}, tick label style={font=\small},
  xmin=0, xmax=100, ymin=0.76, ymax=1.04,
  xtick={0,20,40,60,80,100}, ytick={0.8,0.9,1.0},
  grid=major, grid style={gray!18},
]
  \addplot[name path=lfa, draw=none] table[x=k,y=p10] {figures/battery_real_lifefrac.dat};
  \addplot[name path=lfb, draw=none] table[x=k,y=p90] {figures/battery_real_lifefrac.dat};
  \addplot[blue!12] fill between[of=lfa and lfb];
  \addplot[blue, line width=1pt] table[x=k,y=mean] {figures/battery_real_lifefrac.dat};
  \draw[dashed, gray!70] (axis cs:70,0.76) -- (axis cs:70,1.04);
  \node[font=\scriptsize, text=black!55] at (axis cs:34,0.805) {fit};
  \node[font=\scriptsize, text=black!55] at (axis cs:86,0.805) {forecast};
\end{axis}
\end{tikzpicture}%
\hspace{3mm}%
% (b) absolute first-100-cycle grid -- the structure-blind negative control (shares the y-axis of (a))
\begin{tikzpicture}
\begin{axis}[
  width=6.2cm, height=5.0cm,
  title={(b) absolute first-$100$-cycle grid}, title style={font=\small},
  xlabel={Cycle index $k$},
  label style={font=\small}, tick label style={font=\small},
  xmin=0, xmax=100, ymin=0.76, ymax=1.04,
  xtick={0,20,40,60,80,100}, ytick={0.8,0.9,1.0}, yticklabels={},
  grid=major, grid style={gray!18},
]
  \addplot[name path=aba, draw=none] table[x=k,y=p10] {figures/battery_real_absolute.dat};
  \addplot[name path=abb, draw=none] table[x=k,y=p90] {figures/battery_real_absolute.dat};
  \addplot[blue!12] fill between[of=aba and abb];
  \addplot[blue, line width=1pt] table[x=k,y=mean] {figures/battery_real_absolute.dat};
  \node[font=\scriptsize, text=black!55, align=center] at (axis cs:50,0.90) {mean fade\\$<0.002$ SOH};
\end{axis}
\end{tikzpicture}
\caption{Why the battery forecast task uses a life-fraction grid. Curves show mean state of health $Q/Q_0$ across the 117 Severson cells, with shading marking the 10--90th percentile range; both panels share the same $Q/Q_0$ axis. \textbf{(a)}~Resampling each cell over its observed lifetime exposes the flat early regime and late-life rollover that the held-out forecast must predict; the dashed line marks the fit/forecast split at $k_{\mathrm{split}}=70$. \textbf{(b)}~The absolute first-100-cycle window is nearly flat, with mean fade below $0.002$ SOH, and therefore provides little signal for distinguishing degradation structures. The life-fraction grid makes the held-out tail structurally informative: candidate programs calibrated through the same frozen DMCI interpreter are selected by whether their structure predicts late-life curvature, not by how well they fit an uninformative early flat region.}
\label{fig:battery_real}
\end{figure}
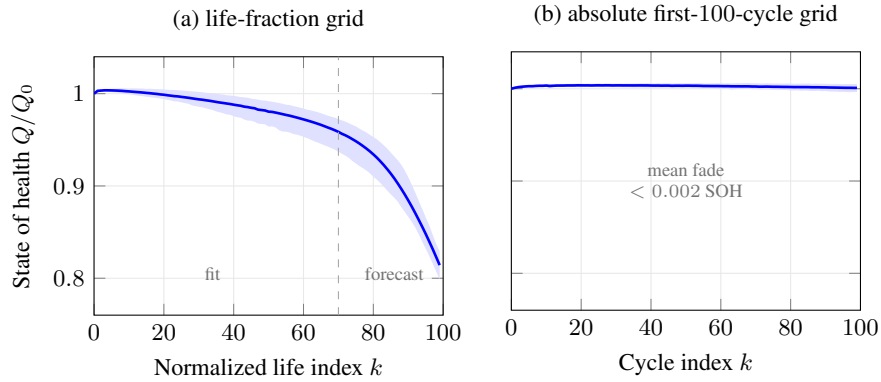

\paragraph{What did the interpreter contribute? (inner-fitter ablation).} To separate OpenEvolve's structure search from DMCI's differentiable calibration, we hold each discovered structure, the batched interpreter, the per-cell parameterization ($468$--$702$ free parameters), and the compute budget fixed, and replace DMCI's exact-gradient Adam with gradient-free differential evolution (the Experiment~I baseline, Appendix~\ref{app:exp_i}). Held-out skill collapses: on the identical structures the gradient-free fit scores $9$--$27\times$ worse (island0 $0.072\!\to\!1.98$ and island1 $0.027\!\to\!0.25$ at the late split), because it cannot calibrate the high-dimensional per-cell fit within matched compute, consistent with the dimension scaling of Experiment~I. OpenEvolve proposes the structure; DMCI's exact gradients are what make that structure forecastable, and neither alone suffices. The comparison is at matched wall-clock: converging the per-cell fit by gradient-free search would need orders of magnitude more evaluations, which is itself the point.

\paragraph{When co-search overfits its own fitness: a real-data stress test.} The identical machinery on real influenza forecasting exposes an instructive failure mode. In FluZoo, a Qwen3.6-35B zoo of compartmental, seasonal, and regionally-coupled epidemic programs is calibrated through the interpreter on CDC ILINet weighted influenza-like illness and selected on held-out seasons (full study in Appendix~\ref{app:exp_fluzoo_ext}). Under the cheap $60$-step search-time fitness the validation-selected evolved program appeared to beat a hand-written regional SEIR; re-scored under the \emph{same} rigorous $300$-step calibration, its advantage \emph{vanished}, tying SEIR on both validation ($0.0148$ vs.\ $0.0149$) and test ($0.0181$ vs.\ $0.0181$). The under-converged inner fitness acted as an implicit regularizer that flattered the more flexible program, and the outer search optimized that loose proxy rather than true skill. The lesson, that the inner-loop fitness must be converged enough to be a faithful skill signal or the outer structure search overfits it, is exactly why the battery result is reported at the rigorous budget and anchored, on its synthetic leg, to \emph{structural} recovery against ground truth rather than a marginal RMSE margin. The honest scope of program-and-parameter co-search is this pair: it recovers correct structure when the data discriminate it, and degrades to the baseline when they do not.

\paragraph{The payoff: program search as data.} Experiment~A established that one program becomes differentiable for free; the co-search shows that \emph{a search over programs} does. Because each candidate mechanism is an S-expression consumed by one verified interpreter rather than a hand-built differentiable module, a language model can propose hundreds of structurally distinct models and a single fit loop calibrates them all, with gradient correctness following from the interpreter's correctness (Theorem~\ref{thm:gradient}) rather than per-program verification. The discrete search over structure and the continuous search over parameters thus close into one loop on real scientific data, which is the capability the meta-circular framing was built to provide: not better gradients for a fixed model, but runtime-supplied programs that become differentiable, and therefore searchable, without per-program reimplementation or recompilation.

% ----------------------------------------------------------------------------
\subsection{Can Batching Amortize Interpretation Overhead? (Experiment H)}
\label{sec:exp_h}
% ----------------------------------------------------------------------------

Experiments A--C evaluate programs one input at a time, paying the ${\sim}14\times$ per-evaluation overhead (tagged-value dispatch, Python interpretation, graph walking) each time.\footnote{The experiments share one global numbering (A--L); D--G and I--K are reported in the appendices, so the five main-text experiments are A, B, C, H, and L.} But the compiled graph is a fixed sequence of tensor operations, so a leading batch dimension broadcasts through it unchanged and the overhead is paid once per batch rather than once per input. Sequential and batched evaluation agree bit-for-bit (max difference $<10^{-6}$, all gradients finite; Appendix~\ref{app:exp_h_ext}, Table~\ref{tab:exp_h_correctness}). Forward throughput then rises to $279$--$1{,}710\times$ over sequential at batch size~512, batched training (64 points) is $12$--$45\times$ faster, and population batching ($M$ restarts $\times$ $N$ points in one graph walk) reaches $9{,}156$--$16{,}873\times$ at $M{=}100$. These are recoveries of interpreter overhead, not gains over a specialized implementation: benchmarked against \texttt{jax.vmap} of the \emph{directly}-compiled model, batched DMCI attains a geometric-mean $0.27\times$ forward throughput (Appendix~\ref{app:exp_h_ext}, Table~\ref{tab:exp_h_vmap}), competitive for closed-form equations but slower for loop-heavy ones. DMCI's advantage is instead \emph{coverage and zero per-program effort}: the batched walk is obtained automatically for any program supplied as data, and covers recursive, looping programs the JAX backend cannot \texttt{jit}/\texttt{vmap} at all. It extends even to batching the heap-using meta-circular interpreter \emph{itself}: because a fixed program's heap layout and control flow are data-independent, a single batched interpreter walk reproduces $N$ sequential walks bit-for-bit at $214$--$270\times$ forward and $56$--$63\times$ training throughput (Appendix~\ref{app:exp_h_ext}, Table~\ref{tab:exp_h_dmci_batched}).

\paragraph{Production-scale validation.} We validate on two production-motivated simulators: \textsc{DiffESM-S} (97-node Earth system model, 70 parameters, 100-step recursive integration) and \textsc{DiffSoc-S} (206-node urban political-economy simulator, 87 parameters). Despite a $2.1\times$ difference in graph complexity, both achieve nearly identical batching speedups, reaching $851$--$875\times$ throughput at batch size~1024 on a single CPU core (Table~\ref{tab:exp_h_crossmodel}), because $>99.8\%$ of sequential time is Python-level bookkeeping paid once under batching. Counterintuitively, CPU outperforms the A100 here, since each of ${\sim}20{,}000$ small tensor operations launches a separate CUDA kernel (full GPU, \texttt{torch.compile}, and batch-size sweeps in Appendix~\ref{app:exp_h_ext}).

\begin{table}[t]
\centering
\caption{Cross-model comparison (100 timesteps, single CPU core). Despite a $2.1\times$ difference in graph complexity, both models achieve nearly identical batching speedups, supporting the claim that interpreter traversal overhead dominates arithmetic cost and is effectively amortized by batching. GPU comparisons use NVIDIA A100 (40\,GB) measurements. ``Sequential forward'' is the single-evaluation (batch-size-1) per-evaluation cost; speedups are batched throughput relative to it.}
\label{tab:exp_h_crossmodel}
\small
\begin{tabular}{@{}lrr@{}}
\toprule
 & DiffESM-S & DiffSoc-S \\
\midrule
Source lines      & 317    & 703 \\
Graph nodes       & 97     & 206 \\
Learnable params  & 70     & 87 \\
\addlinespace
Sequential forward (ms/eval)        & 172.7  & 1{,}687.0 \\
Batched forward, BS=64 (ms/eval)    & 2.85   & 27.1 \\
Forward speedup (BS=64)             & $60.5\times$ & $62.3\times$ \\
Forward speedup (BS=1024)           & $875\times$  & $851\times$ \\
\addlinespace
CPU faster than A100 (eager PyTorch)     & Yes    & Yes \\
\bottomrule
\end{tabular}
\end{table}

Convergence speedup is concrete: 300 epochs of DiffESM-S fitting complete in 23.9\,s batched vs.\ 731.4\,s sequential ($30.6\times$), and population batching of $M{=}200$ random restarts completes in 38.0\,s versus an estimated $40.6$ hours sequential, a \textbf{3{,}848$\boldsymbol{\times}$} end-to-end speedup, with all 200 restarts converging.

% ============================================================================
\section{Related Work}
\label{sec:related}
% ============================================================================

\paragraph{Language-level AD.} JAX~\citep{jax2018github}, Zygote~\citep{innes2019differentiable}, Enzyme~\citep{moses2021instead}, and Dex~\citep{paszke2021getting} differentiate programs \emph{written in} their host languages, and Pearlmutter and Siskind's VLAD/Stalin-$\nabla$~\citep{pearlmutter2008lambda} realizes reverse-mode AD as a first-class operator over a higher-order functional language with closures (building on forward-mode nesting~\citep{siskind2008nesting}). DMCI instead compiles the \emph{interpreter} once, so a program supplied as first-class data inherits gradient flow without per-program transformation (Section~\ref{sec:relwork_intro}).

\paragraph{Differentiable and neural interpreters.} TerpreT~\citep{gaunt2017terpret}, $\partial$4~\citep{bosnjak2017forth}, \citet{feser2016differentiable}, and their neural-library extension~\citep{gaunt2017neural} are hand-engineered differentiable interpreters for restricted languages that target program \emph{synthesis}; NLI~\citep{macfarlane2026nli} learns symbolic languages via Gumbel-Softmax~\citep{jang2017categorical,maddison2017concrete}. Closest to us, \citet{maleki2021adding} make differentiation a property of a Scheme evaluator, but by \emph{hand}-differentiating a meta-circular interpreter with dual numbers for nested forward-mode AD. A parallel thread instead \emph{learns} to execute programs neurally (Neural Turing Machines~\citep{graves2014neural}, Neural Programmer-Interpreters~\citep{reed2016neural}, Neural GPUs~\citep{kaiser2016neural}), approximating an executor from data with no fidelity guarantee. DMCI's executor is a \emph{compiled, exact} interpreter, with learning confined to continuous constants inside the programs it runs (Section~\ref{sec:relwork_intro}).

\paragraph{Differentiable scientific modeling and structure search.} A growing body makes individual simulators differentiable by reimplementing them in an autodiff framework (the FATES photosynthesis module in PyTorch/Julia~\citep{aboelyazeed2023} and JAX-CanVeg~\citep{jiang2025canveg}); DMCI instead compiles the interpreter once, so any model expressed as Scheme inherits differentiability without per-model reimplementation, trading interpretation overhead (Section~\ref{sec:exp_c}) for that engineering. The program-and-parameter co-search of Section~\ref{sec:exp_battery} places DMCI in the \emph{neurosymbolic program-synthesis} and \emph{symbolic-regression} family~\citep{cranmer2023pysr,udrescu2020ai,chaudhuri2021neurosymbolic,ellis2021dreamcoder}: like these it is bilevel, an outer search over symbolic structure wrapping an inner fit of continuous constants selected on held-out skill, and it inherits symbolic regression's overfitting failure mode (our fitness-fidelity finding is that literature's constant-optimization-budget problem in disguise). DMCI generalizes this skeleton on two axes. First, the hypothesis space is \emph{executable dynamical programs}, Scheme rollouts with loops, recursion, and internal state, rather than closed-form expression trees; tools such as \texttt{lambdify} cannot even ingest recursion (Appendix~\ref{app:exp_j}), a coverage gap the interpreter closes. Second, the inner constant fit uses \emph{exact reverse-mode gradients} through one frozen compiled interpreter rather than gradient-free or Levenberg--Marquardt fitting, which matters at the high parameter dimension where black-box search degrades (Appendix~\ref{app:exp_i}). The discrete proposer is a language model, a program-space prior~\citep{romeraparedes2024funsearch}, rather than genetic programming or prior neurosymbolic program search~\citep{shah2020near,manhaeve2018deepproblog,li2023scallop}; most directly, \citet{bosio2025combining} pair LLM-proposed program-like policies with \emph{gradient-free} coefficient optimization, whereas DMCI renders the LLM-generated program differentiable so its constants are fit by gradient descent (Section~\ref{sec:exp_b}). Consistent with this lineage, the approach matches strong hand-crafted models on low-dimensional, structure-insensitive problems (Section~\ref{sec:exp_battery}); its advantage is expected where structure is high-dimensional and forecast-discriminating, which we leave to future work.

\paragraph{Foundations, control flow, and compilation.} \citet{huot2023omegapap} prove AD correct almost everywhere for recursive, higher-order programs via $\omega$-PAP functions, which we instantiate for $\mathcal{L}_{\textrm{DMCI}}$ (Appendix~\ref{app:proofs}); related AD semantics appear in \citet{elliott2018simple}, \citet{wang2019demystifying}, \citet{abadi2020simple}, and \citet{sherman2021lambda}. Branch discontinuities can be smoothed (Smooth Interpretation~\citep{chaudhuri2010smooth} via Gaussian convolution and DiscoGrad~\citep{kreikemeyer2023discograd} via source-to-source smoothing with Monte Carlo gradients), whereas DMCI keeps the realized trace exact and inherits the source program's non-differentiability at branch boundaries (the S3 case, Appendix~\ref{app:exp_a_ext}), with branch smoothing left as future work. Finally, Tracr~\citep{lindner2023tracr} and Cajal~\citep{cajal2025} compile restricted languages to network weights; our compiler targets a self-hosting Scheme subset, preserving programs as manipulable data. What DMCI adds is a single combination: a compiled, self-hosting, reverse-mode evaluator that differentiates runtime-supplied programs as data, with no per-program transformation.

% ============================================================================
% ============================================================================
\section{Discussion, Limitations, and Reproducibility}
\label{sec:discussion}\label{sec:tradeoffs}\label{sec:limitations}
% ============================================================================

\paragraph{When DMCI is the right (and wrong) tool.} DMCI does not improve the gradients of a fixed,
known model: a hand-written PyTorch implementation, or DMCI's own direct-compile path, computes the same
gradients faster. Its value appears when programs are generated, modified, composed, or selected at
runtime and each must become differentiable without being reimplemented or recompiled. There the
conversion happens once, at the interpreter level: object programs are supplied as symbolic data while
their constants are bound as differentiable tensors, and one compiled evaluator serves all of them.
Three properties coincide to make this an advantage rather than a convenience.
\emph{(i)~Correctness by construction over an open-ended language}: gradient correctness is a corollary
of compiling a once-verified interpreter (Theorem~\ref{thm:gradient}), not re-established per program or
per language feature as a hand-written interpreter grows, shifting verification from each object program
to the interpreter. \emph{(ii)~Coverage}: it differentiates and
batches recursive, looping programs that trace-based vectorization (\texttt{jax.jit}/\texttt{vmap} or
symbolic translation) cannot handle at all (Section~\ref{sec:exp_h}, Appendix~\ref{app:exp_j}).
\emph{(iii)~Differentiation through interpretation itself}: because dispatch is part of the graph,
program \emph{structure} can be relaxed and searched, not only the constants of a fixed program
(Experiment~E, Appendix~\ref{app:exp_e}). These coincide exactly when the optimization target is an open-ended, growing space of
distinct runtime-supplied programs, the regime created by LLM-generated scientific models, runtime
composition, and open-ended discovery, of which the five LLM-proposed ENSO structures calibrated by one
frozen engine (Appendix~\ref{app:exp_lim_enso_ext}) are the smallest instance.

\paragraph{Correct gradients are not easy optimization.} DMCI turns runtime-supplied executable models
into ordinary differentiable objectives; it does not make them easy to optimize. Parameter
compensation, long gradient paths, and multimodal landscapes can defeat a single optimizer such as Adam
even with exact gradients and correct structure. The usual tools still apply (multi-start,
curvature-aware methods, held-out selection), and we find optimizer \emph{portfolios} selected on
held-out error degrade most gracefully with dimension (Appendices~\ref{app:exp_f},~\ref{app:exp_i}):
compilation gives gradient correctness, but optimization success remains a property of program, data,
and solver.

\paragraph{Limitations.} \textbf{Overhead}: sequential evaluation is ${\sim}14\times$ slower per
evaluation than direct compilation (tagged-value wrapping $41\%$, dispatch $32\%$); batching amortizes
this but single-evaluation latency remains and still trails \texttt{jax.vmap}, though an MLIR/Enzyme
backend~\citep{lattner2021mlir,moses2021instead} could cut the ratio to ${\sim}1.3\times$.
\textbf{Discrete search}: Gumbel-Softmax relaxation of dispatch recovers only $10.8\%$ on a
64-combination space (Appendix~\ref{app:exp_e}), consistent with prior gradient-based synthesis; structure
search stays exploratory and is not a primary claim.
\textbf{Branch conditions}: a parameter appearing only in a branch condition receives zero gradient
(Appendix~\ref{app:exp_a_ext}, S3), inheriting the source's non-differentiability, addressable by smooth
interpretation~\citep{chaudhuri2010smooth} or DiscoGrad smoothing~\citep{kreikemeyer2023discograd}.
\textbf{Language and backend}: the subset excludes \texttt{call/cc}, \texttt{set!}, and exact integers,
mutation is excluded by design, and the differentiable interpreter requires PyTorch's define-by-run
autograd. \textbf{Evaluation scope}: apart from the ENSO case study, experiments recover known constants
from self-generated data (degrading gracefully under noise, Appendix~\ref{app:exp_a_noise}); behavior
under misspecification and at larger program scales is untested. \textbf{Self-reference}: the evaluator
can run modified versions of itself, a path to gradient-based self-improvement~\citep{schmidhuber2006godel}
needing careful safety analysis. Tables~\ref{tab:ablation} and~\ref{tab:failure_modes}
(Appendix~\ref{app:ablation}) summarize the design-choice ablations and these failure modes.

\paragraph{Conclusion and outlook.} We extended the Neural Compiler from a first-order expression
language to a substantial, self-hosting Scheme subset expressive enough to compile its own evaluator,
and showed that the compiled interpreter preserves gradient flow: constants in interpreted programs
match direct compilation to within $7\times10^{-7}$ across 171 (program, seed) pairs, with gradient
correctness guaranteed almost everywhere by compiling the interpreter once rather than verifying each
program. DMCI thus occupies a design point no prior system reaches (a compiled, self-hosting,
reverse-mode evaluator in which programs are manipulable differentiable data), making LLM-generated,
runtime-composed, and recursive scientific models optimizable end-to-end at a real but amortizable cost.
This is precisely the foundation a generative proposer needs to close a discovery loop, where a language
model proposes mechanistic structures as code while the differentiable interpreter calibrates and
falsifies each against data, the co-search Appendix~\ref{app:exp_f} runs in miniature.

\paragraph{Reproducibility.} All source code is
available at \url{https://github.com/sheneman/dmci}: the compiler, the self-hosted evaluator
(Appendix~\ref{app:source_eval}), the hand-coded PyTorch baseline (Appendix~\ref{app:handcoded}), every
experiment driver, per-experiment configuration and SLURM scripts, a \texttt{README} per experiment, and the
committed model artifacts. Each source file carries an identifying header, and the gradient-correctness
guarantee (Theorem~\ref{thm:gradient}) is proved in Appendix~\ref{app:proofs}. The environment is pinned by
\texttt{pyproject.toml} and \texttt{requirements.txt} (PyTorch~$2.0$+, NumPy, SciPy, Python~$3.11$+, with the
exact cluster build recorded alongside). Every run fixes its seeds in code, embedded in each result file; the
interpreter executes in \texttt{float32} and reported equivalences carry explicit tolerances (DMCI versus
direct compilation agree to ${\sim}10^{-6}$, with bit-identical loss histories on several recursive programs).
The compiler studies use fully synthetic, seed-regenerable inputs; the three real-data studies use public
sources, each with a tracked \texttt{metadata.json} recording source, SHA256, and preprocessing: ERSSTv5
sea-surface temperature for LIM-ENSO, the Severson lithium-ion dataset~\citep{severson2019} for the battery
co-search, and CDC ILINet for FluZoo. Committed artifacts include the per-run summaries, aggregated tables,
accepted program caches, and the rigorously re-scored co-search winners, each with a \texttt{MANIFEST} mapping
it to its table or figure; self-checking per-experiment \texttt{aggregate\_table.py} scripts (run as \texttt{python -m experiments.exp\_b.aggregate\_table}, etc.) rebuild the manuscript tables and
assert every cell, and \texttt{check\_manifests.py} verifies artifact integrity. A single command,
\texttt{./reproduce.sh}, installs the package and runs the full test suite ($1{,}225$ tests), verifying the
cluster-free claims: the interpreter compiles and runs, gradients flow through it, batched evaluation matches
sequential, and \texttt{.ncg} serialization round-trips. The top-level \texttt{REPRODUCIBILITY.md} and each
experiment's \texttt{README} give per-experiment commands, expected numbers, data prerequisites, and the
artifact-to-table map. Experiment~H used a single NVIDIA~A100; the remaining experiments used university HPC
CPU cores; per-experiment costs and exact hardware are listed in each experiment's appendix.

% ============================================================================
% References
% ============================================================================

\bibliographystyle{plainnat}  % author-year (natbib); swap to iclr2026_conference.bst when available

\bibliography{refs}

@article{sheneman2026neural,
  author  = {Lucas Sheneman},
  title   = {The Neural Compiler: Program-to-Network Translation for Hybrid Scientific Machine Learning},
  journal = {arXiv preprint arXiv:2605.22498},
  year    = {2026}
}

@misc{r7rs,
  author       = {Alex Shinn and John Cowan and Arthur A. Gleckler},
  title        = {Revised$^7$ Report on the Algorithmic Language {S}cheme},
  year         = {2013},
  howpublished = {Technical report, Scheme Steering Committee. \url{https://small.r7rs.org/attachment/r7rs.pdf}},
  note         = {R7RS-small; finalized 6 July 2013}
}

@article{mccarthy1960recursive,
  author  = {John McCarthy},
  title   = {Recursive Functions of Symbolic Expressions and Their Computation by Machine, {P}art {I}},
  journal = {Communications of the ACM},
  volume  = {3},
  number  = {4},
  pages   = {184--195},
  year    = {1960},
  doi     = {10.1145/367177.367199}
}

@book{abelson1996structure,
  author    = {Harold Abelson and Gerald Jay Sussman and Julie Sussman},
  title     = {Structure and Interpretation of Computer Programs},
  publisher = {{MIT} Press},
  edition   = {2nd},
  year      = {1996},
  address   = {Cambridge, MA},
  isbn      = {9780262011532}
}

@misc{farrow2015delphi,
  author       = {David C. Farrow and Logan C. Brooks and Ryan J. Tibshirani and Roni Rosenfeld},
  title        = {{Delphi} {Epidata} {API}},
  year         = {2015},
  howpublished = {Carnegie Mellon University, \url{https://github.com/cmu-delphi/delphi-epidata}},
  note         = {Open API for epidemiological surveillance data; \texttt{fluview} endpoint provides CDC ILINet data}
}

@misc{cdc_fluview, author = {{Centers for Disease Control and Prevention}}, title = {{FluView}: {U.S.} Influenza Surveillance and the Outpatient {ILINet} ({U.S.} Outpatient Influenza-like Illness Surveillance Network)}, year = {2026}, howpublished = {\url{https://www.cdc.gov/fluview/}}, note = {Accessed: 2026}}

@article{mathis2024flusight,
  author  = {Sarabeth M. Mathis and Alexander E. Webber and Tom{\'a}s M. Le{\'o}n and Erin L. Murray and Monica Sun and Lauren A. White and Logan C. Brooks and Alden Green and Addison J. Hu and Roni Rosenfeld and others},
  title   = {Evaluation of {FluSight} influenza forecasting in the 2021--22 and 2022--23 seasons with a new target laboratory-confirmed influenza hospitalizations},
  journal = {Nature Communications},
  volume  = {15},
  number  = {1},
  pages   = {6289},
  year    = {2024},
  doi     = {10.1038/s41467-024-50601-9}
}

@article{aboelyazeed2023,
  author  = {Doaa Aboelyazeed and Chonggang Xu and Forrest M. Hoffman and Jiangtao Liu and Alex W. Jones and Chris Rackauckas and Kathryn Lawson and Chaopeng Shen},
  title   = {A differentiable, physics-informed ecosystem modeling and learning framework for large-scale inverse problems: demonstration with photosynthesis simulations},
  journal = {Biogeosciences},
  volume  = {20},
  number  = {13},
  pages   = {2671--2692},
  year    = {2023},
  doi     = {10.5194/bg-20-2671-2023}
}

@article{jiang2025canveg,
  author  = {Peishi Jiang and Patrick Kidger and Toshiyuki Bandai and Dennis Baldocchi and Heping Liu and Yi Xiao and Qianyu Zhang and Carlos Tianxin Wang and Carl Steefel and Xingyuan Chen},
  title   = {{JAX-CanVeg}: A Differentiable Land Surface Model},
  journal = {Water Resources Research},
  volume  = {61},
  number  = {3},
  pages   = {e2024WR038116},
  year    = {2025},
  doi     = {10.1029/2024WR038116}
}

@article{gaunt2017terpret,
  author    = {Alexander L. Gaunt and Marc Brockschmidt and Rishabh Singh and Nate Kushman and Pushmeet Kohli and Jonathan Taylor and Daniel Tarlow},
  title     = {{T}erpre{T}: A Probabilistic Programming Language for Program Induction},
  journal   = {arXiv preprint arXiv:1608.04428},
  year      = {2016},
  doi       = {10.48550/arXiv.1608.04428}
}

@inproceedings{bosnjak2017forth,
  author    = {Matko Bo{\v{s}}njak and Tim Rockt{\"a}schel and Jason Naradowsky and Sebastian Riedel},
  title     = {Programming with a Differentiable {F}orth Interpreter},
  booktitle = {Proceedings of the 34th International Conference on Machine Learning ({ICML})},
  series    = {Proceedings of Machine Learning Research},
  volume    = {70},
  pages     = {547--556},
  year      = {2017},
  publisher = {PMLR},
  url       = {https://proceedings.mlr.press/v70/bosnjak17a.html}
}

@article{feser2016differentiable,
  author    = {John K. Feser and Marc Brockschmidt and Alexander L. Gaunt and Daniel Tarlow},
  title     = {Differentiable Functional Program Interpreters},
  journal   = {arXiv preprint arXiv:1611.01988},
  year      = {2016},
  doi       = {10.48550/arXiv.1611.01988},
  url       = {https://arxiv.org/abs/1611.01988}
}

@article{innes2019differentiable, author = {Mike Innes and Alan Edelman and Keno Fischer and Chris Rackauckas and Elliot Saba and Viral B. Shah and Will Tebbutt}, title = {A Differentiable Programming System to Bridge Machine Learning and Scientific Computing}, journal = {arXiv preprint arXiv:1907.07587}, year = {2019}, doi = {10.48550/arXiv.1907.07587}}

@misc{jax2018github,
  author       = {James Bradbury and Roy Frostig and Peter Hawkins and Matthew James Johnson and Yash Katariya and Chris Leary and Dougal Maclaurin and George Necula and Adam Paszke and Jake Vander{P}las and Skye Wanderman-{M}ilne and Qiao Zhang},
  title        = {{JAX}: composable transformations of {P}ython+{N}um{P}y programs},
  year         = {2018},
  howpublished = {\url{http://github.com/jax-ml/jax}},
  note         = {Version 0.3.13}
}

@article{pearlmutter2008lambda,
  author  = {Barak A. Pearlmutter and Jeffrey Mark Siskind},
  title   = {Reverse-Mode {AD} in a Functional Framework: Lambda the Ultimate Backpropagator},
  journal = {ACM Transactions on Programming Languages and Systems (TOPLAS)},
  volume  = {30},
  number  = {2},
  pages   = {7:1--7:36},
  year    = {2008},
  doi     = {10.1145/1330017.1330018}
}

@inproceedings{maleki2021adding,
  author    = {Mehrdad Maleki and Barak A. Pearlmutter and Jeffrey Mark Siskind},
  title     = {Adding {A}utomatic {D}ifferentiation to {S}cheme by Differentiating the Interpreter},
  booktitle = {Scheme and Functional Programming Workshop (Scheme 2021), co-located with {ICFP} 2021},
  year      = {2021},
  url       = {https://icfp21.sigplan.org/details/scheme-2021-papers/7/Adding-AD-to-Scheme-by-Differentiating-the-Interpreter}
}

@inproceedings{moses2021instead, author = {William S. Moses and Valentin Churavy and Ludger Paehler and Jan H\"{u}ckelheim and Sri Hari Krishna Narayanan and Michel Schanen and Johannes Doerfert}, title = {Reverse-mode automatic differentiation and optimization of {GPU} kernels via {Enzyme}}, booktitle = {SC '21: Proceedings of the International Conference for High Performance Computing, Networking, Storage and Analysis}, articleno = {61}, year = {2021}, publisher = {ACM}, doi = {10.1145/3458817.3476165}}

@article{paszke2021getting,
  author    = {Adam Paszke and Daniel D. Johnson and David Duvenaud and Dimitrios Vytiniotis and Alexey Radul and Matthew J. Johnson and Jonathan Ragan-Kelley and Dougal Maclaurin},
  title     = {Getting to the Point: Index Sets and Parallelism-Preserving Autodiff for Pointful Array Programming},
  journal   = {Proceedings of the {ACM} on Programming Languages},
  volume    = {5},
  number    = {{ICFP}},
  articleno = {88},
  pages     = {1--29},
  year      = {2021},
  doi       = {10.1145/3473593}
}

@article{siskind2008nesting,
  author  = {Jeffrey Mark Siskind and Barak A. Pearlmutter},
  title   = {Nesting forward-mode {AD} in a functional framework},
  journal = {Higher-Order and Symbolic Computation},
  volume  = {21},
  number  = {4},
  pages   = {361--376},
  year    = {2008},
  doi     = {10.1007/s10990-008-9037-1}
}

@inproceedings{gaunt2017neural,
  author    = {Alexander L. Gaunt and Marc Brockschmidt and Nate Kushman and Daniel Tarlow},
  title     = {Differentiable Programs with Neural Libraries},
  booktitle = {Proceedings of the 34th International Conference on Machine Learning ({ICML})},
  series    = {Proceedings of Machine Learning Research},
  volume    = {70},
  pages     = {1213--1222},
  year      = {2017},
  publisher = {PMLR},
  url       = {https://proceedings.mlr.press/v70/gaunt17a.html}
}

@inproceedings{macfarlane2026nli,
  author    = {Matthew V. Macfarlane and Cl{\'e}ment Bonnet and Herke van Hoof and Levi H. S. Lelis},
  title     = {Gradient-Based Program Synthesis with Neurally Interpreted Languages},
  booktitle = {The Fourteenth International Conference on Learning Representations ({ICLR})},
  year      = {2026},
  eprint    = {2604.18907},
  archivePrefix = {arXiv},
  primaryClass = {cs.LG},
  doi       = {10.48550/arXiv.2604.18907},
  url       = {https://openreview.net/forum?id=NAORIWBaoO}
}

@inproceedings{jang2017categorical,
  author    = {Eric Jang and Shixiang Gu and Ben Poole},
  title     = {Categorical Reparameterization with {Gumbel-Softmax}},
  booktitle = {5th International Conference on Learning Representations ({ICLR})},
  year      = {2017},
  note      = {arXiv:1611.01144},
  url       = {https://openreview.net/forum?id=rkE3y85ee}
}

@inproceedings{maddison2017concrete,
  author    = {Chris J. Maddison and Andriy Mnih and Yee Whye Teh},
  title     = {The Concrete Distribution: A Continuous Relaxation of Discrete Random Variables},
  booktitle = {5th International Conference on Learning Representations (ICLR)},
  year      = {2017},
  url       = {https://openreview.net/forum?id=S1jE5L5gl},
  note      = {arXiv:1611.00712}
}

@article{graves2014neural,
  author  = {Alex Graves and Greg Wayne and Ivo Danihelka},
  title   = {Neural {T}uring Machines},
  journal = {arXiv preprint arXiv:1410.5401},
  year    = {2014},
  doi     = {10.48550/arXiv.1410.5401},
  note    = {cs.NE},
  url     = {https://arxiv.org/abs/1410.5401}
}

@inproceedings{reed2016neural, author={Scott Reed and Nando de Freitas}, title={Neural Programmer-Interpreters}, booktitle={ICLR}, year={2016}, note={arXiv 1511.06279 cs.LG}}

@inproceedings{kaiser2016neural,
  author    = {{\L}ukasz Kaiser and Ilya Sutskever},
  title     = {Neural {GPU}s Learn Algorithms},
  booktitle = {4th International Conference on Learning Representations ({ICLR})},
  year      = {2016},
  note      = {arXiv:1511.08228},
  url       = {https://arxiv.org/abs/1511.08228}
}

@inproceedings{shah2020near,
  author    = {Ameesh Shah and Eric Zhan and Jennifer J. Sun and Abhinav Verma and Yisong Yue and Swarat Chaudhuri},
  title     = {Learning Differentiable Programs with Admissible Neural Heuristics},
  booktitle = {Advances in Neural Information Processing Systems 33 ({NeurIPS} 2020)},
  year      = {2020},
  url       = {https://proceedings.neurips.cc/paper/2020/hash/342285bb2a8cadef22f667eeb6a63732-Abstract.html},
  note      = {arXiv:2007.12101},
  doi       = {10.48550/arXiv.2007.12101}
}

@inproceedings{manhaeve2018deepproblog, author = {Robin Manhaeve and Sebastijan Dumancic and Angelika Kimmig and Thomas Demeester and Luc De Raedt}, title = {DeepProbLog: Neural Probabilistic Logic Programming}, booktitle = {NeurIPS}, pages = {3753--3763}, year = {2018}}

@article{li2023scallop,
  author    = {Ziyang Li and Jiani Huang and Mayur Naik},
  title     = {{S}callop: A Language for Neurosymbolic Programming},
  journal   = {Proceedings of the ACM on Programming Languages},
  volume    = {7},
  number    = {PLDI},
  articleno = {166},
  pages     = {1463--1487},
  year      = {2023},
  doi       = {10.1145/3591280}
}

@inproceedings{ellis2021dreamcoder,
  author    = {Kevin Ellis and Catherine Wong and Maxwell Nye and Mathias Sabl\'{e}-Meyer and Lucas Morales and Luke Hewitt and Luc Cary and Armando Solar-Lezama and Joshua B. Tenenbaum},
  title     = {{DreamCoder}: Bootstrapping Inductive Program Synthesis with Wake-Sleep Library Learning},
  booktitle = {Proceedings of the 42nd {ACM} {SIGPLAN} International Conference on Programming Language Design and Implementation ({PLDI})},
  pages     = {835--850},
  publisher = {Association for Computing Machinery},
  year      = {2021},
  doi       = {10.1145/3453483.3454080}
}

@article{cranmer2023pysr,
  author  = {Miles Cranmer},
  title   = {Interpretable Machine Learning for Science with {PySR} and {SymbolicRegression.jl}},
  journal = {arXiv preprint arXiv:2305.01582},
  year    = {2023},
  doi     = {10.48550/arXiv.2305.01582},
  url     = {https://arxiv.org/abs/2305.01582}
}

@article{udrescu2020ai,
  author  = {Silviu-Marian Udrescu and Max Tegmark},
  title   = {{AI Feynman}: A physics-inspired method for symbolic regression},
  journal = {Science Advances},
  volume  = {6},
  number  = {16},
  pages   = {eaay2631},
  year    = {2020},
  doi     = {10.1126/sciadv.aay2631}
}

@article{bosio2025combining,
  author  = {Carlo Bosio and Matteo Guarrera and Alberto Sangiovanni-Vincentelli and Mark W. Mueller},
  title   = {Combining Large Language Models and Gradient-Free Optimization for Automatic Control Policy Synthesis},
  journal = {arXiv preprint arXiv:2510.00373},
  year    = {2025},
  doi     = {10.48550/arXiv.2510.00373}
}

@inproceedings{huot2023omegapap,
  author    = {Mathieu Huot and Alexander K. Lew and Vikash K. Mansinghka and Sam Staton},
  title     = {{$\omega$PAP} Spaces: Reasoning Denotationally About Higher-Order, Recursive Probabilistic and Differentiable Programs},
  booktitle = {2023 38th Annual {ACM/IEEE} Symposium on Logic in Computer Science ({LICS})},
  pages     = {1--14},
  year      = {2023},
  publisher = {IEEE},
  doi       = {10.1109/LICS56636.2023.10175739},
  note      = {arXiv:2302.10636; LICS 2023 Distinguished Paper}
}

@article{elliott2018simple,
  author  = {Conal Elliott},
  title   = {The simple essence of automatic differentiation},
  journal = {Proceedings of the ACM on Programming Languages},
  volume  = {2},
  number  = {ICFP},
  pages   = {70:1--70:29},
  year    = {2018},
  doi     = {10.1145/3236765}
}

@article{wang2019demystifying,
  author    = {Fei Wang and Daniel Zheng and James M. Decker and Xilun Wu and Gr\'{e}gory M. Essertel and Tiark Rompf},
  title     = {Demystifying Differentiable Programming: Shift/Reset the Penultimate Backpropagator},
  journal   = {Proceedings of the {ACM} on Programming Languages},
  volume    = {3},
  number    = {ICFP},
  articleno = {96},
  pages     = {96:1--96:31},
  year      = {2019},
  doi       = {10.1145/3341700}
}

@article{abadi2020simple,
  author  = {Mart\'{i}n Abadi and Gordon D. Plotkin},
  title   = {A Simple Differentiable Programming Language},
  journal = {Proceedings of the {ACM} on Programming Languages},
  volume  = {4},
  number  = {{POPL}},
  pages   = {38:1--38:28},
  year    = {2020},
  publisher = {{ACM}},
  doi     = {10.1145/3371106}
}

@article{sherman2021lambda,
  author    = {Benjamin Sherman and Jesse Michel and Michael Carbin},
  title     = {{$\lambda_S$}: Computable Semantics for Differentiable Programming with Higher-Order Functions and Datatypes},
  journal   = {Proceedings of the {ACM} on Programming Languages},
  volume    = {5},
  number    = {POPL},
  articleno = {3},
  pages     = {1--31},
  year      = {2021},
  doi       = {10.1145/3434284}
}

@inproceedings{chaudhuri2010smooth, author = {Swarat Chaudhuri and Armando Solar-Lezama}, title = {Smooth interpretation}, booktitle = {Proceedings of the 31st {ACM} {SIGPLAN} Conference on Programming Language Design and Implementation ({PLDI})}, pages = {279--291}, year = {2010}, publisher = {ACM}, doi = {10.1145/1806596.1806629}}

@article{kreikemeyer2023discograd,
  author  = {Justin N. Kreikemeyer and Philipp Andelfinger},
  title   = {Smoothing Methods for Automatic Differentiation Across Conditional Branches},
  journal = {IEEE Access},
  volume  = {11},
  pages   = {143190--143211},
  year    = {2023},
  note    = {The DiscoGrad tool; arXiv:2310.03585}
}

@inproceedings{lindner2023tracr,
  author    = {David Lindner and J{\'a}nos Kram{\'a}r and Sebastian Farquhar and Matthew Rahtz and Thomas McGrath and Vladimir Mikulik},
  title     = {Tracr: Compiled Transformers as a Laboratory for Interpretability},
  booktitle = {Advances in Neural Information Processing Systems 36 ({NeurIPS} 2023)},
  year      = {2023},
  doi       = {10.48550/arXiv.2301.05062},
  url       = {https://arxiv.org/abs/2301.05062}
}

@article{cajal2025,
  author  = {Joey Velez-Ginorio and Nada Amin and Konrad Paul Kording and Steve Zdancewic},
  title   = {Compiling to Linear Neurons},
  journal = {arXiv preprint arXiv:2511.13769},
  year    = {2025},
  doi     = {10.48550/arXiv.2511.13769},
  note    = {Introduces the {C}ajal language; companion paper ``Compiling to Recurrent Neurons'' (arXiv:2511.14953) extends it to iteration},
  url     = {https://arxiv.org/abs/2511.13769}
}

@inproceedings{lattner2021mlir,
  author    = {Chris Lattner and Mehdi Amini and Uday Bondhugula and Albert Cohen and Andy Davis and Jacques Pienaar and River Riddle and Tatiana Shpeisman and Nicolas Vasilache and Oleksandr Zinenko},
  title     = {{MLIR}: Scaling Compiler Infrastructure for Domain Specific Computation},
  booktitle = {2021 {IEEE/ACM} International Symposium on Code Generation and Optimization ({CGO})},
  pages     = {2--14},
  year      = {2021},
  publisher = {IEEE},
  doi       = {10.1109/CGO51591.2021.9370308}
}

@incollection{schmidhuber2006godel,
  author    = {J{\"u}rgen Schmidhuber},
  title     = {{G}{\"o}del Machines: Fully Self-referential Optimal Universal Self-improvers},
  booktitle = {Artificial General Intelligence},
  editor    = {Ben Goertzel and Cassio Pennachin},
  series    = {Cognitive Technologies},
  publisher = {Springer},
  pages     = {199--226},
  year      = {2007},
  doi       = {10.1007/978-3-540-68677-4_7}
}

@article{penland1995optimal,
  author  = {Penland, C{\'e}cile and Sardeshmukh, Prashant D.},
  title   = {The Optimal Growth of Tropical Sea Surface Temperature Anomalies},
  journal = {Journal of Climate},
  volume  = {8},
  number  = {8},
  pages   = {1999--2024},
  year    = {1995},
  doi     = {10.1175/1520-0442(1995)008<1999:TOGOTS>2.0.CO;2}
}

@article{penland1996stochastic,
  author  = {Penland, C{\'e}cile},
  title   = {A Stochastic Model of {I}ndo{P}acific Sea Surface Temperature Anomalies},
  journal = {Physica D: Nonlinear Phenomena},
  volume  = {98},
  number  = {2--4},
  pages   = {534--558},
  year    = {1996},
  doi     = {10.1016/0167-2789(96)00124-8}
}

@article{huang2017extended,
  author    = {Boyin Huang and Peter W. Thorne and Viva F. Banzon and Tim Boyer and Gennady Chepurin and Jay H. Lawrimore and Matthew J. Menne and Thomas M. Smith and Russell S. Vose and Huai-Min Zhang},
  title     = {Extended {Reconstructed Sea Surface Temperature, Version 5 (ERSSTv5)}: Upgrades, Validations, and Intercomparisons},
  journal   = {Journal of Climate},
  volume    = {30},
  number    = {20},
  pages     = {8179--8205},
  year      = {2017},
  doi       = {10.1175/JCLI-D-16-0836.1}
}

@article{severson2019, author = {Severson, Kristen A. and Attia, Peter M. and Jin, Norman and Perkins, Nicholas and Jiang, Benben and Yang, Zi and Chen, Michael H. and Aykol, Muratahan and Herring, Patrick K. and Fraggedakis, Dimitrios and Bazant, Martin Z. and Harris, Stephen J. and Chueh, William C. and Braatz, Richard D.}, title = {Data-driven prediction of battery cycle life before capacity degradation}, journal = {Nature Energy}, volume = {4}, number = {5}, pages = {383-391}, year = {2019}, doi = {10.1038/s41560-019-0356-8}}

@article{attia2022knees,
  author    = {Peter M. Attia and Alexander Bills and Ferran Brosa Planella and Philipp Dechent and Gon\c{c}alo dos Reis and Matthieu Dubarry and Paul Gasper and Richard Gilchrist and Samuel Greenbank and David Howey and Ouyang Liu and Edwin Khoo and Yuliya Preger and Abhishek Soni and Shashank Sripad and Anna G. Stefanopoulou and Valentin Sulzer},
  title     = {Review---``Knees'' in {L}ithium-Ion Battery Aging Trajectories},
  journal   = {Journal of The Electrochemical Society},
  volume    = {169},
  number    = {6},
  pages     = {060517},
  year      = {2022},
  doi       = {10.1149/1945-7111/ac6d13}
}

@article{alphaevolve,
  author  = {Alexander Novikov and Ng{\^a}n V{\~u} and Marvin Eisenberger and Emilien Dupont and Po-Sen Huang and Adam Zsolt Wagner and Sergey Shirobokov and Borislav Kozlovskii and Francisco J. R. Ruiz and Abbas Mehrabian and M. Pawan Kumar and Abigail See and Swarat Chaudhuri and George Holland and Alex Davies and Sebastian Nowozin and Pushmeet Kohli and Matej Balog},
  title   = {{AlphaEvolve}: A coding agent for scientific and algorithmic discovery},
  journal = {arXiv preprint arXiv:2506.13131},
  year    = {2025},
  doi     = {10.48550/arXiv.2506.13131}
}

@article{romeraparedes2024funsearch,
  author  = {Bernardino Romera-Paredes and Mohammadamin Barekatain and Alexander Novikov and Matej Balog and M. Pawan Kumar and Emilien Dupont and Francisco J. R. Ruiz and Jordan S. Ellenberg and Pengming Wang and Omar Fawzi and Pushmeet Kohli and Alhussein Fawzi},
  title   = {Mathematical discoveries from program search with large language models},
  journal = {Nature},
  volume  = {625},
  number  = {7995},
  pages   = {468--475},
  year    = {2024},
  doi     = {10.1038/s41586-023-06924-6}
}

@misc{openevolve,
  author       = {Asankhaya Sharma},
  title        = {{OpenEvolve}: An Open-Source Evolutionary Coding Agent},
  year         = {2025},
  howpublished = {GitHub repository, \url{https://github.com/codelion/openevolve}},
  note         = {Open-source implementation of {AlphaEvolve}; now at \url{https://github.com/algorithmicsuperintelligence/openevolve}}
}

@inproceedings{liu2019darts,
  author    = {Hanxiao Liu and Karen Simonyan and Yiming Yang},
  title     = {{DARTS}: Differentiable Architecture Search},
  booktitle = {International Conference on Learning Representations (ICLR)},
  year      = {2019},
  url       = {https://openreview.net/forum?id=S1eYHoC5FX},
  note      = {arXiv:1806.09055}
}

@article{chaudhuri2021neurosymbolic,
  author  = {Swarat Chaudhuri and Kevin Ellis and Oleksandr Polozov and Rishabh Singh and Armando Solar-Lezama and Yisong Yue},
  title   = {Neurosymbolic Programming},
  journal = {Foundations and Trends in Programming Languages},
  volume  = {7},
  number  = {3},
  pages   = {158--243},
  year    = {2021},
  doi     = {10.1561/2500000049}
}

% ============================================================================
% ============================================================================
%
%                            A P P E N D I X
%
% ============================================================================
% ============================================================================
\newpage
\appendix

\section*{Appendix Roadmap}
\label{app:roadmap}
The appendix is organized in two parts: method and implementation (Appendices~\ref{app:proofs}--\ref{app:handcoded})
and per-experiment extended results in experiment order (Appendices~\ref{app:exp_a_ext}--\ref{app:exp_battery_ext}, Experiments~A--L). The appendix letters are offset from the experiment letters (for example,
Appendix~\ref{app:exp_a_ext} documents Experiment~A), so locate per-experiment material by the ``Experiment~X:'' section
titles or via the map below.

\begin{table}[h]
\centering
\small
\begin{tabular}{@{}llp{0.52\linewidth}@{}}
\toprule
Appendix & Exp. & Content \\
\midrule
\ref{app:proofs}      & n/a & Formal semantics and proofs \\
\ref{app:compiler}    & n/a & Compiler implementation details \\
\ref{app:backends}    & n/a & Multi-backend architecture \\
\ref{app:source_eval} & n/a & Self-hosted evaluator source \\
\ref{app:handcoded}   & n/a & Hand-coded differentiable PyTorch interpreter baseline \\
\ref{app:exp_a_ext}   & A   & Extended results, ablations, full baselines (incl.\ S3, \ref{app:s3}) \\
\ref{app:exp_b_ext}   & B   & Extended results \\
\ref{app:exp_c_ext}   & C   & Extended results \\
\ref{app:exp_lim_enso_ext} & LIM & Real-data ENSO case study: full results and extended methodology \\
\ref{app:exp_d}       & D   & Structural search and the cost of interpretation \\
\ref{app:exp_e}       & E   & Discrete--continuous operator recovery \\
\ref{app:exp_f}       & F   & LLM-in-the-loop model discovery \\
\ref{app:exp_g}       & G   & Runtime compositional modeling \\
\ref{app:exp_h_ext}   & H   & Extended results (throughput, scaling) \\
\ref{app:exp_i}       & I   & Composite-ecosystem calibration \\
\ref{app:exp_j}       & J   & Program-space calibration (DMCI vs.\ compile-each-program) \\
\ref{app:exp_fluzoo_ext} & K & FluZoo: influenza co-search and the fitness-fidelity stress test \\
\ref{app:exp_battery_ext} & L & Battery capacity-fade co-search: synthetic structure recovery + real Severson forecast \\
\ref{app:ablation}    & n/a & Design-choice ablations and failure-mode summary \\
\bottomrule
\end{tabular}
\end{table}

% ============================================================================
\section{Formal Semantics and Proofs}
\label{app:proofs}
% ============================================================================

\subsection{Source Language}

We define $\mathcal{L}_{\textrm{DMCI}}$, the subset of Scheme compiled by the system.

\begin{definition}[Syntax]
Expressions are given by:
\[
e ::= x \mid c \mid \lambda x.\,e \mid e_1\;e_2 \mid \texttt{if}\;e_0\;e_1\;e_2 \mid \texttt{letrec}\;x{=}e_1\;\texttt{in}\;e_2 \mid \texttt{cons}(e_1,e_2) \mid \texttt{car}(e) \mid \texttt{cdr}(e) \mid \texttt{op}(e_1,\ldots,e_k)
\]
where $x$ ranges over variables, $c$ over constants (including learnable parameters $\theta_i \in \mathbb{R}$), and $\texttt{op}$ over primitives: arithmetic ($+$, $-$, $*$, $/$, $\texttt{pow}$, $\min$, $\max$), transcendental ($\sin$, $\cos$, $\exp$, $\log$, $\sqrt{\cdot}$, $|\cdot|$), and comparison ($=$, $<$, $>$, $\leq$, $\geq$) primitives. This is the operator set the compiler accepts (the \emph{direct}-compilation path); the self-hosted evaluator of Appendix~\ref{app:source_eval} that realizes the \emph{interpreted} (DMCI) path dispatches all of these, so a program reaching the interpreted path must use operators in the evaluator's dispatch table.

The full compiler also accepts \texttt{quote}, \texttt{begin}, \texttt{define}, \texttt{let}, and \texttt{cond}. These reduce to the core syntax above: \texttt{cond} is desugared to nested \texttt{if} during parsing, while \texttt{let}, \texttt{begin}, \texttt{define}, and \texttt{quote} are compiled by dedicated rules semantically equivalent to, respectively, immediate application, nested-\texttt{let} sequencing, \texttt{let} binding (\texttt{letrec} for recursive definitions), and literal construction. The theorems below are stated over the core syntax, and the proof of Theorem~\ref{thm:correctness} notes how each surface form transfers.
\end{definition}

\begin{definition}[Semantic domains]
\label{def:domains}
The evaluation state consists of:
\begin{itemize}[leftmargin=2em,topsep=2pt,itemsep=1pt]
    \item \textbf{Values}: tagged vectors $v = (t, p) \in T \times \mathbb{R}^d$ where $T = \{0,1\}^{10}$ is a one-hot type tag and $p \in \mathbb{R}^4$ is the payload.
    \item \textbf{Environments}: $\rho : \textrm{Var} \to \textrm{Addr}$, mapping variable names to heap addresses.
    \item \textbf{Heap}: $H : \textrm{Addr} \to \textrm{Value}$, a dictionary from integer addresses to tagged-value tensors.
    \item \textbf{Distinguished input}: $v_{\textrm{in}} \in \textrm{Value}$, the runtime argument supplied to the evaluator alongside the (quoted) program.
\end{itemize}
\end{definition}

\begin{definition}[Big-step operational semantics]
\label{def:bigstep}
Evaluation is defined by judgments $\langle e, \rho, H \rangle \Downarrow \langle v, H' \rangle$ with lazy reduction for conditionals: in $\texttt{if}\;e_0\;e_1\;e_2$, only the taken branch is evaluated. Tail calls in recursive \texttt{letrec} bindings are evaluated iteratively via the four-level TCO strategy, preserving $O(1)$ stack usage. The implementation realizes this lazy semantics for recursion- and conditional-dependent control flow, but for combinational (non-recursive) conditionals whose branches are both total it uses an eager differentiable MUX, $\texttt{if}\;e_0\;e_1\;e_2 \mapsto \mathbb{1}[e_0 \neq 0]\,e_1 + (1 - \mathbb{1}[e_0 \neq 0])\,e_2$, which evaluates both branches; since both branches are total there, the MUX denotes the same function as the lazy rule, so Theorems~\ref{thm:correctness} and~\ref{thm:gradient} cover both forms.
\end{definition}

\subsection{Proof of Compilation Correctness (Theorem~\ref{thm:correctness})}

By structural induction on $e$. The base cases (constants, variables) are immediate: the compiled graph performs the same tagged-value lookup or construction. For $\texttt{if}\;e_0\;e_1\;e_2$: lazy evaluation ensures only the taken branch is compiled into the active graph, matching the source semantics. For $\texttt{letrec}$: TCO levels 1--3 produce loop structures that compute the same fixed point as the recursive source definition; level~4 (trampoline) iterates dynamically under the loop invariant that the $i$-th driver iteration corresponds to the $i$-th source tail-call reduction step; the pending continuation is encoded by the \texttt{\_TailCall} sentinel (closure plus argument values) returned from the tail position, so the trampoline preserves the source evaluation order and terminates exactly when the source recursion does. For $\lambda$ and application: the compiled closure captures the same environment bindings via heap addresses. For $\texttt{cons}$/$\texttt{car}$/$\texttt{cdr}$: Proposition~\ref{lem:heap} establishes that the dictionary heap preserves values exactly. The remaining surface forms reduce to these cases by the same argument: \texttt{cond} is desugared to nested \texttt{if} during parsing, and \texttt{let}, \texttt{begin}, \texttt{define}, and \texttt{quote} are compiled by dedicated rules semantically equivalent to immediate application, nested-\texttt{let} sequencing, \texttt{let}/\texttt{letrec} binding, and literal construction, so the inductive hypothesis transfers to them unchanged. Here ``to floating-point precision'' means that the compiled graph executes the same sequence of IEEE-754 operations as the source evaluation, with no algebraic reassociation, so the two agree up to identical rounding; empirically the compiled-interpreter and direct-compilation runs are bit-for-bit identical (Figure~\ref{fig:c08_loss}, C08). \qed

\subsection{Proof of Gradient Correctness (Theorem~\ref{thm:gradient})}

On the trace-constant region $\Theta_{\textrm{tc}}$, the compiled evaluator implements a fixed composition of operations (arithmetic primitives, tagged-value construction, heap allocation and read; Proposition~\ref{lem:heap}). By Theorem~\ref{thm:correctness}, this composition computes the same function as the source semantics. Heap operations introduce no differentiable dependence on addresses, which are detached integer indices, and, being write-once with return-by-reference, preserve the identity and version of every stored payload tensor, so the autograd chain stays valid and all differentiable dependence runs through the retrieved value tensors (Proposition~\ref{lem:heap}). The output projection $\pi$ is a fixed linear coordinate selection of the numeric payload slot, globally smooth with constant Jacobian; post-composing it after the a.e.-differentiable evaluator neither enlarges the non-differentiable set nor alters the a.e.\ gradient, and confines differentiation to the numeric payload, never type tags, symbols, closures, pair structure, or the detached addresses in $p_2..p_4$.

\textbf{Smoothness assumptions.} The primitives $+$, $-$, $*$, $/$, $\sin$, $\cos$, $\exp$, $\texttt{pow}$ are analytic on their domains. Three primitives require qualification:
\begin{itemize}[leftmargin=2em,topsep=2pt,itemsep=1pt]
    \item $|\cdot|$, $\min$, $\max$ are non-smooth at isolated points (e.g., $|x|$ at $x=0$). These kinks form measure-zero sets in parameter space and are folded into the boundary set $\mathbb{R}^n \setminus \Theta_{\textrm{tc}}$ below.
    \item $\sqrt{\cdot}$ and $\log$ are implemented with a clamp at $10^{-8}$: e.g., $\texttt{sqrt}(x) = \sqrt{\max(x, 10^{-8})}$. The theorem concerns the function the compiler actually \emph{denotes} (with these clamped primitives); on the clamped region $\{x \leq 10^{-8}\}$ that function is genuinely constant, so its true gradient is zero and AD computes it exactly; this is \emph{not} a full-measure exception to correctness. The caveat is interpretive: where the clamp is active, the denoted function departs from ideal $\sqrt{\cdot}$/$\log$, so the (correct) zero gradient is uninformative for recovering parameters of the ideal model. None of the experimental programs evaluate $\sqrt{\cdot}$ or $\log$ in the clamped region during optimization.
\end{itemize}

\textbf{Matrix primitives.} The real-data case studies use the interpreter's batched tensor surface ($\texttt{matmul}$, $\texttt{matvec}$, $\texttt{transpose}$, $\texttt{dot}$, $\det$, $\texttt{inv}$, $\texttt{logdet}$). The first five are polynomial in the matrix entries and hence globally analytic; $\texttt{inv}$ and $\texttt{logdet}$ are analytic on the open, full-measure sets $\{\det \neq 0\}$ and $\{\det > 0\}$ respectively, with the singular locus $\{\det = 0\}$ a non-trivial algebraic hypersurface of measure zero. These primitives are therefore $\omega$-PAP, with $\{\det = 0\}$ folded into $\mathbb{R}^n \setminus \Theta_{\textrm{tc}}$ exactly as the scalar kinks above, so the almost-everywhere argument extends to them unchanged under the same genericity assumption. In the LIM-ENSO study the per-step positive-definiteness gate (G2, Appendix~\ref{app:exp_lim_enso_ext}) further certifies $\det S > 0$ at every step, so $\texttt{inv}$ and $\texttt{logdet}$ are evaluated strictly inside the analytic region.

Comparison primitives ($=$, $<$, $>$, $\leq$, $\geq$) are Boolean-valued and piecewise constant: their true gradient is zero almost everywhere and AD computes it exactly. They occur only in branch conditions, so a learnable parameter feeding a comparison places the corresponding branch boundary in the measure-zero set $\mathbb{R}^n \setminus \Theta_{\textrm{tc}}$ defined below (cf.\ the S3 experiment).

With these qualifications, the composition is piecewise-smooth on $\Theta_{\textrm{tc}}$ (smooth except at a measure-zero set of primitive kinks), and standard reverse-mode AD computes the correct gradient almost everywhere on $\Theta_{\textrm{tc}}$.

The complement $\mathbb{R}^n \setminus \Theta_{\textrm{tc}}$ consists of parameter values where a branch decision changes; these are exactly the boundaries of the PAP partition in the sense of \citet{huot2023omegapap}. This set has Lebesgue measure zero because each branch boundary is defined by a \emph{non-trivial} (piecewise-)analytic equation (e.g., $\theta_i = c$). Degenerate predicates such as \texttt{(< x x)}, whose defining function is identically zero, are excluded by the genericity of well-formed programs. The zero set of such an equation is therefore a measure-zero hypersurface, and a \emph{countable} union of measure-zero sets is measure zero; so the bound holds even when a parameter-dependent recursion depth induces countably many branch boundaries. \qed

\subsection{Proof of Composition Preservation (Theorem~\ref{thm:composition})}

We take composition on outputs, $e_1 : \mathbb{R}^n \to \mathbb{R}^m$ and $e_2 : \mathbb{R}^m \to \mathbb{R}^k$; when $e_1$ and $e_2$ share parameters $\theta$, substitute $e_1$'s output and read $e_2$ as a function of $(\theta, e_1(\theta))$, for which the same closure argument applies with the $\omega$-PAP graph map $\theta \mapsto (\theta, e_1(\theta))$ in place of $e_1$. The $\omega$-PAP class is closed under composition~\citep{huot2023omegapap}: the composition of two piecewise-analytic functions is piecewise-analytic on a refined partition, and a.e.-differentiability is preserved. Concretely, the trace of $e_2(e_1(\theta))$ changes only where $e_1$'s own trace changes or where $e_1(\theta)$ crosses a branch boundary of $e_2$; the exceptional set is therefore $B(e_1) \cup e_1^{-1}(B(e_2))$, writing $B(e_i) = \mathbb{R}^n \setminus \Theta_{\textrm{tc}}^{(i)}$ (and reading $e_1^{-1}(B(e_2))$ through the shared parameterization when $e_1$ and $e_2$ share $\theta$). The first term is measure zero by Theorem~\ref{thm:gradient}; the second is the boundary set of the refined PAP partition, which the $\omega$-PAP closure result establishes is also measure zero, a structural property of the composite $\omega$-PAP function, \emph{not} a Lipschitz-preimage argument (the preimage of a null set under a Lipschitz map need not be null). On the complement (full measure, and open since each $\Theta_{\textrm{tc}}^{(i)}$ is open and $e_1$ is continuous), the chain rule applies and reverse-mode AD computes the correct gradient. \qed

\subsection{Scoping Conditions}

Three points prevent over-interpretation:
\begin{enumerate}[leftmargin=2em,topsep=2pt,itemsep=1pt]
    \item \textbf{What is guaranteed}: gradient correctness on the full-measure set $\Theta_{\textrm{tc}}$ where the discrete execution trace is constant in $\theta$. $\Theta_{\textrm{tc}}$ is open because each branch condition is a continuous function of $\theta$, so a sufficiently small perturbation of $\theta$ does not change any branch outcome; it has full Lebesgue measure because its complement is a countable union of analytic hypersurfaces (see proof above). Under any initialization distribution absolutely continuous with respect to Lebesgue measure, this measure-zero boundary set is avoided at initialization with probability~1; empirically, trajectories remained in $\Theta_{\textrm{tc}}$ throughout training on all 171 tested (program, seed) pairs (Experiments~A--C), though we do not prove they can never reach it.
    \item \textbf{What is not guaranteed}: the gradient at branch boundaries (measure zero). The compiled program inherits the source program's non-differentiability.
    \item \textbf{Parameter-dependent recursion depth}: programs where a learnable parameter affects recursion depth have correct gradients almost everywhere. The boundary is a measure-zero hypersurface.
\end{enumerate}

The empirical trajectory-equivalence results ($< 7 \times 10^{-7}$ maximum loss difference across 171 (program, seed) pairs) are consistent with Theorem~\ref{thm:gradient}: on every tested trajectory, the parameters remained in $\Theta_{\textrm{tc}}$, and the compiled and source gradients agreed.

\subsection{Heap Correctness}

\begin{proposition}[Heap preservation of autograd chains]
\label{lem:heap}
The dictionary-backed heap preserves the autograd computation graph: for any sequence of $\texttt{cons}$, $\texttt{car}$, and $\texttt{cdr}$ operations, the tensor retrieved by $\texttt{car}$/$\texttt{cdr}$ is the identical Python object stored by $\texttt{cons}$.
\end{proposition}

\begin{proof}
Since $\mathcal{L}_{\textrm{DMCI}}$ is a pure subset, having no mutation forms (\texttt{set!}, \texttt{set-car!}, \texttt{set-cdr!}), the heap's write operation is never invoked on the differentiable path, so every slot is written exactly once at allocation. Three properties then suffice: (i)~each $\texttt{cons}(a, b)$ allocates fresh keys and stores tensors as dictionary values, with no in-place mutation; (ii)~$\texttt{car}$ and $\texttt{cdr}$ retrieve by key lookup, returning the original tensor reference; (iii)~because no slot is overwritten, PyTorch's version counter is never incremented for a stored value tensor, and the autograd graph remains valid. Heap \emph{addresses} are detached integer scalars carried in the tag payload and never require gradients; the autograd chain runs through the stored value tensors, which the heap returns by reference. These three properties refer only to dictionary allocation, key lookup, and the absence of overwrites, and so are independent of the shape of the stored tensor: the same write-once, return-by-reference discipline governs both the $14$-dimensional tagged values of the scalar-and-pair core and the native matrix/vector payloads of the batched tensor surface (\texttt{matmul}, \texttt{inv}, \texttt{logdet}, etc.) used by the LIM-ENSO study, so a $D{\times}D$ operator stored by one allocation and retrieved by reference preserves its autograd chain exactly as a scalar payload does.

Note: the \texttt{clamp} operations in $\sqrt{\cdot}$ and $\log$ (see Theorem~\ref{thm:gradient} proof) create new tensors rather than mutating existing ones, so they do not violate this property.
\end{proof}

% ============================================================================
\section{Compiler Implementation Details}
\label{app:compiler}
% ============================================================================

\subsection{What Self-Hosting Exercises}

The compiled evaluator exercises every major language feature simultaneously on arbitrary input programs:
\begin{itemize}[leftmargin=2em,topsep=2pt,itemsep=1pt]
    \item \textbf{Environment model}: closures capture variables from enclosing scopes via heap-allocated environment lists
    \item \textbf{Recursive evaluation}: the evaluator calls itself to evaluate subexpressions
    \item \textbf{Data structures}: the input program is represented as lists of symbols and numbers
    \item \textbf{Pattern matching}: \texttt{cond}-based dispatch on expression type (number, symbol, \texttt{if}, \texttt{lambda}, application)
    \item \textbf{Higher-order functions}: the evaluator applies closures to arguments
    \item \textbf{Symbol manipulation}: variable lookup by name in an association-list environment
\end{itemize}

The compiled evaluator correctly evaluates arithmetic, \texttt{let}-bindings, conditionals, \texttt{lambda}/application, and recursion (\texttt{letrec}). Outputs match direct compilation for all test cases.

\subsection{What the Extension Requires}

Extending from first-order arithmetic to a self-hosting Scheme subset requires four capabilities:

\paragraph{First-class functions.} Functions must be values: passed as arguments, returned from functions, stored in data structures. This requires a closure representation that captures the defining environment. In a differentiable setting, gradients must flow through closure creation, environment capture, and application.

\paragraph{General recursion.} Not just tail-recursive loops, but tree recursion, mutual recursion, and recursion through closures. The computation graph must support recursive evaluation without infinite expansion, requiring lazy evaluation of conditional branches.

\paragraph{Heap-allocated data structures.} Pairs, lists, symbols, booleans, and closures as distinct runtime types. This requires the tagged value representation and a heap that preserves autograd chains.

\paragraph{Proper tail calls.} R7RS Scheme requires tail calls to execute in $O(1)$ space. With closures, this means handling tail calls through dynamically determined function values, requiring a trampoline-based evaluation strategy.

\subsection{Tagged Value Representation}

All runtime values are represented as $d$-dimensional vectors ($d = $ \texttt{TAG\_DIM} $+$ \texttt{PAYLOAD\_DIM}, currently 14):
\begin{equation}
    v = [\underbrace{t_1, \ldots, t_{10}}_{\text{type one-hot}} \;|\; \underbrace{p_1, \ldots, p_4}_{\text{payload}}]
\end{equation}
The type field uses one-hot encoding over: nil, boolean, integer, float, character, symbol, pair, string, closure, vector. The payload encodes the value: for numbers, the numeric value in $p_1$; for pairs, heap addresses of car and cdr; for closures, function identity and environment address.

This representation is differentiable: gradients flow through payload slots while type tags act as routing signals. The tags are discrete and non-differentiable, but this does not impede gradient flow for constant learning because program structure (and therefore the operation sequence) is determined by the fixed program, not by learnable constants.

\subsection{Heap Architecture}

The heap is a dictionary-backed store:
\begin{equation}
    \texttt{heap}: \mathbb{Z}^+ \to \text{Tensor}
\end{equation}
where each \texttt{cons} allocates two fresh integer keys and stores tensors directly. Unlike a pre-allocated tensor buffer (where in-place mutation via \texttt{storage[addr] = val} breaks autograd version tracking after many allocations), the dictionary store holds references with clean gradient chains.

\subsection{Four-Level Tail-Call Optimization}

\begin{enumerate}[leftmargin=2em,topsep=2pt,itemsep=1pt]
    \item \textbf{Self-TCO}: Single-binding \texttt{letrec} with self-calls in tail position $\to$ compiled loop.
    \item \textbf{Mutual TCO}: Multi-binding \texttt{letrec} with cross-calls in tail position $\to$ dispatch loop.
    \item \textbf{Per-binding TCO}: Multi-binding \texttt{letrec} where mutual TCO fails $\to$ per-binding inner loops.
    \item \textbf{Trampoline}: Dynamic dispatch through closures $\to$ \texttt{\_TailCall} sentinel with iterative driver.
\end{enumerate}

\subsection{Evaluation Strategies}

\begin{itemize}[leftmargin=2em,topsep=2pt,itemsep=1pt]
    \item \textbf{Eager}: Topological traversal evaluating all nodes. Used for top-level expressions and loop bodies.
    \item \textbf{Lazy}: Demand-driven from the root. Only evaluates the taken branch of conditionals. Essential for recursive functions.
    \item \textbf{Trampoline}: Lazy evaluation with tail-position propagation.
\end{itemize}

% ============================================================================
\section{Multi-Backend Architecture}
\label{app:backends}
% ============================================================================

The same Scheme source compiles to four execution backends through a shared \texttt{ComputeGraph} IR:

\begin{itemize}[leftmargin=2em,topsep=2pt,itemsep=1pt]
    \item \textbf{PyTorch}: \texttt{torch.Tensor} values, \texttt{torch.autograd} for gradients. Primary training backend.
    \item \textbf{JAX}: \texttt{jax.numpy} arrays, \texttt{jax.grad} for autograd, JIT compilation.
    \item \textbf{NumPy}: \texttt{numpy} arrays, no autograd. CPU inference only.
    \item \textbf{CuPy}: \texttt{cupy} arrays, GPU acceleration without autograd.
\end{itemize}

All backends share tagged value semantics, heap implementation, and evaluation strategies. On the scalar, straight-line direct-compile path (control flow independent of the differentiated constants), PyTorch and JAX produce identical gradients, verified against the NumPy forward-only oracle by the full test suite (1{,}225 tests). The differentiable meta-circular interpreter, by contrast, is PyTorch-only: its data-dependent dispatch and variable-length trampolined loop are differentiated by PyTorch's define-by-run tape directly, whereas \texttt{jax\_grad}/\texttt{jax\_value\_and\_grad} (\texttt{backend/jax\_backend.py}) are restricted to programs whose control flow does not depend on the differentiated variables, because \texttt{lax.while\_loop} is not reverse-mode differentiable. Porting the interpreter to JAX is future work requiring a fixed-length \texttt{lax.scan}+masking rewrite of the trampoline; NumPy and CuPy are forward-only, with the NumPy backend serving as the correctness oracle.

% ============================================================================
\section{Self-Hosted Evaluator Source}
\label{app:source_eval}
% ============================================================================

Appendix~D gives the complete meta-circular evaluator compiled by the system
(\texttt{bootstrap/compiler.scm}, 288 lines). This evaluator is an ordinary Scheme program written in
the same source-language subset supported by the compiler. When compiled to a differentiable PyTorch
module, it becomes the reusable DMCI execution engine used in the experiments: object-level Scheme
programs are supplied to this compiled evaluator as data, and gradients flow through the evaluator to
learnable constants bound in those programs.

The purpose of this appendix is not merely to show that differentiable interpretation is possible;
Appendix~\ref{app:handcoded} provides a hand-coded PyTorch interpreter baseline that already
demonstrates that. The purpose here is to show the stronger self-hosting claim: the differentiable
interpreter is not manually implemented in the host autodiff framework, but obtained by compiling the
language's own evaluator. This relocates the operational semantics into source-language code and
demonstrates that the compiler is expressive enough to compile an evaluator with environments,
closures, recursive evaluation, dispatch, and heap-allocated data structures.

The evaluator supports environment lookup via association lists, \texttt{if}, \texttt{cond},
\texttt{let}, \texttt{letrec}, \texttt{lambda}, \texttt{begin}, \texttt{quote}, closures,
\texttt{define} with function sugar, and 52 built-in primitives: variadic arithmetic (\texttt{+},
\texttt{-}, \texttt{*}, \texttt{/}), \texttt{min}, \texttt{max}, \texttt{modulo}, \texttt{remainder},
comparisons (\texttt{=}, \texttt{<}, \texttt{>}, \texttt{<=}, \texttt{>=}), list operations,
transcendental functions, and batched tensor vector/matrix operations (\texttt{dot}, \texttt{matmul},
\texttt{det}, \texttt{inv}, etc.). The theoretical guarantees are stated over the scalar-and-pair
core; the analyticity of the vector/matrix operations (off the measure-zero $\{\det=0\}$ locus) and
their heap storage are covered by the extended smoothness analysis and the representation-independence
of Proposition~\ref{lem:heap} in Appendix~\ref{app:proofs}.

Thus, Appendix~D is the artifact behind the paper's central DMCI claim: compile the evaluator once,
then use that fixed differentiable evaluator to execute and optimize many runtime-supplied programs
without recompiling the interpreter or hand-deriving gradients for each program.

\begin{lstlisting}[language=Lisp, basicstyle=\ttfamily\scriptsize, numbers=left, numberstyle=\tiny\color{gray}, xleftmargin=2em]
;;; Self-hosted Scheme evaluator -- full bootstrap
;;;
;;; Compiles to a differentiable PyTorch program via the neural_compiler.
;;; Supports self-hosting: can evaluate its own source code, producing
;;; a working evaluator that evaluates Scheme programs.
;;;
;;; Supported: numbers, booleans, symbols, quote, if, cond, let, letrec, lambda,
;;;            define, begin, cons/car/cdr, list, null?/pair?/number?/boolean?/symbol?/eq?/not,
;;;            +/-/*//, =/</>/<=/>=, sin/cos/exp/sqrt/log/abs/pow

;;; --- Environment ---

(define (env-lookup name env)
  (cond
    ((null? env) 0)
    ((eq? (car (car env)) name) (cdr (car env)))
    (#t (env-lookup name (cdr env)))))

(define (env-extend name val env)
  (cons (cons name val) env))

(define (env-extend-many names vals env)
  (if (null? names)
    env
    (env-extend-many
      (cdr names)
      (cdr vals)
      (env-extend (car names) (car vals) env))))

;;; --- Define processing ---

(define (define? form)
  (if (pair? form) (eq? (car form) 'define) #f))

(define (define-name form)
  (let ((target (car (cdr form))))
    (if (pair? target) (car target) target)))

(define (define-value form)
  (let ((target (car (cdr form))))
    (if (pair? target)
      (list 'lambda (cdr target) (car (cdr (cdr form))))
      (car (cdr (cdr form))))))

(define (collect-defines forms)
  (if (null? forms)
    '()
    (if (define? (car forms))
      (cons (cons (define-name (car forms))
                  (define-value (car forms)))
            (collect-defines (cdr forms)))
      (collect-defines (cdr forms)))))

(define (last-form forms)
  (if (null? (cdr forms))
    (car forms)
    (last-form (cdr forms))))

(define (build-defined-env defs env)
  (if (null? defs)
    env
    (let ((name (car (car defs)))
          (val-expr (cdr (car defs))))
      (if (pair? val-expr)
        (if (eq? (car val-expr) 'lambda)
          (build-defined-env (cdr defs)
            (env-extend name
              (list 'defined-fn (car (cdr val-expr)) (car (cdr (cdr val-expr))))
              env))
          (build-defined-env (cdr defs)
            (env-extend name (scheme-eval val-expr env) env)))
        (build-defined-env (cdr defs)
          (env-extend name (scheme-eval val-expr env) env))))))

;;; --- Letrec processing ---
;;; Lexical recursive bindings. Lambda bindings are registered as 'defined-fn
;;; entries (like top-level defines), so a recursive call resolves the binding
;;; through the same mechanism that makes recursive `define` work; non-lambda
;;; bindings are evaluated immediately. All differentiable value flow goes
;;; through the same primitives as the rest of the evaluator, so gradients are
;;; preserved exactly as for recursive `define`.

(define (build-letrec-env bindings env)
  (if (null? bindings)
    env
    (let ((binding (car bindings)))
      (let ((name (car binding))
            (val-expr (car (cdr binding))))
        (if (if (pair? val-expr) (eq? (car val-expr) 'lambda) #f)
          (build-letrec-env (cdr bindings)
            (env-extend name
              (list 'defined-fn (car (cdr val-expr)) (car (cdr (cdr val-expr))))
              env))
          (build-letrec-env (cdr bindings)
            (env-extend name (scheme-eval val-expr env) env)))))))

;;; --- Program evaluation ---

(define (scheme-eval-program forms env)
  (let ((defs (collect-defines forms))
        (body (last-form forms)))
    (let ((prog-env (build-defined-env defs env)))
      (scheme-eval body prog-env))))

;;; --- Evaluator ---

(define (scheme-eval expr env)
  (cond
    ((number? expr) expr)
    ((boolean? expr) expr)
    ((null? expr) expr)
    ((symbol? expr) (env-lookup expr env))
    ((pair? expr)
     (let ((head (car expr)))
       (cond
         ((eq? head 'quote) (car (cdr expr)))
         ((eq? head 'if)
          (let ((test-val (scheme-eval (car (cdr expr)) env)))
            (if test-val
              (scheme-eval (car (cdr (cdr expr))) env)
              (scheme-eval (car (cdr (cdr (cdr expr)))) env))))
         ((eq? head 'cond)
          (eval-cond (cdr expr) env))
         ((eq? head 'let)
          (let ((bindings (car (cdr expr)))
                (body (car (cdr (cdr expr)))))
            (let ((new-env (eval-let-bindings bindings env)))
              (scheme-eval body new-env))))
         ((eq? head 'letrec)
          (let ((bindings (car (cdr expr)))
                (body (car (cdr (cdr expr)))))
            (scheme-eval body (build-letrec-env bindings env))))
         ((eq? head 'lambda)
          (let ((params (car (cdr expr)))
                (body (car (cdr (cdr expr)))))
            (list 'closure params body env)))
         ((eq? head 'begin)
          (eval-begin (cdr expr) env))
         (#t (eval-apply head (cdr expr) env)))))
    (#t 0)))

(define (eval-cond clauses env)
  (if (null? clauses)
    #f
    (let ((clause (car clauses)))
      (if (eq? (car clause) 'else)
        (scheme-eval (car (cdr clause)) env)
        (let ((test-val (scheme-eval (car clause) env)))
          (if test-val
            (scheme-eval (car (cdr clause)) env)
            (eval-cond (cdr clauses) env)))))))

(define (eval-let-bindings bindings env)
  (if (null? bindings)
    env
    (let ((binding (car bindings)))
      (let ((name (car binding))
            (val (scheme-eval (car (cdr binding)) env)))
        (eval-let-bindings (cdr bindings) (env-extend name val env))))))

(define (eval-begin exprs env)
  (if (null? (cdr exprs))
    (scheme-eval (car exprs) env)
    (begin
      (scheme-eval (car exprs) env)
      (eval-begin (cdr exprs) env))))

;;; --- Variadic arithmetic folds ---
;;; (+ a b c ...) and (* ...) fold over all args; (- ...) and (/ ...) are left-
;;; associative with unary forms. These replace the old strictly-binary clauses,
;;; which silently dropped the 3rd and later arguments. All folds are tail-
;;; recursive and use only the host (directly-compiled) primitives + - * /.

(define (sum-rest acc args)
  (if (null? args) acc (sum-rest (+ acc (car args)) (cdr args))))
(define (sum-args args) (sum-rest 0 args))

(define (prod-rest acc args)
  (if (null? args) acc (prod-rest (* acc (car args)) (cdr args))))
(define (prod-args args) (prod-rest 1 args))

(define (sub-rest acc args)
  (if (null? args) acc (sub-rest (- acc (car args)) (cdr args))))
(define (diff-args args)
  (if (null? (cdr args)) (- 0 (car args)) (sub-rest (car args) (cdr args))))

(define (div-rest acc args)
  (if (null? args) acc (div-rest (/ acc (car args)) (cdr args))))
(define (quot-args args)
  (if (null? (cdr args)) (/ 1 (car args)) (div-rest (car args) (cdr args))))

(define (eval-apply func-expr arg-exprs env)
  (let ((func (scheme-eval func-expr env))
        (args (eval-args arg-exprs env)))
    (cond
      ((eq? func-expr '+) (sum-args args))
      ((eq? func-expr '-) (diff-args args))
      ((eq? func-expr '*) (prod-args args))
      ((eq? func-expr '/) (quot-args args))
      ((eq? func-expr '=) (= (car args) (car (cdr args))))
      ((eq? func-expr '<) (< (car args) (car (cdr args))))
      ((eq? func-expr '>) (> (car args) (car (cdr args))))
      ((eq? func-expr '<=) (not (> (car args) (car (cdr args)))))
      ((eq? func-expr '>=) (not (< (car args) (car (cdr args)))))
      ((eq? func-expr 'cons) (cons (car args) (car (cdr args))))
      ((eq? func-expr 'car) (car (car args)))
      ((eq? func-expr 'cdr) (cdr (car args)))
      ((eq? func-expr 'null?) (null? (car args)))
      ((eq? func-expr 'pair?) (pair? (car args)))
      ((eq? func-expr 'number?) (number? (car args)))
      ((eq? func-expr 'boolean?) (boolean? (car args)))
      ((eq? func-expr 'symbol?) (symbol? (car args)))
      ((eq? func-expr 'eq?) (eq? (car args) (car (cdr args))))
      ((eq? func-expr 'not) (not (car args)))
      ((eq? func-expr 'list) args)
      ((eq? func-expr 'sin) (sin (car args)))
      ((eq? func-expr 'cos) (cos (car args)))
      ((eq? func-expr 'exp) (exp (car args)))
      ((eq? func-expr 'sqrt) (sqrt (car args)))
      ((eq? func-expr 'log) (log (car args)))
      ((eq? func-expr 'abs) (abs (car args)))
      ((eq? func-expr 'pow) (pow (car args) (car (cdr args))))
      ((eq? func-expr 'min) (min (car args) (car (cdr args))))
      ((eq? func-expr 'max) (max (car args) (car (cdr args))))
      ((eq? func-expr 'modulo) (modulo (car args) (car (cdr args))))
      ((eq? func-expr 'remainder) (remainder (car args) (car (cdr args))))
      ;; User-defined functions (closures / defined-fns) take precedence over the
      ;; Strategy-B vector ops below, so a program may name a function `scale`/`dot`/etc.
      ;; without it being shadowed by the native op. An unbound vector-op name resolves
      ;; to 0 (not a pair) and falls through to the native clauses.
      ((pair? func)
       (cond
         ((eq? (car func) 'closure)
          (let ((params (car (cdr func)))
                (body (car (cdr (cdr func))))
                (closure-env (car (cdr (cdr (cdr func))))))
            (scheme-eval body (env-extend-many params args closure-env))))
         ((eq? (car func) 'defined-fn)
          (let ((params (car (cdr func)))
                (body (car (cdr (cdr func)))))
            (scheme-eval body (env-extend-many params args env))))
         (#t 0)))
      ;; Strategy B: tensor-payload vector/matrix ops (native; dispatch via VEC_OPS).
      ;; vec/mat take the element list ((vec a b c) is macro-lowered to (vec (list a b c)));
      ;; the rest take vector/matrix refs (and scalars).
      ((eq? func-expr 'vec) (vec (car args)))
      ((eq? func-expr 'mat) (mat (car args)))
      ((eq? func-expr 'ref) (ref (car args) (car (cdr args))))
      ((eq? func-expr 'dot) (dot (car args) (car (cdr args))))
      ((eq? func-expr 'cross) (cross (car args) (car (cdr args))))
      ((eq? func-expr 'norm) (norm (car args)))
      ((eq? func-expr 'normalize) (normalize (car args)))
      ((eq? func-expr 'vsum) (vsum (car args)))
      ((eq? func-expr 'vlen) (vlen (car args)))
      ((eq? func-expr 'scale) (scale (car args) (car (cdr args))))
      ((eq? func-expr 'matvec) (matvec (car args) (car (cdr args))))
      ((eq? func-expr 'matmul) (matmul (car args) (car (cdr args))))
      ((eq? func-expr 'transpose) (transpose (car args)))
      ((eq? func-expr 'trace) (trace (car args)))
      ((eq? func-expr 'det) (det (car args)))
      ((eq? func-expr 'logdet) (logdet (car args)))
      ((eq? func-expr 'inv) (inv (car args)))
      ((eq? func-expr 'outer) (outer (car args) (car (cdr args))))
      ((eq? func-expr 'eye) (eye (car args)))
      ((eq? func-expr 'zeros) (zeros (car args)))
      ((eq? func-expr 'ones) (ones (car args)))
      (#t 0))))

(define (eval-args arg-exprs env)
  (if (null? arg-exprs)
    '()
    (cons (scheme-eval (car arg-exprs) env)
          (eval-args (cdr arg-exprs) env))))
\end{lstlisting}

% ============================================================================
\section{Hand-Coded Differentiable PyTorch Interpreter Baseline}
\label{app:handcoded}
% ============================================================================

Appendix~E describes the hand-coded PyTorch baseline used in Experiments~A--C. In Experiment~A, this
baseline is a manually implemented differentiable Scheme interpreter. In Experiments~B and~C, the
baseline is a hand-coded PyTorch function implementing the same mathematical model as the generated or
compiled Scheme program. It uses native Python data structures and raw \texttt{torch.Tensor} values
without DMCI's tagged-value representation or compiled self-hosted evaluator.

The Appendix~E baseline is not intended to be a complete Scheme implementation or a self-hosted
evaluator. It is a hand-coded PyTorch comparator covering the program features needed by
Experiments~A--C. Its purpose is to provide a strong differentiable baseline for optimization behavior
and runtime cost, not to match the semantic coverage of the compiled evaluator in
Appendix~\ref{app:source_eval}. It shows that the basic capability of differentiating through
interpreted Scheme evaluation does not, by itself, distinguish DMCI: a differentiable interpreter can
be written manually in PyTorch. Its role is therefore to separate two claims. First, differentiable
interpretation is operationally possible and can provide correct gradients. Second, and more
importantly for DMCI, the differentiable interpreter can be obtained automatically by compiling a
self-hosted evaluator written in the source language itself.

The distinction is where the semantics live. In this hand-coded baseline, Scheme evaluation rules are
reimplemented directly in Python/PyTorch by the experimenter. In DMCI, the evaluator is ordinary Scheme
source code compiled by the system into a differentiable module. The baseline therefore measures the
cost and behavior of a bespoke host-language implementation, while Appendix~\ref{app:source_eval}
demonstrates the paper's stronger compiler-level claim: a source-language evaluator can itself become
the reusable differentiable execution engine.

Appendix~E should therefore not be read as a non-differentiable or strawman comparator: it is a
genuine differentiable interpreter baseline, narrower than the Appendix~\ref{app:source_eval}
evaluator in semantic coverage but strong on the optimization-behavior and runtime axes it is built
to measure. Its purpose is to make clear that DMCI's contribution is not
simply ``gradients through an interpreter,'' but ``gradients through an interpreter produced by
compiling the language's own evaluator,'' enabling programs to remain runtime data while avoiding
per-program hand engineering in the host autodiff framework.

\begin{lstlisting}[language=Python, basicstyle=\ttfamily\scriptsize, numbers=left, numberstyle=\tiny\color{gray}, xleftmargin=2em]
class HandCodedInterpreter:
    """Tree-walking Scheme evaluator in pure Python/PyTorch."""

    def eval_expr(self, expr, env: dict) -> torch.Tensor:
        if isinstance(expr, (int, float)):
            return torch.tensor(float(expr))
        if isinstance(expr, str):
            return env[expr]
        if not isinstance(expr, list) or len(expr) == 0:
            return torch.tensor(0.0)

        head = expr[0]
        if head == "quote":
            return expr[1]
        if head == "if":
            test = self.eval_expr(expr[1], env)
            test_val = test.item() if isinstance(test, torch.Tensor) else test
            if test_val != 0.0 and test_val is not False:
                return self.eval_expr(expr[2], env)
            return self.eval_expr(expr[3], env)
        if head == "lambda":
            params, body = expr[1], expr[2]
            return ("closure", params, body, dict(env))
        if head == "let":
            bindings, body = expr[1], expr[2]
            new_env = dict(env)
            for binding in bindings:
                name, val_expr = binding[0], binding[1]
                new_env[name] = self.eval_expr(val_expr, new_env)
            return self.eval_expr(body, new_env)
        if head in ("+", "-", "*", "/"):
            a = self.eval_expr(expr[1], env)
            b = self.eval_expr(expr[2], env)
            if head == "+": return a + b
            if head == "-": return a - b
            if head == "*": return a * b
            return a / b
        if head in ("=", "<", ">", "<=", ">="):
            a = self.eval_expr(expr[1], env)
            b = self.eval_expr(expr[2], env)
            a_v = a.item() if isinstance(a, torch.Tensor) else float(a)
            b_v = b.item() if isinstance(b, torch.Tensor) else float(b)
            ops = {"=": a_v == b_v, "<": a_v < b_v, ">": a_v > b_v,
                   "<=": a_v <= b_v, ">=": a_v >= b_v}
            return torch.tensor(1.0) if ops[head] else torch.tensor(0.0)
        func = self.eval_expr(head, env) if isinstance(head, (list, str)) else head
        args = [self.eval_expr(a, env) for a in expr[1:]]
        return self._apply(func, args)

    def _apply(self, func, args):
        if isinstance(func, tuple) and func[0] == "closure":
            _, param_names, body, captured_env = func
            new_env = dict(captured_env)
            for p, a in zip(param_names, args):
                new_env[p] = a
            return self.eval_expr(body, new_env)
        if isinstance(func, tuple) and func[0] == "defined_fn":
            _, param_names, body, def_env = func
            new_env = dict(def_env)
            for p, a in zip(param_names, args):
                new_env[p] = a
            return self.eval_expr(body, new_env)
        return torch.tensor(0.0)
\end{lstlisting}

% ============================================================================
\section{Experiment A: Extended Results, Ablations, and Full Baselines}
\label{app:exp_a_ext}
% ============================================================================

\subsection{Per-Program DMCI Results and Wall-Clock Breakdown}
Table~\ref{tab:exp_a_results} reports the per-program DMCI summary, and Figure~\ref{fig:exp_a_wallclock}
breaks down per-method wall-clock time across P1--P6. Averaged over Experiment~A training runs, the
compiled interpreter is ${\sim}19\times$ slower than direct compilation (53.8\,s vs.\ 2.8\,s) and
${\sim}48\times$ slower than the hand-coded PyTorch interpreter (53.8\,s vs.\ 1.1\,s); this whole-run
training ratio exceeds the ${\sim}14\times$ \emph{per-evaluation} overhead isolated in the decomposition
(Table~\ref{tab:overhead_decomp}).

\begin{table}[t]
\centering
\caption{Experiment~A: DMCI results (6 programs $\times$ 10 seeds, 60 runs). All quantities are \emph{training-loss} measurements: each program is fit on a fixed 8-point input grid with no held-out split, and the synthetic data are noiseless, so training loss equals recovery quality. A run counts as converged iff its best (minimum) training loss falls below $10^{-3}$; all 60 runs converged, so every Train Loss entry is correspondingly below $10^{-3}$. Conv: converged seeds out of 10. Train Loss and Param Error are seed-means of, respectively, the best (minimum) training loss reached and the closest approach $|\hat{\theta} - \theta^*|$ to each target during optimization, i.e.\ the early-stopping / best-checkpoint values, not the last iterate. Avg Epoch is the mean epoch at which the training loss first crosses $10^{-3}$; Avg Time is the mean total run wall-clock. The loss is summed over the input grid, so its magnitude is comparable \emph{within} but not \emph{across} program families.}
\label{tab:exp_a_results}
\small
\begin{tabular}{@{}lrrrrl@{}}
\toprule
Program & Conv & Avg Epoch & Train Loss & Avg Time & Param Error \\
\midrule
P1 (single const) & 10/10 & 20 & $1.5\times10^{-4}$ & 2.4\,s & $\alpha$: 0.0006 \\
P2 (multi const) & 10/10 & 124 & $2.7\times10^{-6}$ & 9.0\,s & $a$: 0.0002, $b$: $<\!10^{-4}$ \\
P3 (recursive) & 10/10 & 61 & $4.0\times10^{-5}$ & 58.6\,s & $\alpha$: 0.0002 \\
P4 (higher-order) & 10/10 & 48 & $1.1\times10^{-5}$ & 10.7\,s & $\alpha$: 0.0005 \\
P5 (multi-function) & 10/10 & 247 & $2.0\times10^{-4}$ & 31.6\,s & $a$: 0.005, $b$: 0.009 \\
P6 (composed) & 10/10 & 2017 & $8.5\times10^{-4}$ & 210.3\,s & $a$: 0.014, $b$: 0.010, $c$: 0.067 \\
\bottomrule
\end{tabular}
\end{table}

% Figure: Exp A - Wall-clock comparison across methods
% Grouped bar chart: mean wall time per program per method
% Usage: \input{figures/fig_a_wallclock.tex}
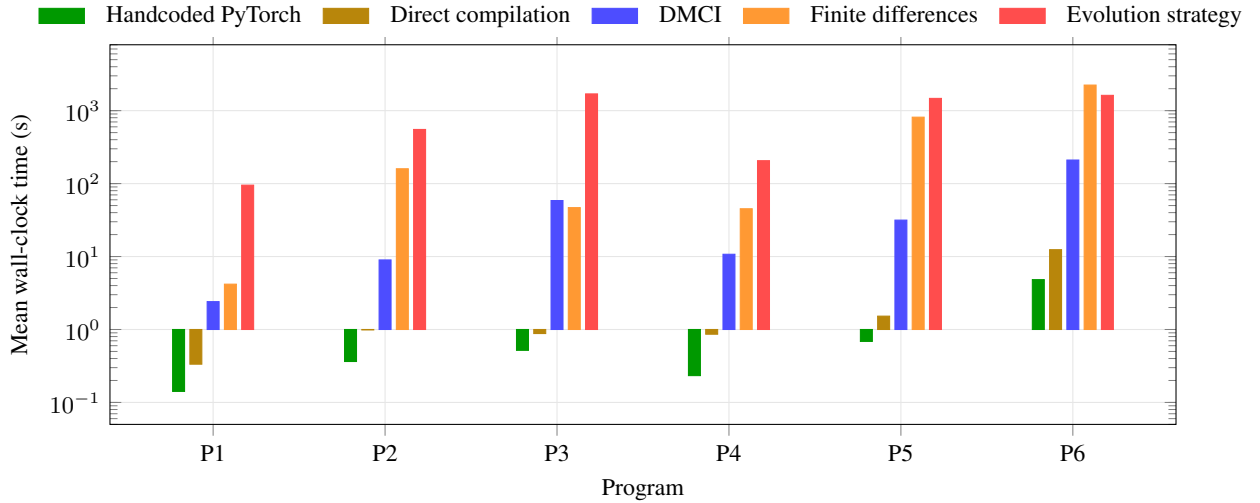
\begin{figure}[tbp]
\centering
\vspace*{2\baselineskip}% extra separation from the table above when placed together
\begin{tikzpicture}
\definecolor{darkgold}{RGB}{184,134,11}
\begin{axis}[
    width=0.95\textwidth,
    height=0.40\textwidth,
    ybar,
    bar width=4.5pt,
    ymode=log,
    ylabel={Mean wall-clock time (s)},
    xlabel={Program},
    symbolic x coords={P1,P2,P3,P4,P5,P6},
    xtick=data,
    x tick label style={font=\small},
    tick label style={font=\small},
    label style={font=\small},
    ymin=0.05, ymax=8000,
    grid=major,
    grid style={gray!20},
    enlarge x limits=0.12,
    legend style={
        font=\small,
        at={(0.5,1.02)},
        anchor=south,
        legend columns=5,
        draw=none,
        column sep=6pt,
    },
    legend cell align={left},
    area legend,
]

% Handcoded PyTorch
\addplot[fill=green!60!black, draw=green!60!black]
    coordinates {
        (P1, 0.1417) (P2, 0.3640) (P3, 0.5181) (P4, 0.2332) (P5, 0.6866) (P6, 4.8235)
    };
\addlegendentry{Handcoded PyTorch}

% Direct compilation
\addplot[fill=darkgold, draw=darkgold]
    coordinates {
        (P1, 0.3346) (P2, 0.9868) (P3, 0.8798) (P4, 0.8615) (P5, 1.5244) (P6, 12.4166)
    };
\addlegendentry{Direct compilation}

% DMCI (compiled interpreter)
\addplot[fill=blue!70, draw=blue!70]
    coordinates {
        (P1, 2.4094) (P2, 8.9860) (P3, 58.5817) (P4, 10.7477) (P5, 31.5880) (P6, 210.3447)
    };
\addlegendentry{DMCI}

% Finite differences
\addplot[fill=orange!80, draw=orange!80]
    coordinates {
        (P1, 4.1827) (P2, 159.7470) (P3, 46.8343) (P4, 45.1982) (P5, 813.8878) (P6, 2241.4814)
    };
\addlegendentry{Finite differences}

% Evolution strategy
\addplot[fill=red!70, draw=red!70]
    coordinates {
        (P1, 95.2060) (P2, 552.2642) (P3, 1696.1473) (P4, 206.0453) (P5, 1472.2267) (P6, 1620.9487)
    };
\addlegendentry{Evolution strategy}

\end{axis}
\end{tikzpicture}
\caption{Mean training wall-clock time across five optimization methods on programs P1--P6 (10 seeds each, log scale, early stopping at convergence). Handcoded PyTorch and direct compilation are fastest but require manual or per-program engineering. DMCI introduces approximately one order of magnitude runtime overhead relative to direct compilation while preserving identical optimization trajectories and convergence behavior. Finite differences and evolution strategies are substantially slower; moreover, finite differences did not converge on P6 and the evolution strategy did not converge on P2, P5, or P6 (those runs reached the termination criterion without meeting the convergence threshold, so their bars reflect time-to-termination rather than time-to-convergence).}
\label{fig:exp_a_wallclock}
\end{figure}

\begin{table}[t]
\centering
\caption{Full Scheme source for the Experiment~A program families (compact summary in Table~\ref{tab:programs}). Each family defines a program with learnable constants (initialized away from targets) to be recovered by gradient descent; complexity increases from scalar arithmetic (P1) to nested function composition (P6).}
\label{tab:programs_full}
\small
\begin{tabular}{@{}lllll@{}}
\toprule
ID & Capability tested & Constants & Targets & Program \\
\midrule
P1 & Scalar arithmetic & $\alpha$ & 0.5 & \texttt{(* alpha (* x x))} \\
\addlinespace[4pt]
P2 & Independent parameters & $a, b$ & 3.0, 0.5 & \texttt{(+ a (* b (* x x)))} \\
\addlinespace[4pt]
P3 & Fixed-depth recursion & $\alpha$ & 2.0 &
  \makecell[tl]{
  \texttt{(define (poly x n)}\\
  \texttt{~~(if (= n 0) 0}\\
  \texttt{~~~~(+ (* alpha x)}\\
  \texttt{~~~~~~~(poly x (- n 1)))))}\\
  \texttt{(poly x 3)}
  } \\
\addlinespace[4pt]
P4 & Higher-order closure application & $\alpha$ & 1.5 &
  \makecell[tl]{
  \texttt{(define (twice f x) (f (f x)))}\\
  \texttt{(twice (lambda (y) (+ y alpha)) x)}
  } \\
\addlinespace[4pt]
P5 & Multiple function scopes & $a, b$ & 2.0, 1.5 &
  \makecell[tl]{
  \texttt{(define (f x) (* a x))}\\
  \texttt{(define (g x) (+ b (* x x)))}\\
  \texttt{(+ (f x) (g x))}
  } \\
\addlinespace[4pt]
P6 & Nested composition & $a, b, c$ & 2.0, 0.5, 1.0 &
  \makecell[tl]{
  \texttt{(define (f x) (+ (* a (* x x)) c))}\\
  \texttt{(define (g x) (+ x b))}\\
  \texttt{(f (g x))}
  } \\
\bottomrule
\end{tabular}
\end{table}

\subsection{Methods}

Five optimization methods are compared:
\begin{enumerate}[leftmargin=2em,topsep=2pt,itemsep=1pt]
    \item \textbf{Direct compilation} (autograd). Compile the target program directly to a differentiable module. This is the upper bound: shortest gradient path.
    \item \textbf{Compiled interpreter} (autograd). DMCI: compile the self-hosted evaluator and pass the target program as quoted data.
    \item \textbf{Hand-coded PyTorch interpreter} (autograd). Tree-walking evaluator in Python/PyTorch ($\sim$100 lines). Isolates the effect of compilation.
    \item \textbf{Finite differences}. Central differences: $(f(\theta+\epsilon) - f(\theta-\epsilon)) / 2\epsilon$.
    \item \textbf{Evolution strategy}. $(\mu,\lambda)$-ES with $1/5$ sigma adaptation.
\end{enumerate}

\subsection{Training Protocol}

For each (method, program, seed) triple: initialize constants with Gaussian perturbation ($\pm 30\%$), controlled by seed. Eight training points on $[0.5, 3.0]$. Adam (lr $= 0.05$), up to 3000 epochs. A run is counted as converged at the first epoch its loss falls below $10^{-3}$, after which training continues for a fixed 50-epoch window and then stops; runs that never cross $10^{-3}$ train the full 3000 epochs. Reported losses are the best (minimum) training loss reached per run (the early-stopping / best-checkpoint value), not the last-epoch value. Ten seeds per pair yield 300 runs.

\subsection{Full Six-Program Convergence Figure}

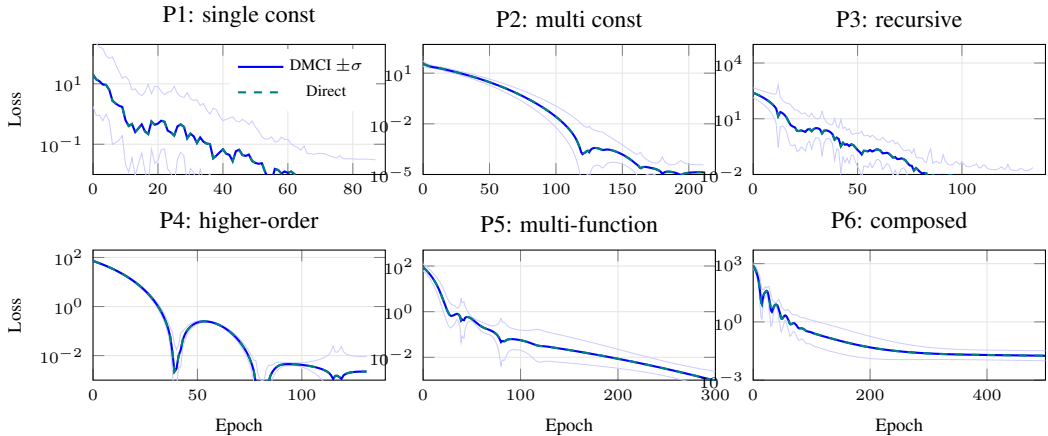
\begin{figure*}[ht]
\centering
\begin{tikzpicture}
\pgfplotsset{
    every axis/.append style={
        width=0.33\textwidth, height=0.20\textwidth,
        tick label style={font=\tiny},
        label style={font=\scriptsize},
        title style={font=\small, yshift=-3pt},
        grid=major, grid style={gray!20},
        every axis plot/.append style={thick, mark=none},
    }
}

\begin{semilogyaxis}[
    name=p1,
    title={P1: single const},
    ylabel={\scriptsize Loss},
    xmin=0, xmax=90, ymin=1e-2, ymax=200,
    legend style={font=\tiny, at={(0.97,0.97)}, anchor=north east, draw=none},
]
\addplot[blue!20, very thin, forget plot] table[x=epoch, y=upper] {figures/exp_a_p1_single_const_compiled_interp.dat};
\addplot[blue!20, very thin, forget plot] table[x=epoch, y=lower] {figures/exp_a_p1_single_const_compiled_interp.dat};
\addplot[blue] table[x=epoch, y=mean] {figures/exp_a_p1_single_const_compiled_interp.dat};
\addlegendentry{DMCI $\pm\sigma$}
\addplot[teal, dashed] table[x=epoch, y=mean] {figures/exp_a_p1_single_const_direct.dat};
\addlegendentry{Direct}
\end{semilogyaxis}

\begin{semilogyaxis}[
    name=p2, at={(p1.east)}, anchor=west, xshift=0.5cm,
    title={P2: multi const},
    xmin=0, xmax=220, ymin=1e-5, ymax=500,
]
\addplot[blue!20, very thin, forget plot] table[x=epoch, y=upper] {figures/exp_a_p2_multi_const_compiled_interp.dat};
\addplot[blue!20, very thin, forget plot] table[x=epoch, y=lower] {figures/exp_a_p2_multi_const_compiled_interp.dat};
\addplot[blue] table[x=epoch, y=mean] {figures/exp_a_p2_multi_const_compiled_interp.dat};
\addplot[teal, dashed] table[x=epoch, y=mean] {figures/exp_a_p2_multi_const_direct.dat};
\end{semilogyaxis}

\begin{semilogyaxis}[
    name=p3, at={(p2.east)}, anchor=west, xshift=0.5cm,
    title={P3: recursive},
    xmin=0, xmax=140, ymin=1e-2, ymax=1e5,
]
\addplot[blue!20, very thin, forget plot] table[x=epoch, y=upper] {figures/exp_a_p3_recursive_compiled_interp.dat};
\addplot[blue!20, very thin, forget plot] table[x=epoch, y=lower] {figures/exp_a_p3_recursive_compiled_interp.dat};
\addplot[blue] table[x=epoch, y=mean] {figures/exp_a_p3_recursive_compiled_interp.dat};
\addplot[teal, dashed] table[x=epoch, y=mean] {figures/exp_a_p3_recursive_direct.dat};
\end{semilogyaxis}

\begin{semilogyaxis}[
    name=p4, at={(p1.south west)}, anchor=north west, yshift=-1.0cm,
    title={P4: higher-order},
    xlabel={\scriptsize Epoch}, ylabel={\scriptsize Loss},
    xmin=0, xmax=140, ymin=1e-3, ymax=200,
]
\addplot[blue!20, very thin, forget plot] table[x=epoch, y=upper] {figures/exp_a_p4_higher_order_compiled_interp.dat};
\addplot[blue!20, very thin, forget plot] table[x=epoch, y=lower] {figures/exp_a_p4_higher_order_compiled_interp.dat};
\addplot[blue] table[x=epoch, y=mean] {figures/exp_a_p4_higher_order_compiled_interp.dat};
\addplot[teal, dashed] table[x=epoch, y=mean] {figures/exp_a_p4_higher_order_direct.dat};
\end{semilogyaxis}

\begin{semilogyaxis}[
    name=p5, at={(p4.east)}, anchor=west, xshift=0.5cm,
    title={P5: multi-function},
    xlabel={\scriptsize Epoch},
    xmin=0, xmax=300, ymin=1e-3, ymax=500,
]
\addplot[blue!20, very thin, forget plot] table[x=epoch, y=upper] {figures/exp_a_p5_multi_function_compiled_interp.dat};
\addplot[blue!20, very thin, forget plot] table[x=epoch, y=lower] {figures/exp_a_p5_multi_function_compiled_interp.dat};
\addplot[blue] table[x=epoch, y=mean] {figures/exp_a_p5_multi_function_compiled_interp.dat};
\addplot[teal, dashed] table[x=epoch, y=mean] {figures/exp_a_p5_multi_function_direct.dat};
\end{semilogyaxis}

\begin{semilogyaxis}[
    name=p6, at={(p5.east)}, anchor=west, xshift=0.5cm,
    title={P6: composed},
    xlabel={\scriptsize Epoch},
    xmin=0, xmax=500, ymin=1e-3, ymax=5000,
]
\addplot[blue!20, very thin, forget plot] table[x=epoch, y=upper] {figures/exp_a_p6_composed_compiled_interp.dat};
\addplot[blue!20, very thin, forget plot] table[x=epoch, y=lower] {figures/exp_a_p6_composed_compiled_interp.dat};
\addplot[blue] table[x=epoch, y=mean] {figures/exp_a_p6_composed_compiled_interp.dat};
\addplot[teal, dashed] table[x=epoch, y=mean] {figures/exp_a_p6_composed_direct.dat};
\end{semilogyaxis}
\end{tikzpicture}
\caption{Full convergence curves for programs P1--P6, mean $\pm$ 1 std over 10 seeds. DMCI (solid blue) and direct compilation (dashed teal) produce indistinguishable trajectories on all six programs. The P6 panel is truncated at 500 epochs for readability; the curve continues to descend to ${\sim}0.0008$ by its mean convergence epoch ${\sim}2{,}017$.}
\label{fig:exp_a_convergence_full}
\end{figure*}

\subsection{Autograd vs.\ Gradient-Free}

Finite differences converge in 50/60 runs (83.3\%) but require $12.5\text{--}534\times$ more wall time than direct compilation (mean ${\sim}195\times$). The evolution strategy converges in only 27/60 runs (45.0\%), failing completely on P2 and P6 (0/10 each) and P5 (1/10); on partially converging programs (e.g., P3: 7/10) the remaining seeds leave best losses above the $10^{-3}$ threshold. These results validate that exact gradients through the interpreter provide qualitatively better optimization than numerical or black-box alternatives.

\subsection{Ablations}

\paragraph{Dict vs.\ tensor heap.} Replacing the dictionary-backed heap with a pre-allocated tensor buffer causes autograd failures after several dozen \texttt{cons} operations due to in-place mutation breaking version tracking. The self-hosted interpreter performs hundreds of \texttt{cons} calls; the dict heap is essential.

\paragraph{Lazy vs.\ eager evaluation.} Switching to eager evaluation of conditionals causes infinite expansion on any recursive program (P3--P6). Lazy evaluation is necessary, not optional, for DMCI.

\paragraph{Gradient path analysis.} Direct compilation produces 9--19 autograd graph nodes; the compiled interpreter produces 18--30 (5--18 extra nodes for interpreter dispatch, environment lookup, and tagged-value wrapping, roughly doubling the count for the simplest programs). Despite the longer path, DMCI gradients still match direct compilation to ${\sim}10^{-6}$ across all programs, so the added dispatch nodes preserve gradient quality rather than eroding it.

\subsection{Overhead Decomposition}

\begin{table}[ht]
\centering
\caption{Overhead decomposition: DMCI vs.\ direct compilation on P2, 100 iterations.}
\label{tab:overhead_decomp}
\small
\begin{tabular}{@{}lrrr@{}}
\toprule
Component & Direct (ms) & DMCI (ms) & \% of overhead \\
\midrule
Tagged-value wrap/unwrap & 8.6 & 240.4 & 41.4\% \\
Python overhead & 29.6 & 209.4 & 32.1\% \\
Evaluator graph walking & 1.7 & 96.5 & 16.9\% \\
Heap operations & 0.0 & 30.0 & 5.4\% \\
Dispatch (\texttt{tagged\_if}/\texttt{eq?}) & 0.0 & 12.5 & 2.2\% \\
Tagged arithmetic wrappers & 0.0 & 10.8 & 1.9\% \\
Raw arithmetic primitives & 2.2 & 2.5 & 0.1\% \\
PyTorch runtime / autograd & 1.8 & 1.9 & 0.0\% \\
\midrule
\textbf{Total} & \textbf{43.9} & \textbf{604.0} & \textbf{13.8$\times$} \\
\bottomrule
\end{tabular}
\end{table}

% Figure: Overhead decomposition horizontal stacked bar chart
% Visualizes where the ~14x DMCI overhead comes from
% Usage: \input{figures/fig_h_overhead_decomp.tex}
% Note: the amortization brace is drawn AFTER \end{axis}. pgfplots defers nodes
% placed at (axis cs:...), so any \draw referencing them must come after the
% axis closes -- drawing inside the axis silently fails ("No shape named ...").
\begin{figure}[t]
\centering
\begin{tikzpicture}
\begin{axis}[
    width=0.92\columnwidth,
    height=0.35\columnwidth,
    xmin=-5, xmax=105,  % symmetric 5-unit margins: bar spans 0..100
    xtick={0,10,20,30,40,50,60,70,80,90,100},
    xlabel={Fraction of DMCI overhead (\%)},
    xlabel style={font=\small, yshift=2pt},
    ytick={0},
    yticklabels={},
    tick label style={font=\small},
    ymin=-1, ymax=1,
    clip=false,
]

% --- Segments drawn as explicit rectangles (x in axis cs, +-14pt thickness).
% Using xbar stacked here draws the first segment from a baseline left of 0
% (overshooting into the xmin margin); explicit fills pin the bar to span 0..100.
\addplot[draw=none, forget plot] coordinates {(0,0) (100,0)};
\fill[blue!80!black] ([yshift=-14pt]axis cs:0,0)    rectangle ([yshift=14pt]axis cs:41.4,0);
\fill[red!70!black]  ([yshift=-14pt]axis cs:41.4,0) rectangle ([yshift=14pt]axis cs:73.5,0);
\fill[teal!80]       ([yshift=-14pt]axis cs:73.5,0) rectangle ([yshift=14pt]axis cs:90.4,0);
\fill[orange!80]     ([yshift=-14pt]axis cs:90.4,0) rectangle ([yshift=14pt]axis cs:95.8,0);
\fill[gray!60]       ([yshift=-14pt]axis cs:95.8,0) rectangle ([yshift=14pt]axis cs:100,0);

% --- Inline percentage labels (inside the three large segments) ---
\node[font=\scriptsize, text=white] at (axis cs:20.7, 0)  {41.4\%};
\node[font=\scriptsize, text=white] at (axis cs:57.45, 0) {32.1\%};
\node[font=\scriptsize, text=white] at (axis cs:81.95, 0) {16.9\%};

% --- Semantic labels BELOW the bar (top zone reserved for the brace) ---
\node[font=\tiny, text=blue!80!black, anchor=north] at (axis cs:20.7, 0)  [yshift=-15pt] {Tagged-value wrapping};
\node[font=\tiny, text=red!70!black,  anchor=north] at (axis cs:57.45, 0) [yshift=-15pt] {Python interpreter};
\node[font=\tiny, text=teal!80!black, anchor=north] at (axis cs:81.95, 0) [yshift=-15pt] {Graph walking};

% Narrow segments: staggered labels (text uses reliable yshift) with leader
% lines. Leader anchors are named here and drawn after \end{axis}; relative
% ++ offsets are unreliable in pgfplots and break when the x-range changes.
\node[font=\tiny, text=orange!70!black, anchor=north] at (axis cs:93.1, 0) [yshift=-27pt] {Heap (5.4\%)};
\node[font=\tiny, text=gray!50!black,  anchor=north] at (axis cs:97.9, 0) [yshift=-38pt] {Other (4.2\%)};
\node[inner sep=0pt, outer sep=0pt] (heapTop) at (axis cs:93.1, 0) [yshift=-13pt] {};
\node[inner sep=0pt, outer sep=0pt] (heapBot) at (axis cs:93.1, 0) [yshift=-22pt] {};
\node[inner sep=0pt, outer sep=0pt] (othTop)  at (axis cs:97.9, 0) [yshift=-13pt] {};
\node[inner sep=0pt, outer sep=0pt] (othBot)  at (axis cs:97.9, 0) [yshift=-33pt] {};

% --- Anchor nodes for the amortization brace ---
% Spans tagged-value wrapping + Python interpreter + graph walking (0 -> 90.4%).
% These materialize only after \end{axis}; the brace/label are drawn below.
% braceL/braceR sit at the bar's left edge (0) and the teal/orange boundary (90.4);
% with explicit-rectangle bars the feet land exactly on these coordinates, and
% braceLab (45.2) centers the tooth + label over the spanned region.
\node[inner sep=0pt, outer sep=0pt] (braceL)   at (axis cs:0, 0) [yshift=16pt] {};
\node[inner sep=0pt, outer sep=0pt] (braceR)   at (axis cs:90.4, 0) [yshift=16pt] {};
\node[inner sep=0pt, outer sep=0pt] (braceLab) at (axis cs:45.2, 0) [yshift=29pt] {};

\end{axis}

% --- Amortization brace + label (tagged-value + Python interpreter + graph walking = 90.4%) ---
\draw[decorate, decoration={brace, amplitude=5pt}, line width=0.9pt, black!75]
    (braceL.center) -- (braceR.center);
\node[font=\tiny, text=black!75, anchor=south, inner sep=1pt] at (braceLab.center)
    {90.4\% per-evaluation overhead, amortized by batching};

% --- Leader lines for the narrow segments (vertical; xmin-independent) ---
\draw[orange!70!black, thin] (heapTop.center) -- (heapBot.center);
\draw[gray!50!black, thin]   (othTop.center)  -- (othBot.center);

\end{tikzpicture}
\caption{Overhead decomposition for DMCI relative to direct compilation (P2, 100 fwd--bwd iterations). The cost is dominated by Python-level bookkeeping: tagged-value wrapping (41.4\%), Python interpreter overhead (32.1\%), and graph walking (16.9\%). These three components (90.4\% combined) are per-evaluation costs paid once per batch rather than once per input, and are therefore amortized by batching (Section~\ref{sec:exp_h}). Raw arithmetic contributes $<$0.2\%, explaining the near-linear batching speedups: the arithmetic that scales with batch size is already efficient.}
\label{fig:overhead_decomp}
\end{figure}
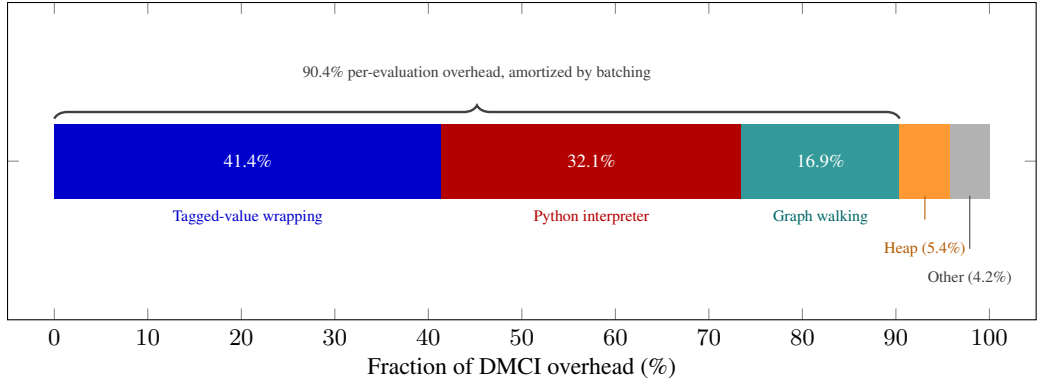

Tagged-value wrapping (41.4\%), Python interpreter overhead (32.1\%), and graph walking (16.9\%), 90.4\% combined (Figure~\ref{fig:overhead_decomp}), are per-evaluation costs paid once per batch and are therefore amortized by batching (Experiment~H): in batched execution each is performed once on a tensor of shape $(B, 14)$ rather than $B$ times on scalars, since the batch dimension is broadcast through the fixed computation graph. Only raw arithmetic ($<$0.2\%) scales with batch size. Reducing single-evaluation \emph{latency} (as opposed to throughput) is a separate problem that batching does not address; it would require compiling the evaluator graph to MLIR with Enzyme for differentiation.

\subsection{S3: Branch-Dependent Constants}
\label{app:s3}

We test the program $\texttt{(if (< x alpha) (* beta x) (* gamma x))}$ with learnable constants $\alpha, \beta, \gamma$ (true values: 1.0, 2.0, 0.5), fitting the piecewise-linear target it produces at those values ($\beta x$ for $x < \alpha$, else $\gamma x$). The comparison $\texttt{(< x alpha)}$ returns a discrete 0 or 1, so $\alpha$ sits inside the branch boundary and receives zero gradient.

Across 10 seeds (with $\alpha$ initialized at 0, away from the target), both direct compilation and DMCI converge 0/10 at final loss $3.93$, with \emph{zero} loss difference between the two methods at every seed. Receiving no gradient, $\alpha$ never leaves its initialization; with $\alpha \approx 0$ all training points fall into a single branch and the program collapses to the best single-line fit ($\gamma \to 0.54$), never recovering the true threshold. Trajectory equivalence thus holds even in this non-convergent case.

A sweep of $\alpha_0 \in [-2, 4]$ (21 points, with $\beta, \gamma$ held at their true values) confirms the mechanism: $\alpha_{\text{final}} = \alpha_0$ exactly at all points. The loss landscape (Figure~\ref{fig:s3_basin}) has a genuine minimum at the true threshold $\alpha^* = 1$ (loss $0$), flanked by flat plateaus where $\alpha$ falls below or above all training data; yet because $\nabla_\alpha = 0$ everywhere, gradient descent cannot descend into that minimum from any initialization.

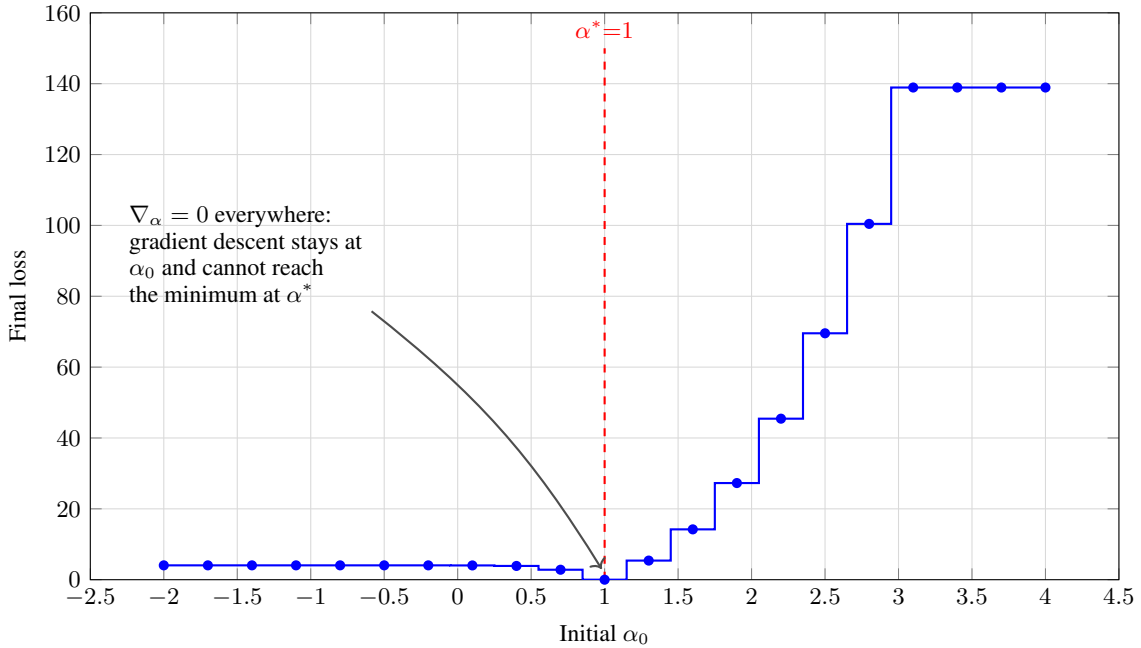
\begin{figure}[ht]
\centering
\begin{tikzpicture}
\begin{axis}[
    width=0.92\columnwidth,
    height=0.55\columnwidth,
    xlabel={Initial $\alpha_0$},
    ylabel={Final loss},
    xmin=-2.5, xmax=4.5,
    ymin=0, ymax=160,
    grid=major,
    grid style={gray!30},
    tick label style={font=\small},
    label style={font=\small},
]
\addplot[blue, thick, const plot mark mid, mark=*, mark size=1.5pt] table[x=alpha0, y=loss] {figures/s3_basin_loss.dat};
% true threshold alpha*=1 is now the global minimum (loss 0)
\draw[red, dashed, thick] (axis cs:1.0,0) -- (axis cs:1.0,150) node[above, font=\footnotesize, text=red] {$\alpha^*{=}1$};
% a real minimum exists, but zero gradient means GD cannot reach it
\node[font=\footnotesize, align=left, anchor=west] (gd) at (axis cs:-2.3,92)
    {$\nabla_\alpha=0$ everywhere:\\gradient descent stays at\\$\alpha_0$ and cannot reach\\the minimum at $\alpha^*$};
\draw[->, gray!60!black, thick] (gd.south east) to[bend left=10] (axis cs:0.98,3);
\end{axis}
\end{tikzpicture}
\caption{S3 basin-of-attraction analysis ($\beta,\gamma$ fixed at their true values; only $\alpha$ swept). A genuine loss minimum exists at the true threshold $\alpha^*{=}1$ (loss~$0$), with flat plateaus where $\alpha$ lies below ($\approx 4$) or above ($\approx 139$) all training points. Because the branch comparison is non-differentiable in $\alpha$ ($\nabla_\alpha\,\texttt{(< x alpha)}=0$), $\alpha_{\text{final}}=\alpha_0$ at every point: gradient descent cannot descend into the minimum from any initialization. The discontinuity is the source program's, inherited rather than introduced by DMCI.}
\label{fig:s3_basin}
\end{figure}

\subsection{Noise Robustness}
\label{app:exp_a_noise}

The main experiments recover constants from noiseless, self-generated data. To test robustness to measurement noise, we re-run the constant-recovery task through the same compiled interpreter on data corrupted by additive Gaussian noise scaled to the signal's standard deviation, $y \mapsto y + \sigma\,\mathrm{std}(y)\,\mathcal{N}(0,1)$, for $\sigma \in \{0, 0.02, 0.05, 0.10, 0.20\}$ ($\sigma$ is the noise level as a fraction of the signal; $\sigma{=}0$ is the noiseless control). For a fair comparison across noise levels (under noise the loss plateaus at the noise floor and never reaches the $10^{-3}$ threshold), each run trains until its own loss plateaus (relative-improvement stop) and we report the constants at the best loss, averaged over five seeds. Data points are batched through a single interpreter walk per epoch (bit-identical to sequential; Appendix~\ref{app:exp_h_ext}).

Table~\ref{tab:exp_a_noise} and Figure~\ref{fig:exp_a_noise} report the mean relative parameter error $\frac{1}{|\theta|}\sum_n |\hat\theta_n - \theta^*_n|/|\theta^*_n|$ versus $\sigma$. Recovery is near-exact in the noiseless case ($10^{-7}$--$10^{-5}$, matching the main results, so the early-stopping protocol introduces no floor artifact) and degrades \emph{gracefully and roughly proportionally} to the noise level: at $\sigma{=}0.10$ ($10\%$ noise) the geometric-mean relative error is $1.9\%$, and at $\sigma{=}0.20$ it is $3.3\%$. The coupled multi-parameter program P5 is the most noise-sensitive ($27.8\%$ at $\sigma{=}0.20$), consistent with the partial non-identifiability of correlated parameters in composed programs (cf.\ P6, Section~\ref{sec:exp_a}); the single-parameter and recursive/higher-order programs (P1, P3, P4) stay within $1$--$2\%$ even at $\sigma{=}0.20$. Gradient correctness is unaffected by noise (the optimizer simply fits a noisier target), so the degradation reflects the statistics of the estimation problem, not the differentiable interpreter.

\begin{table}[ht]
\centering
\caption{Experiment~A noise robustness: mean relative parameter error (5 seeds) vs.\ noise level $\sigma$ (additive Gaussian scaled to the signal std). Recovery is near-exact at $\sigma{=}0$ and degrades gracefully; the coupled multi-parameter program P5 is the most sensitive.}
\label{tab:exp_a_noise}
\begin{tabular}{lccccc}
\toprule
Program & $\sigma{=}0$ & $\sigma{=}0.02$ & $\sigma{=}0.05$ & $\sigma{=}0.10$ & $\sigma{=}0.20$ \\
\midrule
P1 single const   & $8.7\times10^{-6}$ & $4.4\times10^{-3}$ & $9.5\times10^{-3}$ & $1.5\times10^{-2}$ & $1.9\times10^{-2}$ \\
P2 multi const    & $1.5\times10^{-7}$ & $6.1\times10^{-3}$ & $1.2\times10^{-2}$ & $2.2\times10^{-2}$ & $4.2\times10^{-2}$ \\
P3 recursive      & $3.4\times10^{-6}$ & $2.7\times10^{-3}$ & $5.1\times10^{-3}$ & $1.0\times10^{-2}$ & $1.4\times10^{-2}$ \\
P4 higher-order   & $2.9\times10^{-5}$ & $1.5\times10^{-3}$ & $2.6\times10^{-3}$ & $6.7\times10^{-3}$ & $1.1\times10^{-2}$ \\
P5 multi-function & $3.0\times10^{-8}$ & $2.6\times10^{-2}$ & $5.8\times10^{-2}$ & $1.1\times10^{-1}$ & $2.8\times10^{-1}$ \\
\midrule
\textbf{geometric mean} & $1.3\times10^{-6}$ & $4.9\times10^{-3}$ & $9.8\times10^{-3}$ & $1.9\times10^{-2}$ & $3.3\times10^{-2}$ \\
\bottomrule
\end{tabular}
\end{table}

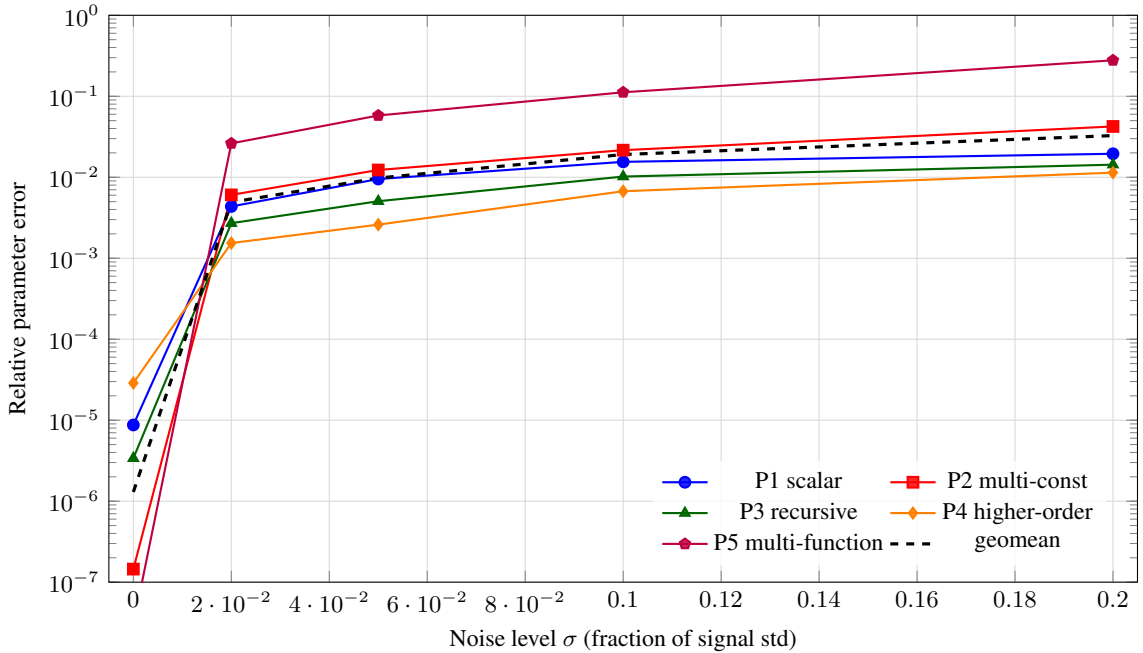
\begin{figure}[ht]
\centering
\begin{tikzpicture}
\begin{semilogyaxis}[
    width=0.92\columnwidth, height=0.55\columnwidth,
    xlabel={Noise level $\sigma$ (fraction of signal std)}, ylabel={Relative parameter error},
    xmin=-0.005, xmax=0.205, ymin=1e-7, ymax=1,
    legend pos=south east, legend columns=2,
    legend style={font=\footnotesize, draw=none, fill=white, fill opacity=0.85, text opacity=1},
    grid=major, grid style={gray!30},
    tick label style={font=\small}, label style={font=\small},
]
\addplot[blue, thick, mark=*] table[x=sigma, y=P1] {figures/exp_a_noise.dat};
\addlegendentry{P1 scalar}
\addplot[red, thick, mark=square*] table[x=sigma, y=P2] {figures/exp_a_noise.dat};
\addlegendentry{P2 multi-const}
\addplot[black!60!green, thick, mark=triangle*] table[x=sigma, y=P3] {figures/exp_a_noise.dat};
\addlegendentry{P3 recursive}
\addplot[orange, thick, mark=diamond*] table[x=sigma, y=P4] {figures/exp_a_noise.dat};
\addlegendentry{P4 higher-order}
\addplot[purple, thick, mark=pentagon*] table[x=sigma, y=P5] {figures/exp_a_noise.dat};
\addlegendentry{P5 multi-function}
\addplot[black, very thick, dashed] table[x=sigma, y=agg] {figures/exp_a_noise.dat};
\addlegendentry{geomean}
\end{semilogyaxis}
\end{tikzpicture}
\caption{Experiment~A noise robustness: relative parameter-recovery error vs.\ noise level $\sigma$ (additive Gaussian, fraction of signal std; mean over 5 seeds, log scale). Recovery is near-exact at $\sigma{=}0$ and degrades roughly proportionally to $\sigma$; the coupled program P5 is most sensitive, the single-parameter/recursive programs least.}
\label{fig:exp_a_noise}
\end{figure}

\subsection{Full Five-Method Baselines}
\label{app:exp_a_full}

\begin{table}[ht]
\centering
\caption{Experiment~A: aggregate results across 10 seeds per (method, program) pair. Loss is the all-seed mean of the best (minimum) training loss reached per seed (the early-stopping / best-checkpoint value, matching Table~\ref{tab:exp_a_results}); a run is converged iff this best loss falls below $10^{-3}$, so the Conv and Loss columns are mutually consistent. Reporting the best checkpoint rather than the last iterate also avoids the divergent-outlier inflation that affects last-epoch means (notably for the evolution strategy, whose failed programs nonetheless reach best losses far above the $10^{-3}$ threshold).}
\label{tab:exp_a_full}
\small
\begin{tabular}{@{}llrrrr@{}}
\toprule
Method & Program & Conv & Epoch & Loss & Wall (s) \\
\midrule
\multirow{6}{*}{DMCI} & P1 single const & 10/10 & 20 & $1.5\times10^{-4}$ & 2.4 \\
 & P2 multi const & 10/10 & 124 & $2.7\times10^{-6}$ & 9.0 \\
 & P3 recursive & 10/10 & 61 & $4.0\times10^{-5}$ & 58.6 \\
 & P4 higher order & 10/10 & 48 & $1.1\times10^{-5}$ & 10.7 \\
 & P5 multi function & 10/10 & 247 & $2.0\times10^{-4}$ & 31.6 \\
 & P6 composed & 10/10 & 2017 & $8.5\times10^{-4}$ & 210.3 \\
\midrule
\multirow{6}{*}{Direct} & P1 single const & 10/10 & 20 & $1.5\times10^{-4}$ & 0.3 \\
 & P2 multi const & 10/10 & 124 & $2.7\times10^{-6}$ & 1.0 \\
 & P3 recursive & 10/10 & 61 & $4.0\times10^{-5}$ & 0.9 \\
 & P4 higher order & 10/10 & 48 & $1.1\times10^{-5}$ & 0.9 \\
 & P5 multi function & 10/10 & 247 & $2.0\times10^{-4}$ & 1.5 \\
 & P6 composed & 10/10 & 2017 & $8.5\times10^{-4}$ & 12.4 \\
\midrule
\multirow{6}{*}{Hand-coded} & P1 single const & 10/10 & 20 & $1.5\times10^{-4}$ & 0.1 \\
 & P2 multi const & 10/10 & 124 & $2.7\times10^{-6}$ & 0.4 \\
 & P3 recursive & 10/10 & 61 & $4.0\times10^{-5}$ & 0.5 \\
 & P4 higher order & 10/10 & 48 & $1.1\times10^{-5}$ & 0.2 \\
 & P5 multi function & 10/10 & 247 & $2.0\times10^{-4}$ & 0.7 \\
 & P6 composed & 10/10 & 2017 & $8.5\times10^{-4}$ & 4.8 \\
\midrule
\multirow{6}{*}{Finite diff.} & P1 single const & 10/10 & 11 & $<\!10^{-3}$ & 4.2 \\
 & P2 multi const & 10/10 & 861 & $5.5\times10^{-4}$ & 159.7 \\
 & P3 recursive & 10/10 & 8 & $<\!10^{-3}$ & 46.8 \\
 & P4 higher order & 10/10 & 84 & $<\!10^{-3}$ & 45.2 \\
 & P5 multi function & 10/10 & 1462 & $7.9\times10^{-4}$ & 813.9 \\
 & P6 composed & 0/10 & n/a & 3.5 & 2241.5 \\
\midrule
\multirow{6}{*}{Evol.\ strategy} & P1 single const & 9/10 & 17 & $2.7\times10^{-4}$ & 95.2 \\
 & P2 multi const & 0/10 & n/a & 0.21 & 552.3 \\
 & P3 recursive & 7/10 & 36 & $1.6\times10^{-3}$ & 1696.1 \\
 & P4 higher order & 10/10 & 16 & $4.9\times10^{-5}$ & 206.0 \\
 & P5 multi function & 1/10 & 44 & $1.8\times10^{-2}$ & 1472.2 \\
 & P6 composed & 0/10 & n/a & 0.16 & 1620.9 \\
\bottomrule
\end{tabular}
\end{table}

The three autograd methods converge on identical epochs because they compute identical gradients; the only difference is wall-clock cost. DMCI averages ${\sim}19\times$ the wall time of direct compilation and ${\sim}48\times$ hand-coded PyTorch. Finite differences fail on P6 (3 parameters, composed structure). Evolution strategies fail on P2, P5, and P6.

% ============================================================================
\section{Experiment B: Extended Results}
\label{app:exp_b_ext}
% ============================================================================

\subsection{System Architecture}

An LLM receives a natural-language description of a scientific model and generates a Scheme program encoding the known physics. The compiler translates this to a differentiable module with learnable parameters. The module trains end-to-end via backpropagation: gradients flow from a data-fitting loss through the compiled program to recover the model's physical constants.

We use a locally-hosted Qwen3.6-35B model via an OpenAI-compatible API. The system prompt constrains Scheme output to binary-only arithmetic operators, explicit recursion via \texttt{letrec}, and numeric-only control flow, requirements imposed by the compiler's tagged-value representation. After iterative prompt refinement of the shared system prompt, the final prompt yields correct, compilable Scheme for all 15 models; the cached generations are committed for reproducibility.

\subsection{Models}

\begin{table}[ht]
\centering
\caption{Experiment~B: 15 LLM-generated scientific models. All programs generated by Qwen3.6-35B and compiled without modification.}
\label{tab:exp_b_models}
\small
\begin{tabular}{@{}llrl@{}}
\toprule
ID & Model & Params & Domain \\
\midrule
M01 & Coulomb's law & 1 & Electrostatics \\
M02 & Beer--Lambert law & 1 & Spectroscopy \\
M03 & Michaelis--Menten & 2 & Enzyme kinetics \\
M04 & Arrhenius equation & 2 & Chemical kinetics \\
M05 & Damped oscillator & 3 & Oscillatory dynamics \\
M06 & Logistic growth & 2 & Population dynamics \\
M07 & Power law & 2 & Empirical scaling \\
M08 & Euler ODE solver & 1 & Numerical methods \\
M09 & Taylor series $e^{ax}$ & 1 & Series approximation \\
M10 & SiLU activation & 2 & ML activation \\
M11 & Recursive filter & 1 & Signal processing \\
M12 & Newton $\sqrt{x}$ & 1 & Iterative refinement \\
M13 & Composed transforms & 2 & Function composition \\
M14 & Anomaly scorer & 3 & Weighted scoring \\
M15 & Horner evaluation & 4 & Polynomial evaluation \\
\bottomrule
\end{tabular}
\end{table}

\subsection{Experimental Protocol}

For each model, we compare four methods: (1)~DMCI (the LLM-generated program interpreted through the compiled self-hosted evaluator), (2)~direct compilation of the same Scheme program without the interpreter layer, (3)~hand-coded PyTorch implementing the same mathematical model manually, and (4)~a pure MLP (two hidden layers, 64 units, tanh activation) learning the input-output mapping from data with no physics.

Training uses Adam (lr $= 0.05$) for up to 3000 epochs. A run is counted as \emph{converged} at the first epoch its loss falls below the threshold $10^{-3}$; training then continues for a fixed 50-epoch window before stopping, and we report the loss at this final epoch (not the best epoch), so it can exceed $10^{-3}$ when a run drifts after first crossing the threshold. Runs that never reach $10^{-3}$ train the full 3000 epochs and are counted as non-converged. Five seeds per (method, model) pair yield 300 total runs. Each seed produces a different initial parameter perturbation ($\pm 30\%$ Gaussian noise around a common starting point).

\subsection{Results}

\paragraph{Trajectory equivalence extends to LLM-generated programs.} DMCI and direct compilation match on convergence epoch for 75/75 (model, seed) pairs with zero final-loss difference. The hand-coded PyTorch function (a manual reimplementation of each model, not the tree-walking interpreter used in Experiment~A) matches DMCI's convergence epoch on 73/75 pairs; the two mismatches (both on M11, the recursive filter) differ by 6--8 epochs and at most $1.4 \times 10^{-4}$ in final loss, attributable to minor floating-point divergence in this iterative model. The key result: \emph{LLM-generated Scheme programs, compiled through the self-hosted interpreter, produce the same optimization dynamics as hand-coded PyTorch implementations.}

\paragraph{Universal convergence for physics-informed methods.} All three physics-informed methods (DMCI, direct compilation, and hand-coded PyTorch) converge on every run (75/75 = 100\%). Models spanning 1--4 parameters, with and without recursion, all reach the convergence threshold. Table~\ref{tab:exp_b_results} summarizes per-model DMCI results.

\begin{table}[ht]
\centering
\caption{Experiment~B DMCI results (5 seeds each). All 75 DMCI runs reach the convergence criterion (loss first below $10^{-3}$); the reported Avg Loss is the mean final-epoch loss, which can exceed $10^{-3}$ for a few models (notably M09, inflated by one post-convergence diverging seed). The MLP Conv column reports pure-MLP convergence (65/75) for comparison.}
\label{tab:exp_b_results}
\small
\begin{tabular}{@{}lrrrrr@{}}
\toprule
Model & Params & Avg Epoch & Avg Loss & Avg Time & MLP Conv \\
\midrule
M01 Coulomb & 1 & 556 & 0.0000 & 222\,s & 0/5 \\
M02 Beer--Lambert & 1 & 66 & 0.0016 & 21\,s & 5/5 \\
M03 Michaelis--Menten & 2 & 1753 & 0.0007 & 378\,s & 5/5 \\
M04 Arrhenius & 2 & 308 & 0.0001 & 130\,s & 5/5 \\
M05 Damped oscillator & 3 & 135 & 0.0000 & 103\,s & 1/5 \\
M06 Logistic growth & 2 & 388 & 0.0001 & 231\,s & 5/5 \\
M07 Power law & 2 & 752 & 0.0003 & 232\,s & 4/5 \\
M08 Euler ODE & 1 & 23 & 0.0030 & 241\,s & 5/5 \\
M09 Taylor $e^{ax}$ & 1 & 165 & 0.4803 & 3127\,s & 5/5 \\
M10 SiLU & 2 & 248 & 0.0004 & 160\,s & 5/5 \\
M11 Recursive filter & 1 & 17 & 0.0002 & 257\,s & 5/5 \\
M12 Newton $\sqrt{x}$ & 1 & 0 & 0.0000 & 203\,s & 5/5 \\
M13 Composed transforms & 2 & 356 & 0.0003 & 215\,s & 5/5 \\
M14 Anomaly scorer & 3 & 100 & 0.0002 & 182\,s & 5/5 \\
M15 Horner polynomial & 4 & 2128 & 0.0010 & 1992\,s & 5/5 \\
\bottomrule
\end{tabular}
\end{table}

\paragraph{Pure MLP baseline.} The MLP converges on 65/75 runs (87\%), failing entirely on M01 (Coulomb, $1/r^2$ singularity, 0/5), with 1/5 on M05 (damped oscillator) and 4/5 on M07 (power law). Among the 14 models where the MLP converges at least once, the comparison is uneven: on the models whose compiled form encodes the exact functional structure its mean loss is $50\times$ (M02) to ${\sim}4{,}100\times$ (M14) worse than DMCI, whereas on the models with smooth, monotonic targets the gap closes and the MLP attains comparable or lower loss (M08, M10, M15); it remains worse on other recursive models, such as the M11 filter ($4.4\times$). Where the compiled program already encodes the exact functional structure, it only tunes parameters; where the target is smooth and monotonic, the MLP's flexibility suffices. We exclude M12 from these ratios: DMCI reaches the exact-solution floor (${\sim}10^{-12}$), making the ratio degenerate. Figure~\ref{fig:exp_b_loss} plots the per-model final loss for all four methods.

% Figure: Exp B - Final loss comparison across methods
% Grouped bar chart: DMCI vs Direct vs Handcoded vs MLP per model
% Usage: \input{figures/fig_b_loss_comparison.tex}
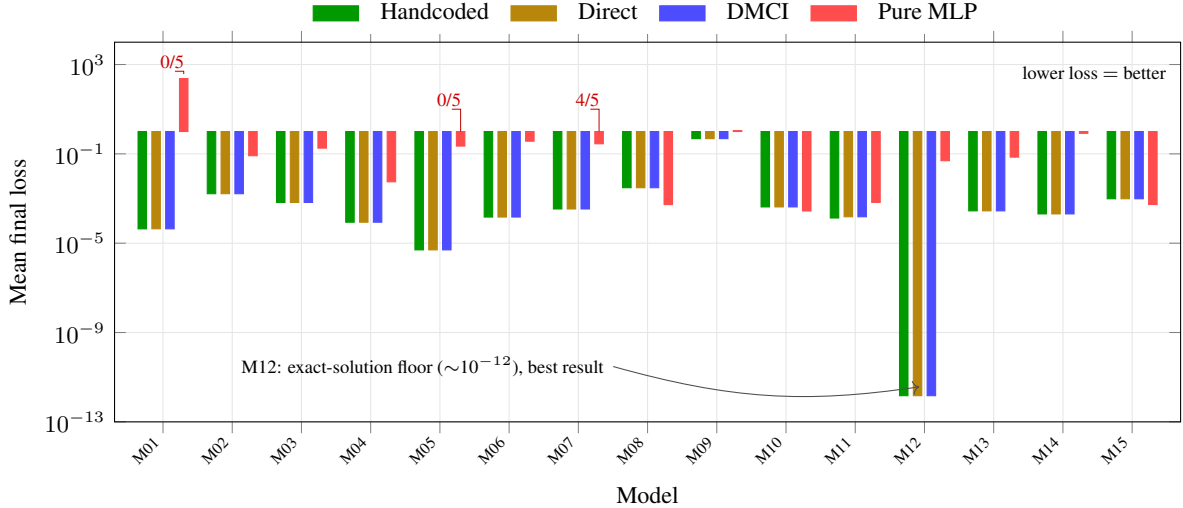
\begin{figure*}[t]
\centering
\begin{tikzpicture}
\definecolor{darkgold}{RGB}{184,134,11}
\begin{axis}[
    width=0.95\textwidth,
    height=0.40\textwidth,
    ybar,
    bar width=3.2pt,
    ymode=log,
    ylabel={Mean final loss},
    xlabel={Model},
    symbolic x coords={M01,M02,M03,M04,M05,M06,M07,M08,M09,M10,M11,M12,M13,M14,M15},
    xtick=data,
    x tick label style={font=\tiny, rotate=45, anchor=east},
    tick label style={font=\small},
    label style={font=\small},
    ymin=1e-13, ymax=1e4,
    grid=major,
    grid style={gray!20},
    enlarge x limits=0.05,
    legend style={
        font=\small,
        at={(0.5,1.02)},
        anchor=south,
        legend columns=4,
        draw=none,
        column sep=6pt,
    },
    legend cell align={left},
    area legend,
]

% Handcoded PyTorch
\addplot[fill=green!60!black, draw=green!60!black]
    coordinates {
        (M01, 4.37e-05) (M02, 1.64e-03) (M03, 6.67e-04) (M04, 8.69e-05) (M05, 5.01e-06)
        (M06, 1.48e-04) (M07, 3.40e-04) (M08, 3.04e-03) (M09, 4.80e-01) (M10, 4.20e-04)
        (M11, 1.34e-04) (M12, 1.48e-12) (M13, 2.80e-04) (M14, 2.04e-04) (M15, 9.73e-04)
    };
\addlegendentry{Handcoded}

% Direct compiled
\addplot[fill=darkgold, draw=darkgold]
    coordinates {
        (M01, 4.39e-05) (M02, 1.64e-03) (M03, 6.67e-04) (M04, 8.69e-05) (M05, 5.01e-06)
        (M06, 1.48e-04) (M07, 3.40e-04) (M08, 3.04e-03) (M09, 4.80e-01) (M10, 4.20e-04)
        (M11, 1.52e-04) (M12, 1.48e-12) (M13, 2.80e-04) (M14, 2.04e-04) (M15, 9.73e-04)
    };
\addlegendentry{Direct}

% DMCI
\addplot[fill=blue!70, draw=blue!70]
    coordinates {
        (M01, 4.39e-05) (M02, 1.64e-03) (M03, 6.67e-04) (M04, 8.69e-05) (M05, 5.01e-06)
        (M06, 1.48e-04) (M07, 3.40e-04) (M08, 3.04e-03) (M09, 4.80e-01) (M10, 4.20e-04)
        (M11, 1.52e-04) (M12, 1.48e-12) (M13, 2.80e-04) (M14, 2.04e-04) (M15, 9.73e-04)
    };
\addlegendentry{DMCI}

% Pure MLP
\addplot[fill=red!70, draw=red!70]
    coordinates {
        (M01, 2.37e+02) (M02, 8.25e-02) (M03, 1.80e-01) (M04, 5.60e-03) (M05, 2.23e-01)
        (M06, 3.70e-01) (M07, 2.82e-01) (M08, 5.35e-04) (M09, 1.10e+00) (M10, 2.77e-04)
        (M11, 6.71e-04) (M12, 4.93e-02) (M13, 7.14e-02) (M14, 8.29e-01) (M15, 5.39e-04)
    };
\addlegendentry{Pure MLP}

% --- Annotations ----------------------------------------------------------
% Direction cue (top-right corner is empty)
\node[font=\scriptsize, anchor=north east] at (rel axis cs:0.995,0.96) {lower loss $=$ better};
% Mark pure-MLP non-convergence: seed-convergence count shown where below 5/5.
% Bars descend from the y=10^0 baseline, so labels are placed in the empty band
% above each model's bars (only M01's MLP bar exceeds 1) to avoid overlap, with a
% thin leader (vertical drop + short horizontal jog) tying each to its MLP bar
% (the rightmost bar of the group, ~+4.8pt from the model's symbolic x).
\node[red!80!black, font=\scriptsize, anchor=south, xshift=3.5pt, inner sep=1pt] (lbl01) at (axis cs:M01, 5e+02) {0/5};
\node[red!80!black, font=\scriptsize, anchor=south, xshift=3.5pt, inner sep=1pt] (lbl05) at (axis cs:M05, 1e+01) {0/5};
\node[red!80!black, font=\scriptsize, anchor=south, xshift=3.5pt, inner sep=1pt] (lbl07) at (axis cs:M07, 1e+01) {4/5};
\draw[red!75!black, line width=0.3pt, shorten <=1pt] (lbl01.south) -| ([xshift=7.8pt] axis cs:M01, 3.72e+02);
\draw[red!75!black, line width=0.3pt, shorten <=1pt] (lbl05.south) -| ([xshift=7.8pt] axis cs:M05, 1);
\draw[red!75!black, line width=0.3pt, shorten <=1pt] (lbl07.south) -| ([xshift=7.8pt] axis cs:M07, 1);
% Highlight M12 as the BEST result so the longest bars are not misread as worst.
\node[font=\scriptsize, anchor=west, align=left] (m12note) at (axis cs:M02, 3e-11)
    {M12: exact-solution floor ($\sim$10$^{-12}$), best result};
\draw[->, gray!60!black, shorten >=2pt] (m12note.east) to[bend right=12] (axis cs:M12, 4e-12);

\end{axis}
\end{tikzpicture}
\caption{Final loss across 15 LLM-generated models (Experiment~B, 5 seeds each; log scale, lower is better). The Handcoded, Direct, and DMCI bars sit at equal height on every model: DMCI and direct compilation are numerically identical, and handcoded differs only marginally (M01, M11), consistent with the trajectory-level gradient equivalence found for Experiment~B (all 75/75 (model, seed) pairs match on convergence epoch with zero final-loss difference). The pure MLP baseline matches or beats DMCI on the models with smooth, monotonic targets (M08, M10, M15) but is up to ${\sim}4{,}300\times$ worse on the sharp closed-form models (M14) and fails to converge on M01 (Coulomb) and M05 (damped oscillator); red labels give the MLP seed-convergence count wherever it falls below 5/5. The gap is not simply about program structure: the MLP is also worse on the recursive filter M11 and the iterative Newton solver M12, so the programmatic inductive bias DMCI preserves automatically is most valuable on the sharp closed-form targets. M12 attains the exact-solution floor (${\sim}10^{-12}$), the best result in the panel; because the bars are drawn from the $10^{0}$ baseline of the log axis, this best result yields the longest bars.}
\label{fig:exp_b_loss}
\end{figure*}

\paragraph{Wall-clock overhead.} DMCI averages 513\,s per run vs.\ 7.0\,s for direct compilation ($73\times$ overhead) and 2.1\,s for hand-coded PyTorch ($248\times$). The overhead is entirely wall-clock; gradient quality is identical. Complex recursive models (M09 Taylor series: 3127\,s; M15 Horner: 1992\,s) dominate the cost due to deep computation graphs through the interpreter. Simple equation models (M02 Beer--Lambert: 21\,s) incur minimal overhead.

\subsection{Why the Extended Language Matters}

With the first-order expression language of the original Neural Compiler~\citep{sheneman2026neural}, these models would be limited to flat arithmetic formulas. The self-hosting Scheme subset lets them use the interpreter's recursion, closures, and helper functions:
\begin{itemize}[leftmargin=2em,topsep=2pt,itemsep=1pt]
    \item Recursion and iteration: M08 (Euler integration, 10-step loop), M09 (Taylor series, 9-term summation), M11 (recursive filter, 8 steps), M12 (Newton's method, 5 iterations)
    \item Closures and higher-order composition: M13 (composed transforms, lambda-based pipeline)
\end{itemize}
Four of the 15 models are implemented with recursion or iteration and one more with closures and higher-order functions; the remaining ten are first-order arithmetic expressions of the kind the original Neural Compiler already handled.

% ============================================================================
\section{Experiment C: Extended Results}
\label{app:exp_c_ext}
% ============================================================================

\subsection{Models}

We select eight models across three categories (Table~\ref{tab:exp_c_models}):

\begin{itemize}[leftmargin=2em,topsep=2pt,itemsep=1pt]
    \item \textbf{Coupled ODEs} (C01--C03): Systems of differential equations solved by recursive Euler integration. Lotka--Volterra (predator-prey), SIR epidemic dynamics, and radioactive decay chains. Each requires 20--30 recursive integration steps with coupled state variables.
    \item \textbf{Iterative/recursive computations} (C04--C05): The logistic map (iterated 10 times) and continued fraction expansion (depth 8). These test gradient flow through long chains of dependent computations.
    \item \textbf{Recursive filters} (C06--C08): Damped pendulum (20-step ODE), second-order IIR filter (12 steps with feedback), and a two-stage cascaded exponential moving average (10 steps). Signal processing models where output at each step depends on previous outputs.
\end{itemize}

\begin{table}[ht]
\centering
\caption{Experiment~C: eight recursive scientific models.}
\label{tab:exp_c_models}
\small
\begin{tabular}{@{}llrrl@{}}
\toprule
ID & Model & Params & Depth & Category \\
\midrule
C01 & Lotka--Volterra & 2 & 20 & Coupled ODE \\
C02 & SIR epidemic & 2 & 30 & Coupled ODE \\
C03 & Decay chain & 2 & 25 & Coupled ODE \\
C04 & Logistic map & 1 & 10 & Iterative \\
C05 & Continued fraction & 1 & 8 & Recursive \\
C06 & Damped pendulum & 2 & 20 & Recursive filter \\
C07 & IIR filter & 3 & 12 & Recursive filter \\
C08 & Cascaded EMA & 2 & 10 & Recursive filter \\
\bottomrule
\end{tabular}
\end{table}

\subsection{Experimental Protocol}

Four methods as in Experiment~B: DMCI, direct compilation, hand-coded PyTorch, and pure MLP. Each model is fit to 20 $(\text{input}, \text{target})$ points sampled evenly over that model's input range (e.g.\ $[2.0, 15.0]$ for C01, $[0.1, 1.5]$ for C06). Training uses Adam (lr $= 0.05$), up to 3000 epochs with the same convergence-based stopping rule as Experiment~B (training continues for 50 epochs after the loss first falls below $10^{-3}$, or runs the full 3000 epochs if it never does). Five seeds per (method, model) pair for six models; three seeds for C01 and C06 (whose long convergence, ${\sim}1{,}400$--$1{,}800$ epochs, produces DMCI runtimes of 6--9.6 hours per seed). Total: 144 runs across the four methods, all complete.

\subsection{Results}

\paragraph{Trajectory equivalence through recursion.} Among the 36 completed (model, seed) pairs with all four methods, DMCI and direct compilation match on convergence epoch for 36/36 pairs. DMCI and hand-coded PyTorch match on 35/36 (the single mismatch, on C08 seed~3, shows convergence epochs of 43 vs.\ 37, a six-epoch difference from floating-point accumulation across 10 recursive steps). Gradients remain faithful to direct compilation even through 30 recursive calls.

\paragraph{All physics-informed methods converge.} DMCI converges on 36/36 completed runs (100\%), direct compilation 36/36, and hand-coded PyTorch 36/36. Table~\ref{tab:exp_c_results} summarizes per-model results.

\begin{table}[ht]
\centering
\caption{Experiment~C results (36 (model, seed) pairs per method, all complete). DMCI converges on all 36 runs; the pure MLP fails to converge on C01 (Lotka--Volterra; all 3 seeds), shown as n/a. The Loss Ratio column is the mean-final-loss ratio (MLP/DMCI); for C07 and C08 this mean is inflated by a single post-convergence diverging MLP seed (robust medians ${\sim}13\times$ and ${\sim}2\times$, respectively). The displayed Loss columns are rounded, whereas ratios are computed from unrounded means, so dividing the rounded columns need not reproduce the ratio exactly (e.g.\ C08: $0.3699/0.0003356=1{,}102$, not $0.3699/0.0003=1{,}233$).}
\label{tab:exp_c_results}
\small
\begin{tabular}{@{}lrrrrrr@{}}
\toprule
Model & Depth & \multicolumn{2}{c}{DMCI} & \multicolumn{2}{c}{MLP} & Loss \\
 & & Epoch & Loss & Epoch & Loss & Ratio \\
\midrule
C01 Lotka--Volterra & 20 & 1410 & 0.0008 & n/a & n/a & n/a \\
C02 SIR epidemic & 30 & 5 & 0.0000 & 140 & 0.0002 & 25$\times$ \\
C03 Decay chain & 25 & 32 & 0.0061 & 611 & 0.2559 & 42$\times$ \\
C04 Logistic map & 10 & 33 & 0.0001 & 70 & 0.0004 & 4$\times$ \\
C05 Continued fraction & 8 & 37 & 0.0023 & 123 & 0.0004 & 0.2$\times$ \\
C06 Damped pendulum & 20 & 1829 & 0.0009 & 227 & 0.0006 & 0.7$\times$ \\
C07 IIR filter & 12 & 49 & 0.0004 & 1339 & 0.0328 & 74$\times$ \\
C08 Cascaded EMA & 10 & 24 & 0.0003 & 939 & 0.3699 & 1102$\times$ \\
\bottomrule
\end{tabular}
\end{table}

\begin{figure}[ht]
\centering
\begin{tikzpicture}
\begin{semilogyaxis}[
    width=0.92\columnwidth,
    height=0.55\columnwidth,
    xlabel={Epoch},
    ylabel={MSE Loss},
    xmin=0, xmax=150,
    ymin=1e-4, ymax=300,
    legend pos=north east,
    legend style={font=\small, draw=none, fill=white, fill opacity=0.8, text opacity=1},
    grid=major,
    grid style={gray!30},
    tick label style={font=\small},
    label style={font=\small},
]
\addplot[blue, thick, mark=none] table[x=epoch, y=mean] {figures/c08_dmci_log.dat};
\addlegendentry{DMCI}
\addplot[red!70!black, only marks, mark=diamond*, mark size=2.5pt, mark repeat=6] table[x=epoch, y=mean] {figures/c08_direct_log.dat};
\addlegendentry{Direct compilation}
\addplot[red, thick, mark=none] table[x=epoch, y=mean] {figures/c08_mlp_log.dat};
\addlegendentry{Pure MLP}
\draw[<->, gray, thick] (axis cs:93,0.000336) -- (axis cs:93,0.0237) node[midway, right, font=\footnotesize, text=black] {$71\times$};
\end{semilogyaxis}
\end{tikzpicture}
\caption{C08 (cascaded EMA): loss trajectory for a two-stage recursive filter with two learnable smoothing constants ($\alpha,\beta$) and recursion depth~10, plotted as mean MSE over 5 seeds on a log scale through epoch~150. DMCI (blue line) and direct compilation (red diamonds) are bit-for-bit identical at every epoch and seed (the curves overplot exactly) confirming that gradients remain faithful through all 10 recursive steps; the two methods differ only in wall-clock cost. The pure MLP (red), lacking the program's structural prior, converges far more slowly: averaged over seeds it first reaches loss~$<\!10^{-3}$ at epoch~939 versus epoch~24 for DMCI, a ${\sim}39\times$ convergence-speed gap. The arrow marks the ${\sim}71\times$ DMCI--MLP loss gap at epoch~93 (a mid-training snapshot). The larger mean final-loss gap quoted in the text ($1{,}102\times$) is inflated by one post-convergence diverging MLP seed (final loss~$1.85$); the other four seeds finish near $5$--$7\times10^{-4}$ ($\approx2\times$ DMCI's $3.4\times10^{-4}$), so the robust, all-seed effect is faster and more stable convergence rather than the mean final-loss ratio.}
\label{fig:c08_loss}
\end{figure}
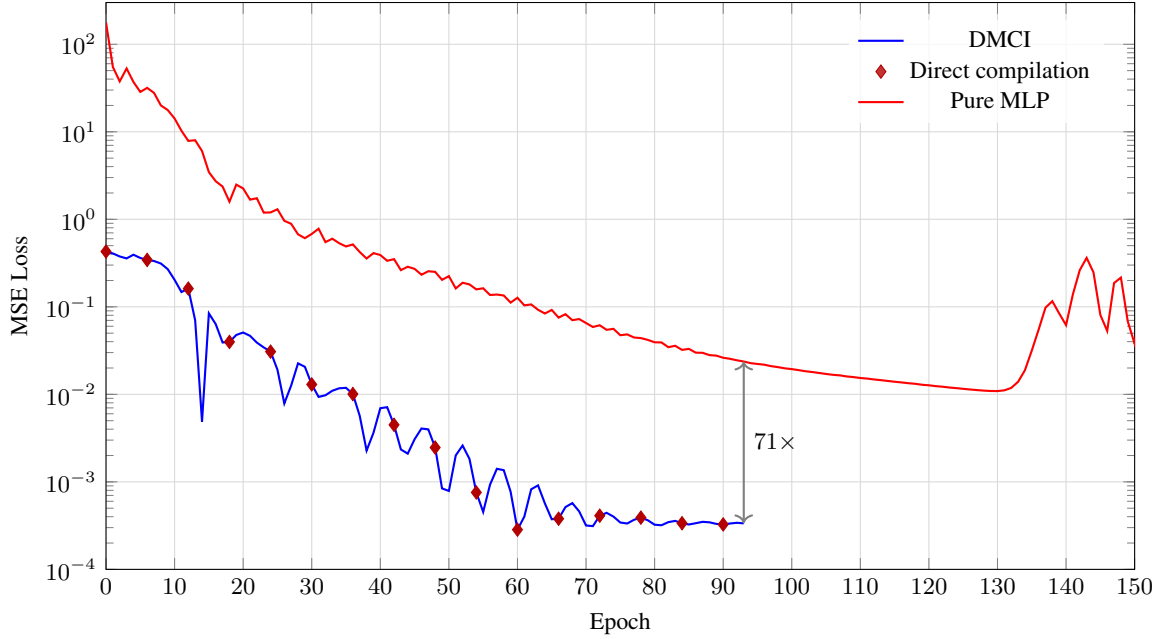

\paragraph{Recursive structure as inductive bias.} The pure MLP converges on every model except C01 (Lotka--Volterra), where it fails on all three seeds; elsewhere it reaches dramatically higher loss. The most striking case is C08 (cascaded EMA): DMCI achieves loss 0.0003 while the MLP reaches only 0.3699, a factor of $1102\times$ in mean final loss, though this mean is inflated by one post-convergence diverging MLP seed and collapses to ${\sim}2\times$ by median (see the Figure~\ref{fig:c08_loss} caption), so the robust effect is faster, more stable convergence rather than the mean-loss ratio. The cascaded filter's two-stage recursive structure imposes strong inductive bias that the MLP cannot match from data alone. Similarly, C07 (IIR filter) shows a $74\times$ mean gap (though this collapses to ${\sim}13\times$ by median, inflated by one diverging MLP seed) and C03 (decay chain) a $42\times$ gap. The exceptions are C05 (continued fraction) and C06 (damped pendulum), where the MLP achieves lower loss than DMCI; their smooth, monotonic outputs are well-suited to neural approximation. Figure~\ref{fig:c08_loss} shows the loss trajectory for C08, illustrating the convergence gap.

\paragraph{Why first-order compilation cannot express these models.}
Lotka--Volterra (C01) requires threading two coupled state variables through 20 recursive Euler integration steps. The damped pendulum (C06) has coupled angular position and velocity over 20 recursive ODE steps. Cascaded EMA (C08) threads two filter stages through 10 recursive steps. In each case, the model requires recursion and mutable state threading that cannot be expressed as flat arithmetic.

\paragraph{Wall-clock cost of deep recursion.} DMCI wall times range from ${\sim}250\,$s (C04, C05) to ${\sim}7$--$8$ hours for C01 and C06; the latter are driven by their ${\sim}1{,}400$--$1{,}800$ epochs to convergence rather than by recursion depth (C02 and C03 use more recursive steps yet finish in under an hour). Direct compilation runs the same models in 2--14\,s (most) to ${\sim}250\,$s (C01, C06), a per-run overhead of $74$--$158\times$.

% ============================================================================
\section{Experiment LIM-ENSO: Real-Data Case Study: Full Results and Extended Methodology}
\label{app:exp_lim_enso_ext}
% ============================================================================

This appendix gives the complete LIM-ENSO case study summarized in Section~\ref{sec:exp_c}: the full
model description, the gate-first protocol, the operator diagnostics, held-out forecast skill
(Table~\ref{tab:lim_enso_forecast}), the optimizer-portfolio and dimension-scaling sweeps
(Tables~\ref{tab:lim_enso_solver} and~\ref{tab:lim_enso_scaling}), the LLM-driven model-selection
comparison (Table~\ref{tab:lim_enso_select}), and all figures
(Figures~\ref{fig:lim_enso_eigenmode}--\ref{fig:lim_enso_forecast}), followed by the extended
methodology (data, binding contract, the Scheme Kalman-NLL program, and gate results).

% ============================================================================
% Flagship experiment: LIM/ENSO Kalman-MLE folded through the DMCI interpreter.
% Single-file paper convention: \input this from paper2.tex inside \section{Experiments}
% (after Experiment C, before Experiment H), and \input the three figure files
% (figures/fig_lim_enso_{eigenmode,forecast,recon}.tex) where the floats should land.
% Requires \providecommand{\TODO}{...} (defined locally in the figure files; redefined here
% so the placeholder tables render if this section is \input on its own).
% ============================================================================

\providecommand{\TODO}{\textcolor{red!75!black}{\textsc{[todo]}}}

% ----------------------------------------------------------------------------
\subsection{Experiment LIM-ENSO: A Real-Data Dynamical-Systems MLE as a Program}
\label{sec:exp_lim_enso}
% ----------------------------------------------------------------------------

Experiments~A--C recover known constants from data the same program generated. The natural
objection is that this is a synthetic setting: the targets are self-consistent, the noise is
controlled, and no observational dataset is ever fit. This experiment answers that objection
directly. We fit a real, multivariate, dynamical-systems maximum-likelihood problem, a Linear
Inverse Model (LIM) of the El~Ni\~no--Southern Oscillation (ENSO), with the \emph{entire}
Kalman-filter negative log-likelihood (NLL) folded through the compiled meta-circular
interpreter. The transition operator and noise covariances enter as ordinary parameter
tensors bound by name into the interpreter's environment; reverse-mode autograd through the
interpreter produces \emph{exact} gradients of the NLL with respect to them; and an optimizer
recovers the LIM by maximum likelihood. The recovered operator is then scored on held-out
ENSO forecast skill and put through the standard LIM operator diagnostics. The headline is a
\emph{capability}: a real-data dynamical-systems MLE, expressed as a program an LLM can emit,
calibrated by exact DMCI gradients, with zero per-structure transcription.

\paragraph{The model.} A LIM treats tropical sea-surface-temperature (SST) anomalies as a
linear-Gaussian state-space process~\citep{penland1995optimal,penland1996stochastic}. We build
the state from the leading empirical orthogonal functions (EOFs) of ERSSTv5~\citep{huang2017extended}
monthly SST anomalies over the tropical Indo-Pacific ($30^\circ$S--$30^\circ$N,
$30^\circ$E--$290^\circ$E, 1950--2024): climatology-removed, linearly detrended, 3-month
running-mean, area-weighted by $\sqrt{\cos\,\mathrm{lat}}$, then reduced to the leading $D$
principal-component (PC) time series by an economy SVD. The state $x_k\in\mathbb{R}^D$ is the
PC vector at month $k$; the LIM is the discrete-time linear-Gaussian model
\begin{equation}
\label{eq:lim_ssm}
x_{k+1} = F\,x_k + w_k, \quad w_k \sim \mathcal{N}(0,Q);
\qquad
y_k = x_k + v_k, \quad v_k \sim \mathcal{N}(0,R),
\end{equation}
with transition operator $F\in\mathbb{R}^{D\times D}$, process covariance $Q$, and
observation covariance $R$ (identity observation operator $H=I$, since the PCs \emph{are} the
observed state). Up to an additive constant and overall scale, the marginal NLL of the observed
PC window is the standard Kalman-filter
accumulation: predict $x^-=Fx$, $P^-=FPF^\top+Q$; innovation $e=y_k-x^-$ with covariance
$S=P^-+R$; gain $K=P^-S^{-1}$; update $x=x^-+Ke$, $P=(I-K)P^-$; and increment
$\mathrm{NLL}\mathrel{+}= \log\det S + e^\top S^{-1} e$ (this drops the per-step $\tfrac12$ factor and the
$\tfrac{D}{2}\log 2\pi$ constant of the Gaussian negative log-likelihood, a monotone affine
transform that leaves the $\arg\min$ over $(F,Q,R)$ unchanged). Fitting a LIM is exactly maximizing
this likelihood over $(F,Q,R)$, a genuine multivariate dynamical-systems MLE rather than a
curve fit.

\paragraph{The method: one program, every operator, exact gradients.} The Kalman NLL of
Eq.~\eqref{eq:lim_ssm} is written \emph{once} as a Scheme program (Appendix~\ref{app:exp_lim_enso_ext},
Listing~\ref{lst:lim_kalman}) and run through the compiled DMCI interpreter. $F$, $Q$, $R$, and
the observation window \texttt{obs} are bound as \texttt{as\_matrix} inputs in the
interpreter's environment; \texttt{(ref obs k)} gathers the $D$-vector of PCs at month $k$,
and the loop counter $k$ (data-independent, as the interpreter requires) drives a
tail-recursive \texttt{recur} that trampolines to an $O(T)$ walk. The covariance algebra uses
the interpreter's batched matrix surface (\texttt{matmul}, \texttt{matvec}, \texttt{transpose},
\texttt{inv}, \texttt{dot}), and the $\log\det S$ term uses a native \texttt{logdet} primitive
(a \texttt{slogdet}-based log-determinant). The \texttt{logdet} primitive matters numerically:
$\det S$ is a product of $D$ sub-unit eigenvalues and underflows ($\sim\!10^{-20}$ at $D=20$)
into the \texttt{log} clamp, corrupting the log-det term for $D\ge 15$; \texttt{logdet} sums
$\log$ of the LU pivots and stays exact (relative error ${\sim}10^{-8}$ at $D=20$). Because the
program is fixed and only the bound tensors change, the \emph{same} compiled artifact serves
every operator: $F$ is supplied as data, so swapping a dense $F$ for a diagonal, low-rank, or
companion-structured $F$ changes only the bound factor tensors and the in-program combine
algebra, never the compiled interpreter and never a line of hand-written gradient code. This is
the program-as-data thesis made concrete on a real scientific model: the LLM proposes the model
\emph{as code}, and DMCI calibrates it with exact gradients and zero per-structure transcription.

\paragraph{Gate-first protocol and correctness oracles.} DMCI is the gradient engine, not a
faster solver, so we never frame this experiment as beating an existing toolkit. We hold two
\emph{correctness oracles}: a NumPy twin whose arithmetic and accumulation order mirror the
Scheme program term-for-term (run in float32 for tight forward parity and float64 for accurate
finite differences), and \texttt{dynamax} as an optional cross-check on the marginal
log-likelihood (compared after matching this objective's additive constant and scale). The MLE runs only behind a numerical go/no-go gate that re-runs, at the real
$(D,T)$, the checks that license the fit: a finite NLL over the whole window (G1), a
positive-definite innovation covariance at every step (G2), forward parity of the DMCI NLL
against the float32 twin (G3), agreement of the DMCI autograd gradient with a central finite
difference on the float64 twin (G4), and a det-underflow margin above the float32 denormal
floor (G5). The gate returns \textbf{GO at $D=6,10,15,20$} on the real PC window: the DMCI NLL
matches the float32 twin within the gate's $2\times 10^{-3}$ relative tolerance, the
autograd-vs-finite-difference relative gradient error is $\approx 5\times 10^{-5}$ (well within
the $3\times 10^{-2}$ tolerance), and the innovation covariance stays positive-definite at every
step, with the worst-step $\det S$ falling from $\approx 1.0\times 10^{-4}$ at $D=6$ to
$\approx 4.1\times 10^{-12}$ at $D=20$, all above the $1.2\times 10^{-38}$ float32 denormal floor.
Only on GO does the fit run.

\paragraph{Scientific reference operator (never an NLL competitor).} The scientific reference
forecaster is the classical Green-function LIM, $G(\tau)=C(\tau)\,C(0)^{-1}$, the
Penland--Sardeshmukh moment estimator~\citep{penland1995optimal,penland1996stochastic} fit
\emph{without} any likelihood from the lag-$\tau$ and lag-$0$ covariances of the train window.
It is the established LIM operator a climate scientist would write down, and it is the
benchmark for \emph{scientific} skill, never an NLL competitor and never a speed contest. The
entire scientific pay-off layer, held-out forecast skill, the physical-units Nino-3.4
reconstruction, and the operator eigen-diagnostics, is pure NumPy/SciPy on the fitted float
matrices; nothing there touches the interpreter. The interpreter's sole job was to make the
likelihood differentiable.

\paragraph{What the recovered operator says about ENSO.} On the real ERSSTv5 train window, the
LIM operator carries the canonical ENSO signature at every $D$. The least-damped
\emph{oscillatory} eigenmode of the continuous generator $L=\log(F)/\Delta t$ falls inside the
ENSO band, with oscillation period ${\sim}2.3$--$3.3$\,yr and $e$-folding decay
${\sim}14$--$24$\,mo (Figure~\ref{fig:lim_enso_eigenmode}); we report this in-band oscillatory
mode deliberately, since the globally least-damped eigenmode is a near-stationary decadal mode
that is not ENSO. The selection rule is fixed in advance and stated explicitly: among the
oscillatory, decaying eigenmodes of $L$ (finite period, $|\mathrm{eig}\,F|<1$) whose period lies
in 2--7\,yr \emph{and} whose $e$-folding decay lies in 4--30\,mo, we report the least-damped one
(largest decay timescale); if none falls in-band we fall back to the least-damped oscillatory
mode overall, flagged out-of-band. This is a principled filter, not a post-hoc pick: it returns an
in-band mode at every $D\in\{6,10,15,20\}$, so the fallback is never exercised here. The physical-units reconstruction is faithful: un-normalizing the PC state,
projecting through the area-weighted EOFs, un-weighting by $\sqrt{\cos\,\mathrm{lat}}$, and
area-averaging over the Nino-3.4 box yields an index that tracks the leading PC at correlation
$r=0.97$ over the full record (Figure~\ref{fig:lim_enso_recon}); since that index is built from
the full $D$-mode state in physical space rather than from PC1 alone, the agreement shows that
the leading mode of the state space closely corresponds to the canonical ENSO index. The dynamics also pass the data
$\tau$-test, the Penland--Sardeshmukh linearity/Markov check that the lag-$2\tau$ Green
function matches the squared lag-$\tau$ Green function: the relative residual is $0.107$ at the
headline $D=10$ (and stays in the linearly-consistent regime across $D$), the non-trivial
data-moment test rather than the trivial operator-level $F^2$ versus $F\!\cdot\!F$ identity.

% Figure: LIM-ENSO recovered ENSO eigenmode (period / decay) vs state dimension D.
% Two panels: (left) oscillation period; (right) e-folding decay timescale.
% The shaded band marks the canonical ENSO window; the mode stays in-band at every D.
% Usage: \input{figures/fig_lim_enso_eigenmode.tex}
% Data: figures/lim_enso_eigenmode.dat (agg/T4_scaling_by_D.csv + Penland-Sardeshmukh G(1))
%   green_* columns are REAL NOW (Penland-Sardeshmukh G(1) on the real ERSSTv5 PCs);
%   fitted_* columns are REAL NOW (DMCI-MLE-recovered F, dmci_adam seed-mean).
% FORMATTING NOTES (pgfplots 1.12 compatibility):
%   - The ENSO band is drawn as a plain filled polygon (\closedcycle), NOT via the
%     `fill between' library: fill-between inside a 2nd groupplot cell corrupts cell
%     advancement in pgfplots 1.12 and stacks both panels on top of each other.
%   - The shared legend lives in panel 1 but is positioned BELOW and centered across
%     both panels: with panel width W and horizontal sep S, the group centre sits at
%     x-fraction 1+S/(2W) of panel 1 (= 1.159 for W=6.3cm, S=2.0cm).
\providecommand{\TODO}{\textcolor{red!75!black}{\textsc{[todo]}}}
\definecolor{darkgreen}{RGB}{0,135,60}
\begin{figure}[tbp]
\centering
\begin{tikzpicture}
\begin{groupplot}[
    group style={group size=2 by 1, horizontal sep=2.0cm},
    width=6.3cm, height=5.0cm, scale only axis,
    xlabel={State dimension $D$}, xmin=4, xmax=22, xtick={6,10,15,20},
    tick label style={font=\scriptsize}, label style={font=\small},
    title style={font=\small}, grid=major, grid style={gray!18},
]
% ---- panel 1: oscillation period (years); ENSO band 2--7 yr ----
% The shared legend is attached here, centered under BOTH panels (group centre = x-frac 1.159).
\nextgroupplot[title={ENSO-mode period}, ylabel={Period (yr)},
    ymin=1.5, ymax=7.5, ytick={2,3,4,5,6,7},
    legend style={font=\scriptsize, draw=gray!40, fill=white,
                  legend columns=3, column sep=10pt,
                  at={(1.159,-0.30)}, anchor=north},
    legend cell align={left}]
% canonical ENSO period band 2--7 yr (filled polygon; see header note)
\addplot[draw=none, fill=blue!8, forget plot] coordinates {(4,2)(22,2)(22,7)(4,7)} \closedcycle;
\addlegendimage{area legend, fill=blue!8}\addlegendentry{Canonical ENSO range}
% Green-function reference operator G(1) (real)
\addplot[mark=*, mark size=2.4pt, draw=darkgreen, thick, mark options={fill=darkgreen}]
    table[x=D, y expr=\thisrow{green_period}/12] {figures/lim_enso_eigenmode.dat};
\addlegendentry{Green-function reference $G(1)$}
% Fitted-LIM operator (DMCI-MLE-recovered, real)
\addplot[mark=square*, mark size=2.4pt, draw=blue!70, thick, mark options={fill=blue!70}]
    table[x=D, y expr=\thisrow{fitted_period}/12] {figures/lim_enso_eigenmode.dat};
\addlegendentry{Fitted LIM (DMCI-MLE)}
% ---- panel 2: e-folding decay (months); ENSO band 4--30 mo ----
\nextgroupplot[title={ENSO-mode decay}, ylabel={$e$-folding (mo)},
    ymin=0, ymax=34, ytick={0,10,20,30}]
\addplot[draw=none, fill=blue!8, forget plot] coordinates {(4,4)(22,4)(22,30)(4,30)} \closedcycle;
\addplot[mark=*, mark size=2.4pt, draw=darkgreen, thick, mark options={fill=darkgreen}]
    table[x=D, y=green_decay] {figures/lim_enso_eigenmode.dat};
\addplot[mark=square*, mark size=2.4pt, draw=blue!70, thick, mark options={fill=blue!70}]
    table[x=D, y=fitted_decay] {figures/lim_enso_eigenmode.dat};
\end{groupplot}
\end{tikzpicture}
\caption{Physical-timescale diagnostic of the fitted LIM operator. For each state dimension $D$,
we extract the ENSO-band oscillatory eigenmode of the continuous-time generator
$L=\log(F)/\Delta t$ and report its oscillation period (left) and $e$-folding decay timescale
(right). The shaded region marks the canonical ENSO range: 2--7\,yr periods and 4--30\,mo decay
times. The DMCI-MLE fitted operator (filled squares, $\blacksquare$) remains inside this range for every $D$,
indicating that the likelihood-calibrated operator recovers a physically plausible ENSO-like
damped oscillation rather than only fitting forecast statistics. The fitted mode also tracks the
classical Penland--Sardeshmukh Green-function LIM reference (filled circles, $\bullet$)
$G(1)=C(1)C(0)^{-1}$, with slightly faster decay at larger $D$. At each $D$, the plotted mode is
the least-damped oscillatory mode satisfying the ENSO-band selection rule in
Section~\ref{sec:exp_lim_enso}; the globally least-damped mode is excluded because it is a
near-stationary decadal mode rather than ENSO.}
\label{fig:lim_enso_eigenmode}
\end{figure}
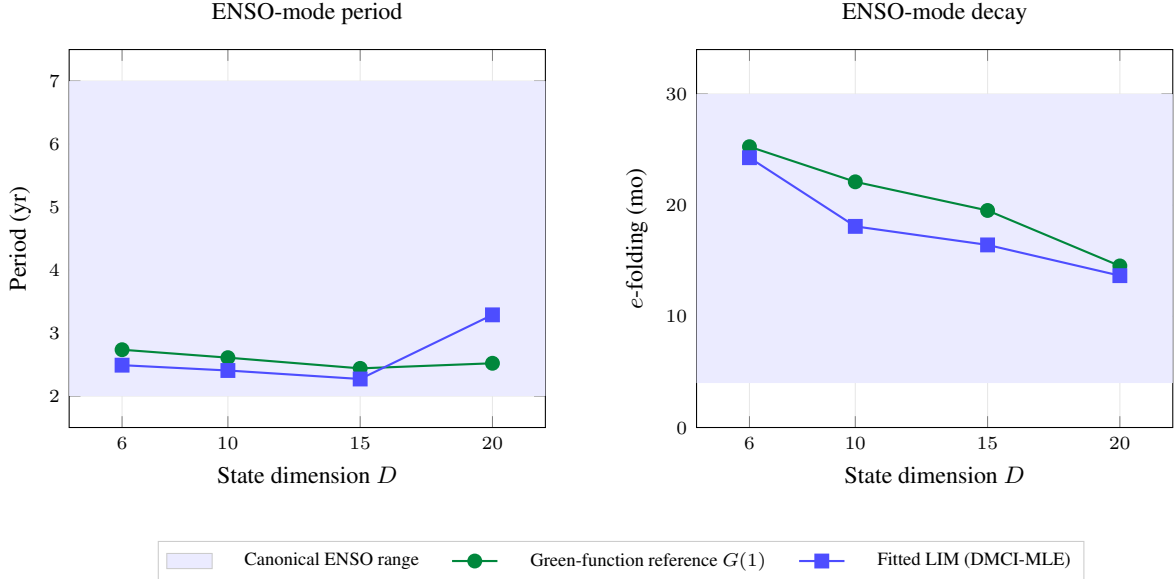

\paragraph{Held-out forecast skill.} On the held-out window, the fitted LIM is scored against
three references, persistence, damped persistence (per-PC lag-$h$ autocorrelation), and the
Green-function operator, by the anomaly correlation coefficient (ACC) and RMSE of the Nino-3.4
index at leads of 3, 6, 9, and 12 months (Table~\ref{tab:lim_enso_forecast},
Figure~\ref{fig:lim_enso_forecast}). The fitted LIM, which maximizes the same NLL the
interpreter differentiates, beats both persistence and damped persistence at \emph{every} lead
and is comparable to the classical Green-function reference, matching it at short leads and
edging it at the 9- and 12-month leads ($0.51$ vs $0.48$ ACC at 12 months). At $D=10$ it attains
$0.90$/$0.65$/$0.50$/$0.51$ Nino-3.4 ACC at the
$3$/$6$/$9$/$12$-month leads, against $0.86$/$0.54$/$0.20$/$-0.08$ for persistence; the skill
gap over the persistence baselines widens with lead, the standard LIM signature. The
Green-function operator likewise beats both persistence baselines at every lead, the scientific
reference curve. We select the optimizer-portfolio winner on held-out forecast skill (Nino-3.4
ACC at the headline lead), never on train NLL: the lesson carried from the composite-calibration
study (Appendix~\ref{app:exp_i}) is that converged is not recovered, so the model that
\emph{forecasts} best wins, not the one that drives the training likelihood lowest.

\begin{table}[t]
\centering
\caption{Held-out Nino-3.4 forecast skill for the DMCI-calibrated LIM. At $D=10$, the dense
fitted LIM operator is evaluated on held-out ERSSTv5 data at 3-, 6-, 9-, and 12-month forecast
leads. The DMCI-MLE fitted LIM outperforms persistence and damped persistence at every lead and
remains comparable to the classical Green-function LIM reference, showing that the operator
recovered by differentiating the Kalman likelihood through the compiled interpreter has real
predictive skill rather than merely fitting the training objective. ACC is the anomaly
correlation coefficient, where higher is better; RMSE is lower better. Best values in each column
are bolded.}
\label{tab:lim_enso_forecast}
\small
\begin{tabular}{@{}lcccccccc@{}}
\toprule
 & \multicolumn{4}{c}{Nino-3.4 ACC $\uparrow$} & \multicolumn{4}{c}{RMSE $\downarrow$} \\
\cmidrule(lr){2-5}\cmidrule(lr){6-9}
Method & $3$mo & $6$mo & $9$mo & $12$mo & $3$mo & $6$mo & $9$mo & $12$mo \\
\midrule
Fitted LIM (DMCI-MLE)        & $0.90$ & $0.65$ & $\mathbf{0.50}$ & $\mathbf{0.51}$ & $0.39$ & $0.69$ & $\mathbf{0.79}$ & $\mathbf{0.77}$ \\
Persistence                  & $0.86$ & $0.54$ & $0.20$ & $-0.08$ & $0.45$ & $0.80$ & $1.07$ & $1.24$ \\
Damped persistence           & $0.85$ & $0.52$ & $0.17$ & $0.22$ & $0.44$ & $0.72$ & $0.84$ & $0.86$ \\
\addlinespace
Green-function $G(\tau)$ (reference) & $\mathbf{0.91}$ & $\mathbf{0.66}$ & $0.49$ & $0.48$ & $\mathbf{0.38}$ & $\mathbf{0.67}$ & $\mathbf{0.79}$ & $0.78$ \\
\bottomrule
\end{tabular}
\end{table}

% Figure: LIM-ENSO Nino-3.4 reconstruction timeseries (truth index vs leading PC).
% The reconstructed Nino-3.4 index tracks PC1 at r=0.97 over the full record.
% Usage: \input{figures/fig_lim_enso_recon.tex}
% Data: figures/lim_enso_nino34_recon.dat (make_exp_lim_enso_data.py).
%   The reconstructed index is built from the FULL D-mode state in physical (SST) space
%   (reconstruct_field sums all D PCs through the area-weighted EOFs), NOT from PC1 alone;
%   the r=0.97 agreement with PC1 is therefore evidence that the leading mode carries ENSO.
\definecolor{brightred}{RGB}{227,26,28}
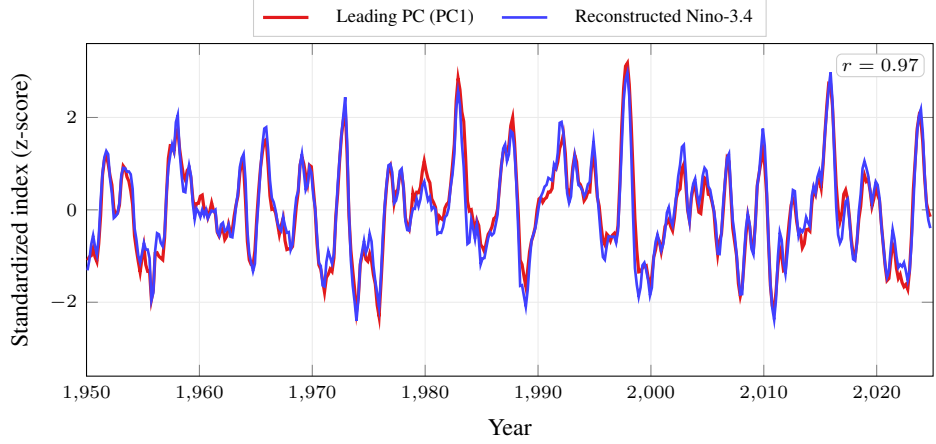
\begin{figure}[tbp]
\centering
\begin{tikzpicture}
\begin{axis}[
    width=11.2cm, height=4.4cm, scale only axis,
    xlabel={Year}, ylabel={Standardized index (z-score)},
    xmin=1950, xmax=2025, xtick={1950,1960,1970,1980,1990,2000,2010,2020},
    ymin=-3.6, ymax=3.6,
    tick label style={font=\scriptsize}, label style={font=\small},
    grid=major, grid style={gray!16},
    legend style={font=\scriptsize, draw=gray!40, fill=white,
                  at={(0.5,1.02)}, anchor=south, legend columns=2, column sep=8pt},
    legend cell align={left},
]
% leading PC (PC1) reference: bright red, a little thicker so it reads against the blue
\addplot[draw=brightred, line width=1.2pt]
    table[x=year, y=pc1] {figures/lim_enso_nino34_recon.dat};
\addlegendentry{Leading PC (PC1)}
% reconstructed Nino-3.4 index (blue, drawn on top)
\addplot[draw=blue!75, line width=1.0pt]
    table[x=year, y=nino34] {figures/lim_enso_nino34_recon.dat};
\addlegendentry{Reconstructed Nino-3.4}
\node[font=\scriptsize, anchor=north east, draw=gray!40, fill=white, inner sep=2pt,
      rounded corners=2pt] at (rel axis cs:0.99,0.97) {$r = 0.97$};
\end{axis}
\end{tikzpicture}
\caption{Reconstructed Nino-3.4 SST-anomaly index (blue) against the leading PC (PC1, red)
over the full 1950--2024 ERSSTv5 record, both $z$-scored for overlay. The index is
reconstructed by un-normalising the PC state, projecting through the area-weighted EOFs,
un-weighting by $\sqrt{\cos\,\mathrm{lat}}$, and area-averaging over the Nino-3.4 box
($5^\circ$S--$5^\circ$N, $170^\circ$W--$120^\circ$W). Crucially, this index is built from the
\emph{full} $D$-mode state in physical (SST) space, not from PC1 alone, so the close agreement
($r=0.97$) shows that the leading state-space mode closely corresponds to the canonical ENSO
variability that Nino-3.4 measures, and that the physical-units reconstruction pipeline is
faithful. Because the reconstruction is performed directly from the fitted EOF/LIM state using
standard climate-analysis procedures, the close agreement provides an independent physical
validation of the model recovered through DMCI calibration.}
\label{fig:lim_enso_recon}
\end{figure}

% Figure: LIM-ENSO held-out Nino-3.4 forecast skill (ACC) vs lead.
% fitted-LIM vs persistence vs damped-persistence vs the Green-function reference.
% Usage: \input{figures/fig_lim_enso_forecast.tex}
% Data: figures/lim_enso_forecast_acc.dat (agg/T3_forecast_skill.csv)
%   all four columns REAL (S0 dense operator, D=10, 3 seeds): fitted_lim is the DMCI-MLE-recovered
%   operator forecast; persistence / damped / green are classical baselines on the held-out PCs.
\providecommand{\TODO}{\textcolor{red!75!black}{\textsc{[todo]}}}
\providecommand{\darkgoldcolor}{}\definecolor{darkgold}{RGB}{184,134,11}
\definecolor{darkgreen}{RGB}{0,135,60}
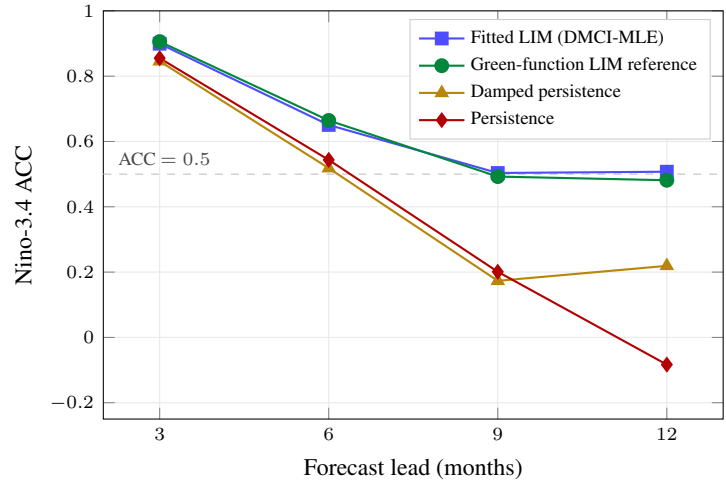
\begin{figure}[tbp]
\centering
\begin{tikzpicture}
\begin{axis}[
    width=8.2cm, height=5.4cm, scale only axis,
    xlabel={Forecast lead (months)}, ylabel={Nino-3.4 ACC},
    xmin=2, xmax=13, xtick={3,6,9,12},
    ymin=-0.25, ymax=1.0, ytick={-0.2,0,0.2,0.4,0.6,0.8,1.0},
    tick label style={font=\scriptsize}, label style={font=\small},
    grid=major, grid style={gray!18},
    legend style={font=\scriptsize, draw=gray!40, fill=white,
                  at={(0.98,0.98)}, anchor=north east},
    legend cell align={left},
]
% ACC=0.5 reference threshold -- kept subordinate to the data curves (thin, light, no legend)
\addplot[gray!50, dashed, forget plot] coordinates {(2,0.5)(13,0.5)};
\node[font=\scriptsize, gray!55!black, anchor=south west] at (axis cs:2.1,0.5) {ACC $=0.5$};
% Marker fills match each line colour (no black fills); shapes keep the curves grayscale-distinct.
% fitted-LIM (DMCI-MLE-recovered operator)
\addplot[mark=square*, mark size=2.2pt, draw=blue!70, thick, mark options={fill=blue!70}]
    table[x=lead, y=fitted_lim] {figures/lim_enso_forecast_acc.dat};
\addlegendentry{Fitted LIM (DMCI-MLE)}
% Green-function LIM reference (green)
\addplot[mark=*, mark size=2.4pt, draw=darkgreen, thick, mark options={fill=darkgreen}]
    table[x=lead, y=green] {figures/lim_enso_forecast_acc.dat};
\addlegendentry{Green-function LIM reference}
% damped persistence
\addplot[mark=triangle*, mark size=2.4pt, draw=darkgold, thick, mark options={fill=darkgold}]
    table[x=lead, y=damped] {figures/lim_enso_forecast_acc.dat};
\addlegendentry{Damped persistence}
% persistence
\addplot[mark=diamond*, mark size=2.4pt, draw=red!70!black, thick, mark options={fill=red!70!black}]
    table[x=lead, y=persistence] {figures/lim_enso_forecast_acc.dat};
\addlegendentry{Persistence}
\end{axis}
\end{tikzpicture}
\caption{Held-out Nino-3.4 forecast skill of the DMCI-calibrated LIM. The dense ($D=10$) LIM
operator learned by differentiating the Kalman-filter likelihood through the compiled interpreter
is evaluated on held-out ERSSTv5 data at 3-, 6-, 9-, and 12-month forecast leads. The fitted LIM
outperforms persistence and damped persistence at every lead and remains comparable to the
classical Green-function LIM reference. This matters for DMCI because it shows that exact gradients
through an interpreted program recover a dynamical operator with real predictive skill, not merely
a low training likelihood. The Green-function LIM is included as a scientific reference operator,
not as an NLL optimization competitor. ACC is the anomaly correlation coefficient; higher is
better.}
\label{fig:lim_enso_forecast}
\end{figure}

\paragraph{Capability and program-as-data.} One compiled artifact, the DMCI interpreter,
renders an \emph{arbitrary} runtime-emitted dynamical model differentiable. The LIM Kalman NLL
is the \emph{same} program across every $(D,\text{structure})$ cell; only the bound $F/Q/R/
\texttt{obs}$ tensors change. The optimizer portfolio we run over this objective,
exact-gradient Adam (primary), multi-start L-BFGS, and a batched differential-evolution
population that binds its whole $[N,D,D]$ population in one interpreter walk, all optimize the
\emph{identical} compiled objective; we compare \emph{optimizers}, not objectives. This is the
distinction from prior differentiable scientific modeling
(Section~\ref{sec:related})~\citep{aboelyazeed2023,jiang2025canveg}: those works hand-port each
model into an autodiff framework, whereas here a single compiled interpreter makes a
real, observational-data MLE, and any structural variant of it, differentiable with no
per-model engineering.

\paragraph{Optimizer comparison and dimension scaling.} The optimizer portfolio (exact-gradient
Adam, multi-start L-BFGS, batched differential evolution) over the shared DMCI NLL, and the full
$D\in\{6,10,15,20\}$ scaling of held-out skill, recovered-mode timescales, and operator
stability, are reported in Tables~\ref{tab:lim_enso_solver} and~\ref{tab:lim_enso_scaling}. The
finding mirrors Experiment~I: the \emph{exact DMCI gradients} are what make the MLE work. Both
gradient optimizers, Adam and multi-start L-BFGS, reach essentially the same optimum (L-BFGS
edges Adam on NLL at every $D$, being curvature-aware) and are $100\%$ seed-robust, whereas the
gradient-free batched differential evolution never optimizes (positive NLL at every $D$), is
$0\%$ robust, and erodes further with dimension. Crucially, this is a genuinely high-dimensional
estimation problem: the dense LIM estimates $k = D^2 + D(D{+}1)/2 + 1$ free parameters (a full
transition matrix $F$, a Cholesky-factored process covariance $Q$, and a scalar observation
covariance $R$), growing from $58$ at $D=6$ to $\mathbf{611}$ at $D=20$
(Table~\ref{tab:lim_enso_scaling}), one to two orders of magnitude beyond the few-parameter
recoveries of Experiments~A--C. That plain gradient descent on exact DMCI gradients fits a
$611$-parameter real-data MLE while gradient-free differential evolution fails outright, and
degrades as $D$ grows, is the regime where differentiating the interpreter most clearly earns
its cost. The point is not that one gradient optimizer beats the other; it is that exact DMCI
gradients turn this real-data MLE into a problem both gradient optimizers solve, where
gradient-free search fails.

\begin{table}[t]
\centering
\caption{Exact DMCI gradients make the real-data LIM maximum-likelihood fit tractable. All
optimizers are applied to the same compiled Kalman-filter NLL objective for the dense $S0$ LIM
operator; only the optimization method changes. The two exact-gradient methods that use gradients
propagated through the compiled interpreter, Adam and multi-start L-BFGS, converge reliably across
seeds, recover stable operators with $\rho(F)<1$, and achieve held-out Nino-3.4 forecast skill
comparable to the classical Green-function LIM reference. In contrast, the gradient-free batched
differential-evolution baseline fails to minimize the likelihood, is not seed-robust, and returns
an unstable operator. This comparison isolates the practical value of DMCI: it turns a
high-dimensional program-defined dynamical-systems likelihood into an optimization problem that
standard gradient methods can solve. Train NLL is reported across state dimensions $D$; ACC@3,
spectral radius $\rho(F)$, stability, and robustness are reported at the headline $D=10$. Lower
NLL is better, higher ACC is better, and $\rho(F)<1$ indicates a stable fitted operator.}
\label{tab:lim_enso_solver}
\small
\begin{tabular}{@{}lcccccccc@{}}
\toprule
 & \multicolumn{4}{c}{Train NLL $\downarrow$ (by $D$)} & \multicolumn{4}{c}{At headline $D=10$} \\
\cmidrule(lr){2-5}\cmidrule(lr){6-9}
Solver & $D{=}6$ & $10$ & $15$ & $20$ & ACC@3 $\uparrow$ & $\rho(F)$ & Stable & Robust \\
\midrule
Exact-grad Adam (DMCI)      & $-2060$ & $-2814$ & $-3676$ & $-4494$ & $0.90$ & $0.96$ & Yes & $3/3$ \\
Multi-start L-BFGS (DMCI)   & $-2065$ & $-2836$ & $-3706$ & $-4551$ & $0.90$ & $0.96$ & Yes & $3/3$ \\
Differential evolution (batched) & $+636$ & $+3878$ & $+7876$ & $+11383$ & $0.87$ & $1.08$ & No & $0/3$ \\
\addlinespace
Green-function $G(\tau)$ (reference) & n/a & n/a & n/a & n/a & $0.91$ & $0.96$ & Yes & n/a \\
\bottomrule
\end{tabular}
\end{table}

\begin{table}[t]
\centering
\caption{Dimension scaling of the DMCI-calibrated LIM. The dense fitted LIM is trained at state
dimensions $D\in\{6,10,15,20\}$, increasing the number of learned parameters from 58 to 611 while
using the same DMCI-interpreted Kalman-filter likelihood. Exact gradients through the compiled
interpreter allow the MLE to remain numerically stable at every dimension: all fits pass the gate
checks, the recovered operators remain stable, and the extracted ENSO-band mode stays within
physically plausible period and decay ranges. Held-out 6-month Nino-3.4 skill peaks at $D=15$,
even though training NLL continues to improve with $D$, showing that model selection is based on
forecast generalization rather than simply pushing dimensionality upward. Lower NLL is better;
higher ACC is better; $\rho(F)<1$ indicates stability.}
\label{tab:lim_enso_scaling}
\small
\begin{tabular}{@{}lcccccccc@{}}
\toprule
$D$ & $k$ (params) & Gate & Train NLL $\downarrow$ & ACC@6 $\uparrow$ & ENSO period (yr) & ENSO decay (mo) & $\min\det S$ & $\mathrm{cond}\,S$ \\
\midrule
6  & $58$  & GO & $-2060$ & $0.62$ & $2.49$ & $24.3$ & $1.0{\times}10^{-4}$ & $1.97$ \\
10 & $156$ & GO & $-2814$ & $0.65$ & $2.41$ & $18.1$ & $7.7{\times}10^{-7}$ & $2.64$ \\
15 & $346$ & GO & $-3676$ & $\mathbf{0.73}$ & $2.27$ & $16.4$ & $1.9{\times}10^{-9}$ & $3.09$ \\
20 & $611$ & GO & $-4494$ & $0.61$ & $3.29$ & $13.7$ & $4.1{\times}10^{-12}$ & $3.62$ \\
\bottomrule
\end{tabular}
\end{table}

\paragraph{Model selection over LLM-proposed structures.} The experiments above all fit a single
structural form of the transition operator, a dense $F$. But a LIM's $F$ admits many structural
hypotheses, each a distinct claim about how ENSO evolves: a diagonal $F$ (decoupled modes), a
low-rank-plus-diagonal $F$ (a few coupled teleconnections on a damped background), an
AR(2)-companion $F$ (an oscillator with explicit memory), or a symmetric-plus-antisymmetric split
(the normal versus non-normal decomposition central to LIM theory). This is where the
program-as-data thesis becomes a \emph{method}, not just an implementation convenience. We give
this menu of structures to a mid-size open LLM (Qwen3.6-35B) served on a local OpenAI-compatible
endpoint, which emits each hypothesis as a one-to-three line $F$-assembly edit to the \emph{same}
17-line Kalman-NLL program (Appendix~\ref{app:exp_lim_enso_ext}): between a dense, a diagonal, an
AR(2), and a normal/non-normal $F$, the only thing that changes is the prelude that builds $F$
from bound factor tensors. Calibration is then two-level: the LLM proposes the discrete
\emph{structure} as code, and exact-gradient descent (the gradient method established above) fits
the \emph{continuous} parameters inside it. Two properties make this a real capability rather than
a templating trick. First, the LLM is a one-shot code generator that never sees the data or the
fit; a static unsupported-operator prescan, a compile check, and a parity check against the NumPy
twin gate every emitted program before any optimization runs (here all five structures were
accepted on the first generation, with forward parity to the twin at the $10^{-6}$ level or
better). Second, everything beneath the structure, the compiled interpreter, the autograd path,
and the MLE driver, is byte-identical across all five hypotheses; no code is recompiled and no
gradient is re-derived by hand between them.

\paragraph{Selection is an Occam tradeoff, not raw fit.} Ranking these structures by training
likelihood alone would be meaningless: the dense $F$ has the most free parameters and can always
match or exceed a constrained one in-sample, and the dimension sweep already showed that excess
capacity \emph{hurts}, with held-out skill peaking at $D=15$ and falling at $D=20$. Selection is
therefore a complexity tradeoff. Table~\ref{tab:lim_enso_select} calibrates each LLM-proposed
structure at $D=15$ (the dimension of peak held-out skill) with the same frozen interpreter and
scores it by train NLL,
by AIC and BIC (which charge for the per-structure parameter count, $123$ for the AR(2)-companion
up to $346$ for the dense and normal/non-normal forms), and by held-out Nino-3.4 forecast skill,
the out-of-sample criterion that actually matters. The three criteria disagree informatively.
By held-out skill and by AIC, the expressive full-operator structures win: the
symmetric-plus-antisymmetric split (S5), the LIM-theoretic normal/non-normal decomposition of a
full $F$, edges the dense operator on both. By BIC, whose heavier penalty rewards parsimony, the
low-rank-plus-diagonal structure (S3) wins outright, capturing most of the forecast skill ($0.66$
versus $0.73$) with $196$ parameters against $346$. The two restrictive extremes fail honestly:
the diagonal operator (S1) forecasts worst by discarding the teleconnections that couple ENSO
modes, and the two-parameter AR(2)-companion (S4) cannot fit a fifteen-dimensional state at all
(positive NLL). The data therefore prefer an expressive operator for raw fit and skill while a
low-rank operator is the most defensible under strict parsimony, the Occam tension this comparison
is built to expose, and every one of these conclusions came from editing a one-line prelude and
re-running the same engine. The final column, ``same DMCI engine?'',
reads \emph{yes} on every row. This closes the loop the paper opened: an LLM proposes a space of
dynamical models \emph{as code}, one differentiable interpreter calibrates and ranks them by exact
gradients and predictive skill, and improving the model becomes a one-line edit to a 17-line
program instead of a per-hypothesis rebuild of the fitting machinery.

\begin{table}[t]
\centering
\caption{Model selection over LLM-proposed transition structures using one fixed DMCI engine. Five
alternative structures for the LIM transition operator $F$ are evaluated at $D=15$, with three
optimization seeds per structure. Each structure is calibrated by the same compiled DMCI interpreter,
the same Kalman-filter likelihood program, and the same maximum-likelihood driver; only the small
program fragment that assembles $F$ from its bound parameters changes. This is the key DMCI capability
illustrated here: structural hypotheses can be proposed as code, fit with exact gradients, and compared
without recompiling the interpreter or re-deriving gradients by hand. $k_F$ is the number of parameters
in $F$, and $k$ is the total number of fitted parameters including $Q$ and $R$. Lower train NLL, AIC,
and BIC are better; higher held-out Nino-3.4 ACC is better. AIC and held-out skill favor the more
expressive full-operator structures, while BIC favors the lower-rank-plus-diagonal structure, exposing
the intended tradeoff between fit, parsimony, and forecast skill.}
\label{tab:lim_enso_select}
\small
\begin{tabular}{@{}lccccccc@{}}
\toprule
Structure & $k_F$ & $k$ & Train NLL $\downarrow$ & AIC $\downarrow$ & BIC $\downarrow$ & Held-out ACC $\uparrow$ & Same DMCI engine? \\
\midrule
S0 dense              & $225$ & $346$ & $-3676$ & $-6659$ & $-5315$ & $0.73$ & yes \\
S1 diagonal           & $15$  & $136$ & $-3036$ & $-5799$ & $-5271$ & $0.52$ & yes \\
S3 low-rank$+$diag    & $75$  & $196$ & $-3434$ & $-6477$ & $\mathbf{-5715}$ & $0.66$ & yes \\
S4 AR(2)-companion    & $2$   & $123$ & $+6439$ & $+13124$ & $+13602$ & $0.57$ & yes \\
S5 sym/antisym        & $225$ & $346$ & $-3695$ & $\mathbf{-6699}$ & $-5354$ & $\mathbf{0.73}$ & yes \\
\bottomrule
\end{tabular}
\end{table}

\subsection{Extended Methodology}

\paragraph{Data.} We use ERSSTv5~\citep{huang2017extended} monthly-mean SST, masked to the
tropical Indo-Pacific ($30^\circ$S--$30^\circ$N, $30^\circ$E--$290^\circ$E), 1950--2024
($T=898$ months, $3{,}373$ ocean grid points). Preprocessing removes the monthly climatology,
linearly detrends, applies a 3-month running mean, area-weights each column by
$\sqrt{\cos\,\mathrm{lat}}$, and takes an economy SVD of the time-by-space weighted-anomaly
matrix to produce $D_{\max}=20$ unit-variance PC time series (the EOFs are the right singular
vectors in weighted space; \texttt{pc\_std} stores the un-normalization factors). A run at
dimension $D$ slices the leading $D$ PCs (nested basis), fits on the first $T_{\text{train}}=360$
months (Jan~1950--Dec~1979), and evaluates held-out skill on the next $T_{\text{test}}=120$ months
(Jan~1980--Dec~1989); the split is fixed in advance and identical across all $D$ and all solvers.
The leading PC explains $42\%$ of the tropical SST-anomaly variance and is, by sign convention,
warm-positive over the Nino box.

\paragraph{Baseline fairness.} Three properties keep the comparison honest. (i)~The Green-function
reference $G(\tau)=C(\tau)C(0)^{-1}$ is fit from the \emph{same} train-window lag-$\tau$ and lag-$0$
covariances of the \emph{same} PC state used by the MLE, so it sees identical data in an identical
representation; it is a scientific skill reference, never an NLL competitor. (ii)~All three optimizers
(exact-gradient Adam, multi-start L-BFGS, and gradient-free differential evolution) minimize the
\emph{identical} compiled Kalman-NLL objective; only the optimizer changes, so the comparison isolates
the optimizer, not the model or the implementation. The differential-evolution population is bound as a
single $[N,D,D]$ tensor and evaluated in one batched interpreter walk, the same machinery the gradient
methods use for the forward pass. (iii)~Held-out skill is the $h$-step Nino-3.4 anomaly correlation and
RMSE of a Kalman analysis warm-started through the held-out window with the fitted $(F,Q,R)$ and
forecast by $\hat{x}_{t+h}=F^h\hat{x}_t$, scored against persistence, damped persistence, and the
Green-function reference. The matched-compute scaling of exact-gradient versus gradient-free
optimization across dimension is studied separately and under an explicit equal-budget protocol in
Experiment~I (Appendix~\ref{app:exp_i}).

\paragraph{Binding contract.} \texttt{pcs.npy} has shape $[T,D_{\max}]$ with rows indexed by
month and columns by PC. The observation window is bound as \texttt{obs = as\_matrix(pcs[:T, :D])};
inside the Scheme program \texttt{(ref obs k)} gathers \emph{row} $k$, the $D$-vector of PCs at
month $k$, with $k$ the data-independent loop counter the interpreter requires. $F$, $Q$, $R$
are bound as \texttt{as\_matrix} inputs; $Q$ is Cholesky-parametrized as $L_qL_q^\top + q_{\text{floor}}I$
($q_{\text{floor}}=10^{-4}$) and $R=\mathrm{softplus}(r_{\text{raw}})I + r_{\text{floor}}I$ with
$r_{\text{floor}}=0.1$ kept physical as the primary float32 PD lever (never a fudge factor). The
estimated parameters are therefore the $D^2$ entries of $F$, the $D(D{+}1)/2$ free entries of the
Cholesky factor $L_q$ (for $Q$), and the single $r_{\text{raw}}$ (for $R$): $k = D^2 +
D(D{+}1)/2 + 1$ in all, which is $58$, $156$, $346$, and $611$ free parameters at $D=6,10,15,20$.

\paragraph{The Scheme Kalman-NLL program (program-as-data).} The complete program is the
following source string, emitted once per $(D,T)$ and reused across all seeds and solvers with
only the bound tensors changing (shown at $D=10$, $T=360$; the $\log\det S$ term uses the native
\texttt{logdet} primitive):

\begin{lstlisting}[language=Lisp]
(loop ((k 0)
       (x  (zeros 10))
       (P  (eye 10))
       (L  0.0))
  (if (= k 360)
      L
      (let* ((xpred (matvec F x))
             (Ppred (+ (matmul (matmul F P) (transpose F)) Q))
             (y     (ref obs k))
             (e     (- y xpred))
             (S     (+ Ppred R))
             (Sinv  (inv S))
             (Kg    (matmul Ppred Sinv))
             (xnew  (+ xpred (matvec Kg e)))
             (Pnew  (matmul (- (eye 10) Kg) Ppred))
             (nll   (+ (logdet S) (dot e (matvec Sinv e)))))
        (recur (+ k 1) xnew Pnew (+ L nll)))))
\end{lstlisting}
\noindent The tail-recursive \texttt{recur} trampolines to an $O(T)$ walk through the compiled
interpreter; the structural integers ($k$, the dimension literals) and the dispatch are
data-independent, so gradients flow only through the bound $F/Q/R$ payload tensors, exactly the
program-determined-control-flow premise of Section~\ref{sec:method}. The NumPy reference twin
(\texttt{reference.py}) mirrors this arithmetic and accumulation order term-for-term, evaluating
\texttt{(matmul (matmul F P) (transpose F))} as $(FP)F^\top$ and the Mahalanobis term as
$e^\top(S^{-1}e)$ so that float32 forward parity (gate G3) and the float64 finite-difference
gradient (gate G4) are tight.
\label{lst:lim_kalman}

\paragraph{Structural variants are one-line edits.} The Kalman-NLL body above is 17 lines, and
it is reused \emph{verbatim} for every transition-operator structure. A different dynamical
hypothesis about $F$ changes only how $F$ is assembled from bound factor tensors: a one-to-three
line \texttt{let*} prelude spliced ahead of the per-step bindings, never the compiled
interpreter and never a line of gradient code. The five structures (the menu an LLM proposes as
code) differ \emph{only} in this prelude, shown here at $D=10$; the factor symbols
\texttt{Dvec}, \texttt{U}, \texttt{V}, \texttt{arow}, \texttt{e0}, \texttt{Sub}, \texttt{M} are
bound as \texttt{as\_matrix}/\texttt{as\_vector} inputs, decoded from the raw parameters by the
same combine-algebra the program runs, so the in-program $F$ matches its float twin bit-for-bit:

\begin{lstlisting}[language=Lisp]
;; S0 dense        : F bound directly (no prelude)               k_F = D^2     = 100
;; S1 diagonal     : (F (* (eye 10) Dvec))                       k_F = D       = 10
;; S3 lowrank+diag : (F (+ (* (eye 10) Dvec)
;;                         (matmul U (transpose V))))            k_F = D + 2Dr = 50
;; S4 AR(2)-compan.: (F (+ (outer e0 arow) Sub))                 k_F = 2
;; S5 sym/antisym  : (Ssym  (scale 0.5 (+ M (transpose M))))
;;                   (Aanti (scale 0.5 (- M (transpose M))))
;;                   (F     (+ Ssym Aanti))                      k_F = D^2     = 100
\end{lstlisting}
\noindent The free-parameter count $k_F$ differs by structure (above), so AIC/BIC genuinely
separate the candidates; and because the interpreter, the autograd path, and the MLE driver are
byte-identical across all five, a model-selection comparison isolates the structural hypothesis,
not the implementation. This is the program-as-data thesis at its most literal: a structural
hypothesis is a one-line edit to a 17-line program, and the same compiled artifact calibrates
every edit with exact gradients.

\paragraph{Gate results.} At the real $(D,T)$ the gate returns GO at $D=6,10,15,20$. Forward
parity against the float32 twin holds within the gate's $2\times 10^{-3}$ relative tolerance at
every $D$; the autograd-vs-central-finite-difference relative gradient error is
$\approx 5\times 10^{-5}$ across $D$ (well within the $3\times 10^{-2}$ tolerance); and the
innovation covariance stays PD at every step, worst-step $\det S$ on this raw gate run falling
from $\approx 1\times 10^{-5}$ ($D=6$) to $\approx 2.6\times 10^{-17}$ ($D=20$), all above the
$1.2\times 10^{-38}$ float32 denormal floor (the fitted-operator dimension sweep of
Table~\ref{tab:lim_enso_scaling} reports correspondingly larger per-fit minima, $1.0\times 10^{-4}$
to $4.1\times 10^{-12}$).

\paragraph{Scientific diagnostics (computed on the fitted/Green-function operators, not the
interpreter).} The forecast and diagnostic layer (\texttt{forecast.py}) is pure NumPy/SciPy:
it warm-starts a Kalman analysis through the held-out window with the fitted $(F,Q,R)$, makes
$h$-step forecasts $\hat{x}_{t+h}=F^h\hat{x}_t$, reconstructs the physical SST-anomaly field and
the Nino-3.4 index, and scores ACC/RMSE against persistence, damped persistence, and the
Green-function reference $G(\tau)=C(\tau)C(0)^{-1}$. The operator diagnostics take the
continuous generator $L=\log(F)/\Delta t$ and report per-mode decay timescales
$-1/\mathrm{Re}(\lambda)$ and oscillation periods $2\pi/|\mathrm{Im}(\lambda)|$; the reported
ENSO mode is the least-damped \emph{oscillatory} mode in the band (period 2--7\,yr, decay
4--30\,mo), excluding the near-stationary decadal mode that is not ENSO.

% ============================================================================
\section{Experiment D: Structural Search and the Cost of Interpretation}
\label{app:exp_d}
% ============================================================================

Experiments A--C hold the program structure fixed. A natural question is whether DMCI is practical when program structure also varies, e.g., in a genetic programming (GP) loop.

\subsection{Design}

GP-based symbolic regression targeting $f(x) = 2\sin(3x) + 0.5x^2$ on 20 points. Population of 50 expression trees over $\{+, -, *, /, \sin, \cos\}$. Each candidate's constants are optimized for 20 Adam epochs (lr $= 0.05$). Two methods: GP + direct-recompile vs.\ GP + DMCI. Five seeds, 1{,}000 candidates each (5{,}000 total).

\subsection{Results}

Table~\ref{tab:exp_d_timing} decomposes the per-candidate compile-and-train cost of direct compilation versus DMCI.

\begin{table}[ht]
\centering
\caption{Experiment~D: per-candidate cost decomposition.}
\label{tab:exp_d_timing}
\small
\begin{tabular}{@{}lrrrrr@{}}
\toprule
Method & $t_{\text{compile}}$ & $t_{\text{train}}$ & $t_{\text{total}}$ & \% Compile & Ratio \\
\midrule
Direct & 4.1\,ms & 148\,ms & 152\,ms & 2.7\% & $1\times$ \\
DMCI & 22.9\,ms & 3{,}824\,ms & 3{,}846\,ms & 0.6\% & $25.3\times$ \\
\bottomrule
\end{tabular}
\end{table}

No wall-clock crossover occurs (Table~\ref{tab:exp_d_timing}, Figure~\ref{fig:exp_d_crossover}): the $25\times$ per-candidate overhead prevents DMCI from becoming faster than direct compilation at any candidate budget. Both methods produce identical GP trajectories (fitness equivalence): we verified per-seed that direct compilation and DMCI explore the same candidate program sequence and select the same best program and best loss across all five seeds (the GP path is identical because both execute the same underlying math). GP solution quality is not the object of study here: only the timing and the method equivalence are.

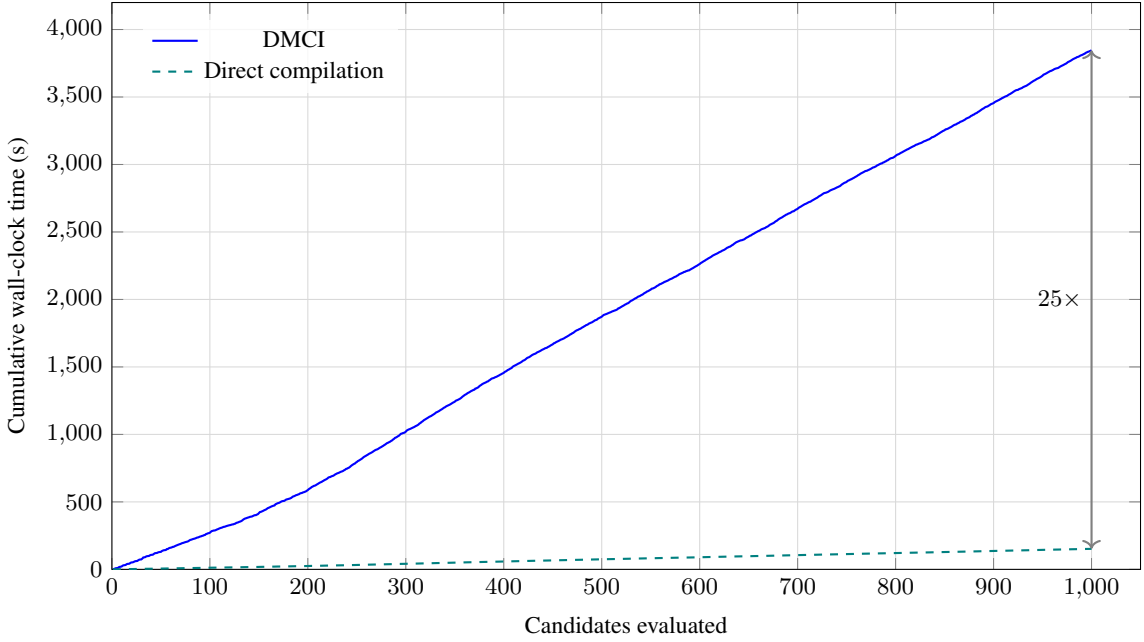
\begin{figure}[ht]
\centering
\begin{tikzpicture}
\begin{axis}[
    width=0.92\columnwidth,
    height=0.55\columnwidth,
    xlabel={Candidates evaluated},
    ylabel={Cumulative wall-clock time (s)},
    xmin=0, xmax=1050,
    ymin=0, ymax=4200,
    legend pos=north west,
    legend style={font=\small, draw=none, fill=white, fill opacity=0.8, text opacity=1},
    grid=major,
    grid style={gray!30},
    tick label style={font=\small},
    label style={font=\small},
]
\addplot[blue, thick, mark=none] table[x=candidates, y=time] {figures/exp_d_dmci_cumulative.dat};
\addlegendentry{DMCI}
\addplot[teal, thick, dashed, mark=none] table[x=candidates, y=time] {figures/exp_d_direct_cumulative.dat};
\addlegendentry{Direct compilation}
\draw[<->, gray, thick] (axis cs:1000,152.10) -- (axis cs:1000,3846.44) node[midway, left, font=\footnotesize, text=black] {$25\times$};
\end{axis}
\end{tikzpicture}
\caption{Experiment~D: cumulative wall-clock time. DMCI is $25\times$ steeper; no crossover occurs.}
\label{fig:exp_d_crossover}
\end{figure}

\subsection{Implications for the Utility Argument}

The absence of a wall-clock crossover means DMCI's value in structural-search workloads is not a runtime advantage over direct compilation. Instead, the value lies in three properties that this experiment does not measure on a clock:

\begin{enumerate}[leftmargin=2em,topsep=2pt,itemsep=1pt]
    \item \textbf{Guaranteed correctness.} Direct compilation produces correct gradients for each candidate, but this correctness must be verified per-program (or trusted based on compiler testing). DMCI provides correctness through the evaluator's compilation: any program representable as an S-expression receives correct gradients by construction.
    \item \textbf{Zero-effort integration.} In Experiment~B, 15 LLM-generated Scheme programs compiled and trained with zero manual adaptation. Direct compilation also works for these programs, but DMCI's evaluator-as-module architecture makes the integration path explicit: string in, differentiable module out.
    \item \textbf{Runtime composition.} Programs can be composed, modified, or selected at runtime without touching the compilation pipeline. This matters for workloads beyond GP, such as an LLM agent that iteratively refines a scientific model, or a meta-learning system that selects among a library of compiled programs.
\end{enumerate}

The $25\times$ per-candidate overhead is the cost of this flexibility. Batched evaluation (Experiment~H, Section~\ref{sec:exp_h}) addresses one dimension of this cost: by evaluating all $N$ data points simultaneously, the per-epoch graph-walking cost is paid once rather than $N$ times, yielding $12$--$45\times$ training speedups. However, the structural search setting of Experiment~D mutates the program each generation, so the evaluator graph changes per candidate; batching helps within each candidate's training loop but does not eliminate the per-candidate overhead. Further reductions through MLIR compilation of the evaluator (Section~\ref{sec:limitations}) could address the remaining tagged-value wrapping and Python dispatch costs, which account for 90\% of the per-evaluation overhead (Table~\ref{tab:overhead_decomp}).

% ============================================================================
\section{Experiment E: Discrete-Continuous Operator Recovery}
\label{app:exp_e}
% ============================================================================

\subsection{Task Design}

Each target function has the form $f^*(x) = a^* \cdot \texttt{op}_1(x, \texttt{op}_2(x))$, where $\texttt{op}_1, \texttt{op}_2 \in \{+, -, \times, \div, \sin, \cos, \exp, \log\}$. Search space: $8 \times 8 = 64$ combinations plus continuous constant $a^*$. Twelve targets spanning three constant regimes.

\subsection{Methods}

Four methods: (1)~DMCI soft-dispatch via Gumbel-Softmax \texttt{soft-choice} (temperature annealed 1.0 to 0.1 over 3000 epochs, Adam lr $= 0.05$), (2)~exhaustive enumeration with least-squares, (3)~evolutionary algorithm (elitist generational GA with tournament selection $k{=}3$, pop 32, 100 generations), (4)~random search (one random operator-pair draw per restart). Twenty restarts per target (240 per method).

\paragraph{Training objective.} DMCI's loss includes a derivative-matching term: $\mathcal{L} = \sum_i (f(x_i) - y_i)^2 + \lambda_d (f'(x_i) - y'_i)^2$ with $\lambda_d = 1.0$, where $f'$ and $y'$ are computed by autograd. This term exploits DMCI's ability to differentiate through the soft-dispatched interpreter and is not available to the non-gradient baselines (exhaustive, evolutionary, random), which optimize value-matching only. The asymmetry favors DMCI on continuous precision; discrete recovery rates are comparable with or without the term, since operator selection is primarily driven by value matching.

\subsection{Results}

Table~\ref{tab:exp_e_results} reports per-target operator-recovery success rates and continuous-constant precision, and Figure~\ref{fig:exp_e_success} visualizes the per-target success rates.

\begin{table}[ht]
\centering
\caption{Experiment~E: operator recovery success rates (20 restarts per target). Mean $|a{-}a^*|$ is shown for DMCI and exhaustive only; evolutionary precision matches exhaustive (both apply least-squares constant refinement) and random rarely recovers the correct operators, so these columns are omitted for space.}
\label{tab:exp_e_results}
\small
\begin{tabular}{@{}llrrrrrrr@{}}
\toprule
 & & & \multicolumn{4}{c}{Success rate} & \multicolumn{2}{c}{Mean $|a{-}a^*|$} \\
\cmidrule(lr){4-7} \cmidrule(lr){8-9}
ID & Target $f^*(x)$ & $a^*$ & DMCI & Exh. & Evol. & Rand. & DMCI & Exh. \\
\midrule
T01 & $0.5 \cdot x \cdot 2x$ & 0.5 & 25\% & 100\% & 20\% & 10\% & 0.0001 & 0.0001 \\
T02 & $0.5 \cdot x \cdot \sin x$ & 0.5 & 5\% & 100\% & 90\% & 0\% & 0.0008 & 0.0009 \\
T03 & $0.5 \cdot x \cdot e^x$ & 0.5 & 5\% & 100\% & 100\% & 0\% & 0.0001 & 0.0001 \\
T04 & $0.5 \cdot (x + x^2)$ & 0.5 & 10\% & 100\% & 75\% & 0\% & 0.0001 & 0.0001 \\
T05 & $\sin(x^2)$ & 1.0 & 10\% & 100\% & 100\% & 0\% & 0.0001 & 0.0019 \\
T06 & $x + \cos x$ & 1.0 & 15\% & 100\% & 100\% & 0\% & 0.0006 & 0.0004 \\
T07 & $x - \sin x$ & 1.0 & 20\% & 100\% & 100\% & 0\% & 0.0012 & 0.0009 \\
T08 & $\cos(x^2)$ & 1.0 & 5\% & 100\% & 95\% & 15\% & 0.0001 & 0.0014 \\
T09 & $2(x + x^2)$ & 2.0 & 15\% & 100\% & 80\% & 0\% & 0.0001 & 0.0002 \\
T10 & $2\cos(2x)$ & 2.0 & 5\% & 100\% & 100\% & 0\% & 0.0000 & 0.0017 \\
T11 & $2\sin(2x)$ & 2.0 & 10\% & 100\% & 90\% & 0\% & 0.0003 & 0.0012 \\
T12 & $2 \cdot x \cdot \cos x$ & 2.0 & 5\% & 100\% & 45\% & 0\% & 0.0004 & 0.0008 \\
\midrule
\textbf{All} & & & \textbf{10.8\%} & \textbf{100\%} & \textbf{82.9\%} & \textbf{2.1\%} & \textbf{0.0004} & \textbf{0.0008} \\
\bottomrule
\end{tabular}
\end{table}

\paragraph{Discrete recovery.} Exhaustive enumeration achieves 100\% (240/240; Wilson 95\% lower bound 98.4\%). Evolutionary: 82.9\% (199/240; Wilson 95\% CI [77.6, 87.2]\%). DMCI: 10.8\% (26/240; [7.5, 15.4]\%). Random: 2.1\% (5/240; [0.9, 4.8]\%).  All intervals are Wilson score intervals at 95\%.

\paragraph{Continuous parameter precision.} Despite modest discrete search performance, DMCI achieves the best continuous parameter precision among methods that recover the correct operators. On the 26 restarts where DMCI correctly recovers both operators, the mean constant error is $|a - a^*| = 3.7 \times 10^{-4}$ with median $1.6 \times 10^{-4}$, compared to $0.0008$ for exhaustive enumeration ($2.2\times$ better); the gradient-based optimizer converges to higher precision than the least-squares baseline. Over all 240 restarts (including the 214 where discrete recovery fails), DMCI's mean $|a - a^*| = 1.11$, reflecting that failed operator recovery yields meaningless constant fits. The table reports precision over successful restarts only. This result is consistent with DMCI's core strength: once the discrete structure is (correctly) determined by the annealed Gumbel-Softmax, the continuous optimizer achieves excellent precision through exact gradients.

\paragraph{Why discrete search is hard.} DMCI's 10.8\% success rate, while $5\times$ above random search, reflects the well-known difficulty of gradient-based discrete optimization. The Gumbel-Softmax relaxation creates a continuous surrogate for a discrete choice, but the loss landscape over operator logits is highly multimodal: each of the 64 combinations defines a different functional form, creating basins separated by regions where the soft-blended function matches no target well. Temperature annealing from $\tau = 1.0$ (soft blending) to $\tau = 0.1$ (near-discrete selection) must traverse these basins, and most restarts converge to incorrect local optima. This is exactly the phenomenon TerpreT documented for differentiable interpreters: gradient descent over discrete program structure converges in a small fraction of restarts.

\begin{figure}[t]
\centering
\begin{tikzpicture}
\begin{axis}[
    ybar,
    width=0.92\columnwidth,
    height=0.5\columnwidth,
    bar width=3pt,
    xlabel={Target},
    ylabel={Success rate (\%)},
    ymin=0, ymax=110,
    xtick=data,
    xticklabels={T01,T02,T03,T04,T05,T06,T07,T08,T09,T10,T11,T12},
    x tick label style={font=\scriptsize, rotate=45, anchor=east},
    tick label style={font=\small},
    label style={font=\small},
    legend style={font=\scriptsize, at={(0.98,0.98)}, anchor=north east, draw=none, legend columns=2},
    grid=major,
    grid style={gray!20},
    ymajorgrids=true,
    xmajorgrids=false,
    area legend,
]

\addplot[fill=blue!70, draw=blue!70] table[x=task, y=dmci] {figures/exp_e1_success_rates.dat};
\addlegendentry{DMCI}

\addplot[fill=orange!80, draw=orange!80] table[x=task, y=evolutionary] {figures/exp_e1_success_rates.dat};
\addlegendentry{Evolutionary}

\addplot[fill=red!70, draw=red!70] table[x=task, y=random] {figures/exp_e1_success_rates.dat};
\addlegendentry{Random}

\end{axis}
\end{tikzpicture}
\caption{Operator recovery success rates by target (Experiment~E, 20 restarts each). Exhaustive enumeration (100\% on all targets, omitted for clarity) is the ceiling. The evolutionary algorithm (orange) achieves the highest non-exhaustive rates on most targets. DMCI (blue) shows per-target rates of 5--25\%, above random search (red, 0--15\%) but well below exhaustive or evolutionary methods. The variation across targets is suggestive of loss-landscape geometry effects: targets with distinct functional signatures (T01, T07) may be easier for gradient-based search than those with similar-looking alternatives (T02, T10). This interpretation is tentative given the small per-target restart counts (20 each), where individual rates have wide binomial confidence intervals; only the aggregate rate (10.8\%) is well-estimated.}
\label{fig:exp_e_success}
\end{figure}
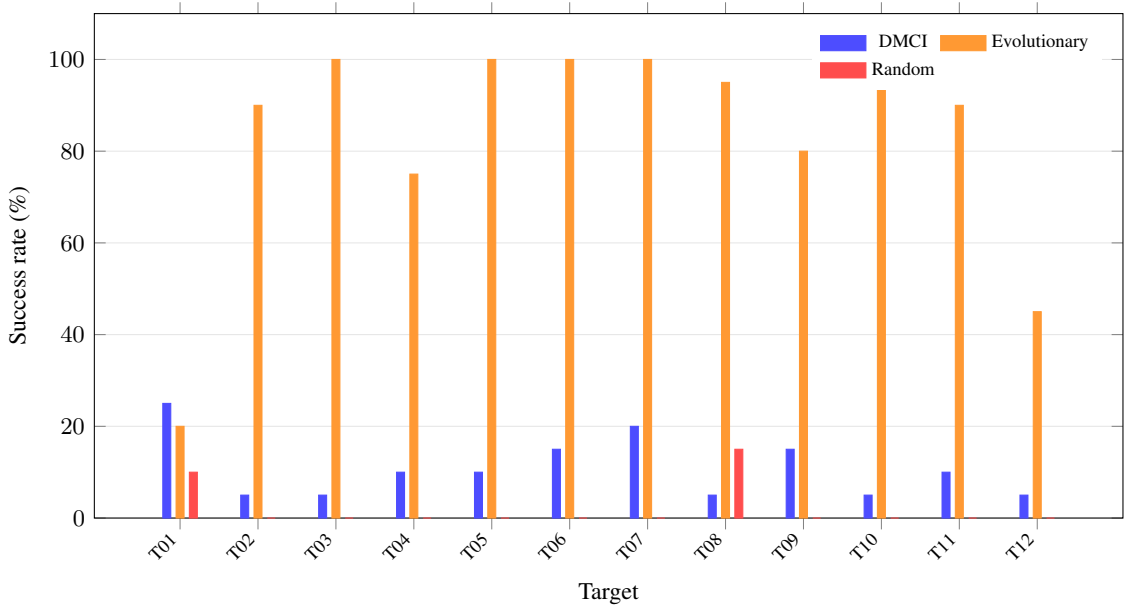

\paragraph{Wall-clock cost.} DMCI averages 689\,s per restart (the full 3000-epoch optimization through the compiled interpreter). Exhaustive enumeration averages 0.012\,s per restart (64 least-squares fits). The evolutionary algorithm averages 0.57\,s. DMCI is ${\sim}57{,}000\times$ slower than exhaustive and ${\sim}1{,}200\times$ slower than evolutionary search on this 64-combination space. For small discrete spaces, exhaustive enumeration dominates on every axis. DMCI's potential lies in scaling to spaces too large for exhaustive search, where gradient signal through the relaxed dispatch could guide the optimizer more efficiently than random or evolutionary methods, but validating this requires larger-scale experiments that we leave to future work.

\subsection{Implications}

Experiment~E confirms and refines DMCI's positioning. The compiled interpreter's strength is continuous parameter optimization: when the program structure is known, DMCI matches direct compilation exactly (Experiments~A--C) and achieves the best constant precision among successful restarts in the joint discrete-continuous setting ($2.2\times$ better than exhaustive on correct-operator restarts, Table~\ref{tab:exp_e_results}). Discrete structure search via Gumbel-Softmax relaxation is feasible but limited by the multimodal loss landscape over operator logits, consistent with TerpreT's findings for all gradient-based program search. The 10.8\% success rate is above random (2.1\%) but far below exhaustive (100\%) and evolutionary (82.9\%) baselines on this small space. Whether gradient-based structure search through a compiled interpreter can be competitive in larger spaces, where exhaustive enumeration is impractical, remains an open question.

% ============================================================================
\section{Experiment F: LLM-in-the-Loop Scientific Model Discovery}
\label{app:exp_f}
% ============================================================================

\paragraph{Motivation.} Experiment~B (Section~\ref{sec:exp_b}) established that an LLM-generated Scheme program compiles to a differentiable module and optimizes end-to-end. Experiment~F closes the loop: the LLM proposes the \emph{structure} of a scientific model as a Scheme expression, DMCI compiles and calibrates its constants by gradient descent, and the residuals are fed back to the LLM to refine the structure. Because a program is runtime data that compiles directly to a differentiable graph, ``a model was discovered and trained'' requires no per-model engineering. The experiment is necessary for two reasons: it tests whether the generation capability extends to autonomous \emph{discovery}, and, by holding the loop fixed while varying only the LLM's reasoning and the parameter optimizer, it isolates \emph{where} discovery fails, directly grounding the correctness-versus-optimization distinction of Section~\ref{sec:tradeoffs}.

\paragraph{Design.} Four targets of increasing structural difficulty: F1 exponential decay, F2 damped oscillation ($a\sin(bx)\,e^{-cx}$), F3 decay-plus-sine ($a\,e^{-bx}+c\sin(dx)$), and F4 logistic growth, each sampled at 64 points with 2\% noise, three seeds. Each iteration: the LLM proposes a Scheme expression in the supported subset (input $x$, named constants $a,b,c,d$, binary operators, no \texttt{define}/\texttt{lambda}); the proposal is parsed, arity-checked, compiled, and evaluated for finiteness (up to two retry prompts on failure); DMCI fits the constants by batched gradient descent; the residuals are analyzed and fed back via a refinement prompt. The loop runs $\le 5$ iterations, stops at held-out MSE $<10^{-3}$, and keeps a best-of-iterations checkpoint. We vary the LLM (Qwen3.6-27B via the MindRouter endpoint; GPT-5.5 via the OpenAI API), the reasoning mode (off/on), and the parameter optimizer (single-start Adam vs.\ the multi-optimizer portfolio of Section~\ref{sec:tradeoffs}: a cheap-first Adam~$\rightarrow$~multi-start L-BFGS cascade with held-out selection).

\paragraph{Results.} Table~\ref{tab:exp_f} reports per-target convergence (converged seeds out of three). The original configuration (Qwen3.6-35B, no reasoning, single-start Adam) converged 6/12. Tracing the failures revealed two \emph{distinct} causes, neither a limitation of DMCI: \textbf{(1) structure-discovery failures}, without reasoning the model one-shot-guessed the wrong functional family (F2); enabling reasoning on a capable model raised F2 from 1/3 to 3/3, lifting the total to 9/12 (the smaller Qwen3.6-27B, even with reasoning, still missed F2, so model capability matters). \textbf{(2) parameter-optimization failures}, on F3, GPT-5.5 with reasoning proposed the \emph{correct} structure $a\,e^{-bx}+c\sin(dx)$ on the first iteration, yet single-start Adam (initialized near~1) could not recover the sine frequency $d\approx3$: the frequency landscape is strongly multimodal and Adam settled in a wrong-frequency basin, after which the residual-feedback loop mis-diagnosed a fitting failure as a structure failure and oscillated between structures. Replacing Adam with the multi-optimizer portfolio recovers the frequency: the full thinking-plus-portfolio discovery loop reaches \textbf{11/12}. The single residual (F3 seed~1) is a loop dynamic rather than a fitting barrier; the loop refined away from the fittable structure (best at iteration~2, final MSE $1.7\times10^{-2}$); feeding the correct iteration-0 F3 structure directly to the multi-start L-BFGS fitter converges all three F3 seeds to $3$--$6\times10^{-4}$ (3/3).

\begin{table}[ht]
\centering
\caption{Experiment~F: discovery-loop convergence (seeds converged out of 3, MSE~$<10^{-3}$). Reasoning fixes structure-discovery failures (F2); the multi-optimizer portfolio fixes parameter-optimization failures (F3). The lone residual (F3 seed~1) is a loop-refinement artifact: an isolated multi-start fit of the correct F3 structure converges 3/3.}
\label{tab:exp_f}
\begin{tabular}{lccccc}
\toprule
Configuration (model, reasoning, optimizer) & F1 & F2 & F3 & F4 & Total \\
\midrule
Qwen3.6-35B, no reasoning, Adam            & 3/3 & 1/3 & 0/3 & 2/3 & \textbf{6/12} \\
Qwen3.6-27B, reasoning, Adam               & 3/3 & 0/3 & 0/3 & 3/3 & \textbf{6/12} \\
GPT-5.5, reasoning, Adam                   & 3/3 & 3/3 & 0/3 & 3/3 & \textbf{9/12} \\
GPT-5.5, reasoning, portfolio              & 3/3 & 3/3 & 2/3 & 3/3 & \textbf{11/12} \\
\midrule
GPT-5.5, reasoning, multi-start (F3 only)  & n/a & n/a & 3/3 & n/a & n/a \\
\bottomrule
\end{tabular}
\end{table}

\paragraph{Conclusion.} End-to-end LLM-to-differentiable-program discovery works, and the experiment localizes its two failure modes precisely: \emph{structure} discovery (improved by reasoning on a capable model) and \emph{parameter} optimization (improved by a multi-optimizer portfolio). Syntax was never the bottleneck; every proposal across the runs was valid Scheme. This is the empirical basis for the correctness-versus-optimization-success limitation (Section~\ref{sec:tradeoffs}): compiling the interpreter guarantees correct gradients, but discovering and fitting a runtime-proposed program also depends on the LLM's structural search and the optimizer's ability to navigate the resulting landscape.

% ============================================================================
\section{Experiment G: Runtime Compositional Modeling}
\label{app:exp_g}
% ============================================================================

\paragraph{Motivation.} Because DMCI treats programs as data, a library of symbolic model components are simply strings, \emph{composition} is string manipulation, and a freshly composed program flows through the same compiled interpreter with no per-composition engineering. Experiment~G makes this concrete and supplies the empirical counterpart to the paper's runtime-composition claim. It tests three things: that gradient descent through the compiled interpreter can identify the \emph{correct} way to compose two components from data; that incorrect compositions are distinguishable; and that swapping a component (``hot-swapping'') recompiles cheaply.

\paragraph{Design.} A small module library (exponential decay, oscillation, polynomial, sigmoid, power law, Gaussian) and three target problems whose ground truth is a known composition: G1 decay-plus-oscillation (\emph{sum}), G2 damped oscillation (\emph{product} of oscillation and decay), and G3 sigmoid-of-polynomial (\emph{chain}); five seeds each. For every problem we fit the individual modules, the correct composition, the plausible \emph{wrong} compositions (different operator or operand order), and a hot-swap (replace one module with an alternative, recompile, and refit). Every composition is built by string manipulation and compiled fresh; we record the mean fitted MSE and the mean recompilation time.

\paragraph{Results.} Table~\ref{tab:exp_g} reports mean MSE (over five seeds) for the correct composition and the best competing wrong composition, plus the mean hot-swap recompilation time. For the \emph{product} (G2) and \emph{chain} (G3) targets the correct composition wins decisively ($4\times10^{-4}$ versus $\ge 1.1\times10^{-1}$ for G2 and versus $\ge 7\times10^{-4}$ for G3), and hot-swap recompilation is $20$--$35$\,ms, i.e.\ negligible and requiring zero engineering. The \emph{sum} target (G1) is the honest exception: decay-plus-oscillation is poorly separable from a chain of the same parts, so the correct sum (0.94) does not clearly beat the wrong chain (0.77); compositional identifiability depends on how distinguishable the candidate structures are on the data, not on any limitation of the mechanism.

\begin{table}[ht]
\centering
\caption{Experiment~G: runtime composition (mean MSE over 5 seeds). Correct compositions win decisively for product (G2) and chain (G3) structures; the sum target (G1) is poorly separable from a chain of the same parts. Hot-swap recompilation is $20$--$35$\,ms regardless.}
\label{tab:exp_g}
\begin{tabular}{llccc}
\toprule
Problem & Correct op & Correct MSE & Best wrong MSE & Hot-swap recompile \\
\midrule
G1 decay $+$ oscillation & sum     & 0.942 & 0.772 (chain) & 20.0\,ms \\
G2 damped oscillation    & product & \textbf{0.0004} & 0.108 (sum) & 35.2\,ms \\
G3 sigmoid of polynomial & chain   & \textbf{0.0004} & 0.0007 (sum) & 27.2\,ms \\
\bottomrule
\end{tabular}
\end{table}

\paragraph{Conclusion.} Runtime composition is real and essentially free on the engineering axis: components compose by string manipulation, the composite is differentiable by construction, and hot-swapping recompiles in tens of milliseconds. Where the candidate structures are distinguishable on the data (product, chain), gradient descent identifies the correct composition; where they are not (the sum target), identifiability, not the composition mechanism, is the limiting factor.

% ============================================================================
\section{Experiment K (FluZoo): Influenza Co-Search and the Fitness-Fidelity Stress Test}
\label{app:exp_fluzoo_ext}
% FluZoo is the cautionary half of the co-search story in \S\ref{sec:exp_battery}: its apparent
% win is a search-time-fitness artifact that vanishes under rigorous re-scoring. All numbers below
% are final, drawn from experiments/exp_fluzoo/results/ (baseline_test.json + _pool_124_125/);
% see that directory's MANIFEST.txt for the artifact-to-claim map.

\paragraph{Data.} Weekly weighted percent influenza-like illness (wILI) is obtained from CDC
ILINet through the Delphi Epidata \texttt{fluview} endpoint~\citep{farrow2015delphi,cdc_fluview}
for the national region and the ten HHS regions ($R=11$ series), MMWR epiweeks 2010w40--2025w39
(latest issue; per-region release dates recorded for a pinned \emph{as-of} snapshot). Values are
expressed as proportions (wILI${}/100$) so they match the model observable. Seasons are split
train (2010--2017), validation (2018, 2021--22, 2022--23), and test (2023--24); the 2020--21
pandemic season is held out as distribution shift. The observation matrix bound into every
program is $\mathrm{obs}\in\mathbb{R}^{T\times 11}$; \texttt{(ref obs k)} gathers week $k$. Short-horizon forecasting of this influenza-like-illness signal is the task of the CDC {FluSight} forecasting challenge~\citep{mathis2024flusight}.

\paragraph{Program contract and generation.} Each generated artifact is two S-expressions: a
machine-readable \texttt{(params ...)} schema declaring every parameter with a constraint kind
(positive, unit, signed-unit, free) and an initial value, and a tail-recursive weekly rollout
\texttt{(loop ((k 0) ...\,(yhat ...)\,(L 0.0)) (if (= k NWEEKS) L (let* (...) (recur ...))))}
that accumulates a Gaussian negative log-likelihood with a floored observation variance and
carries the predicted observable \texttt{yhat}. Seasonal transmission is built from the integer
week counter (e.g.\ $\cos(\omega k+\varphi)$), avoiding external per-step forcing (which would
blow up the interpreter heap). The language model (Qwen3.6-35B) is prompted with the op surface
and these rules plus a recipe sampled across families (compartmental, seasonal, observation,
regional-coupling, closure) to drive diversity.

\paragraph{Validity funnel.} Each proposal passes through, in order: parse and schema, static
operator prescan (rejecting any head outside the interpreter's surface, which would silently
evaluate to zero), compilation, a free-variable contract check, finite forward NLL, finite and
nonzero gradients to every parameter, a stable rollout across random parameter draws, and a
forecastability check (the prediction form emits a finite $R$-vector). On failure a targeted
repair hint is returned and the program regenerated (up to three repairs). A canonical structural
hash (positional renaming of bound variables and parameters, constant bucketing) counts distinct
structures.

\paragraph{Calibration and forecasting.} Each accepted program is bound, as data, to the same
compiled interpreter; its parameters are mapped from unconstrained leaves through their declared
transforms and fit by multi-start Adam summing the per-season training NLL. Selection is on
held-out validation skill, never training NLL. Held-out skill uses filter-then-forecast: with
structural parameters frozen at their trained values, the initial-condition parameters are
re-estimated on the observed weeks up to each test-season forecast origin, and the autonomous
model is rolled $1$--$4$ weeks ahead (the prediction form is the same program with its loop
base case returning \texttt{yhat}; future weeks are bound to a zero-padded observation matrix
since the dynamics are autonomous). Baselines (seasonal-naive, autoregressive, classical regional
SEIR through the same interpreter, random-grammar${}+{}$DMCI, and direct-compile-each) are scored
on identical origins and horizons. The reproduction pipeline is
\texttt{experiments/exp\_fluzoo/} (\texttt{data.build\_data} $\to$ \texttt{llm\_generate} $\to$
\texttt{run\_all} $\to$ \texttt{aggregate}).

\paragraph{The fitness-fidelity finding (the stress test).} Three long-horizon islands were evolved to completion ($200$ iterations each); the global best---a spatial mean-field-coupling structure with dual exposed/infectious reporting, $11$ parameters---attained a search-time validation RMSE of $0.0100$ under the cheap $60$-step calibration. Re-scored under the rigorous consistent protocol ($300$ Adam steps, rolling forecast origins, initial-condition refit), that advantage disappears: the evolved program scores validation $0.01477$ / test $0.01812$, against the hand-written regional SEIR seed at $0.01488$ / $0.01805$ and an SEIRS variant at $0.01473$ / $0.01856$---a tie on both held-out splits. The under-converged $60$-step fitness acts as an implicit regularizer that flatters the more flexible $11$-parameter program, and the outer search optimized that loose proxy rather than true skill; converging the inner fitness erases the gap. We report this as a positive methodological result: in program-and-parameter co-search the inner-loop fitness must be converged enough to be a faithful skill signal, or the outer structure search overfits it. Mitigations we recommend are a converged or nested re-ranked fitness, rolling-origin or $k$-fold cross-validation as the selection signal, complexity-penalized selection, a cross-island agreement filter, and significance testing given few test seasons. On real influenza the long-horizon signal that would separate structures is weak, so the co-search correctly collapses to the baseline; the battery study (\S\ref{sec:exp_battery}) is the complementary case where the data \emph{do} discriminate structure and the recovered mechanism survives the same rigorous re-scoring.

\section{Experiment H: Extended Results}
\label{app:exp_h_ext}
% ============================================================================

Figure~\ref{fig:exp_h_speedup} summarizes batched forward-throughput speedups over sequential DMCI
across batch sizes for the Experiment~B model suite.

% Figure: Exp H - Speedup factor vs batch size (both models)
% Log-log plot showing near-linear throughput scaling
% Usage: \input{figures/fig_h_speedup.tex}
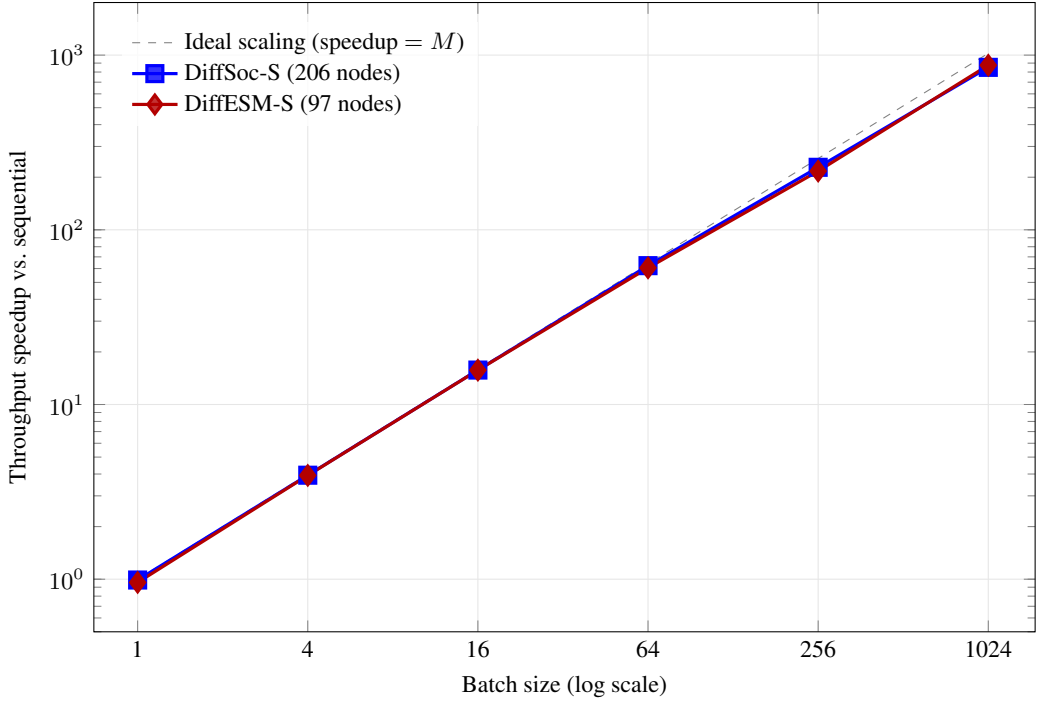
\begin{figure}[t]
\centering
\begin{tikzpicture}
\begin{loglogaxis}[
    width=0.85\columnwidth,
    height=0.60\columnwidth,
    xlabel={Batch size (log scale)},
    ylabel={Throughput speedup vs.\ sequential},
    xmin=0.7, xmax=1500,
    ymin=0.5, ymax=2000,
    xtick={1,4,16,64,256,1024},
    xticklabels={1,4,16,64,256,1024},
    grid=major,
    grid style={gray!20},
    tick label style={font=\small},
    label style={font=\small},
    legend style={font=\small, at={(0.03,0.97)}, anchor=north west, draw=none, fill=white, fill opacity=0.8, text opacity=1},
    legend cell align={left},
]

% Ideal linear speedup reference line
\addplot[gray, thin, dashed, domain=1:1024, samples=2] {x};
\addlegendentry{Ideal scaling (speedup $= M$)}

% DiffSoc-S (206-node model)
\addplot[blue, very thick, mark=square*, mark size=3pt]
    table[x=batch_size, y=speedup] {figures/exp_h_speedup_soc.dat};
\addlegendentry{DiffSoc-S (206 nodes)}

% DiffESM-S (97-node model)
\addplot[red!70!black, very thick, mark=diamond*, mark size=3.5pt]
    table[x=batch_size, y=speedup] {figures/exp_h_speedup_esm.dat};
\addlegendentry{DiffESM-S (97 nodes)}

\end{loglogaxis}
\end{tikzpicture}
\caption{Throughput speedup relative to sequential evaluation ($M{=}1$) as batch size increases (100 timesteps, single CPU core). The dashed line represents ideal linear scaling where speedup equals batch size. Both models closely track this line despite a $2.1\times$ difference in graph complexity (97 vs.\ 206 nodes), confirming that interpreter traversal and Python dispatch costs are amortized across batch elements while vectorized tensor operations scale efficiently. At $M{=}1024$, batching achieves $851$--$875\times$ throughput, demonstrating that the per-evaluation interpreter overhead is substantially amortized at large batch sizes.}
\label{fig:exp_h_speedup}
\end{figure}

\subsection{Implementation}

We extended the tagged-value constructors (\texttt{make\_float}, \texttt{make\_int}, \texttt{make\_bool}) to detect batch dimensions: when given a tensor of shape $(*B)$ rather than a scalar, they produce tagged values of shape $(*B, 14)$. The extractors (\texttt{unwrap\_number}, \texttt{extract\_tag}) already used ellipsis indexing (\texttt{tv[..., :10]}) and required no changes. A new entry point, \texttt{evaluate\_batched}, walks the computation graph once with batched tensors. Constants broadcast naturally against batched inputs. Programs using heap operations (cons, car, cdr, closures) are excluded \emph{from this directly-compiled batched path}; all 12 arithmetic and program-tier models from Experiment~B that do not require heap allocation are supported. Batching the heap-\emph{using} meta-circular interpreter itself is a distinct, later capability, benchmarked separately in the next subsection (Table~\ref{tab:exp_h_dmci_batched}).

\subsection{Batching the Meta-Circular Interpreter Itself}

The eight parts below batch \emph{directly-compiled}, heap-free graphs. The complementary case, batching the heap-\emph{using} compiled interpreter (DMCI) itself, is possible because a fixed program's heap layout and control flow are data-independent: structural (tag and heap-address) values stay scalar while only numeric leaves carry the batch dimension. Table~\ref{tab:exp_h_dmci_batched} measures three programs run \emph{through} the compiled meta-circular interpreter (\texttt{bootstrap/compiler.scm} with \texttt{(scheme-eval $P$ env)}) at batch sizes up to 256 on a single CPU core. A single batched interpreter walk reproduces $N$ sequential walks \emph{bit-for-bit} (max $|\textrm{seq}-\textrm{bat}|=0$, all gradients finite) and yields $214$--$270\times$ forward and $56$--$63\times$ training throughput over the sequential interpreter, the figures quoted in Section~\ref{sec:exp_h}.

\begin{table}[ht]
\centering
\caption{Self-batching the heap-\emph{using} meta-circular interpreter (DMCI itself), the case the directly-compiled batched path of this appendix excludes. Each program is executed \emph{through} the compiled interpreter (\texttt{bootstrap/compiler.scm} $+$ \texttt{(scheme-eval $P$ env)}); a single batched interpreter walk reproduces $N$ sequential walks bit-for-bit. Forward speedup is the sequential ($N$ walks) versus batched (1 walk) ratio; training speedup is batched versus sequential fitting at $N{=}64$ over 15 epochs. Single CPU core; artifact \texttt{dmci\_batched\_bench.json} in \texttt{experiments/exp\_h/results/}.}
\label{tab:exp_h_dmci_batched}
\small
\begin{tabular}{@{}lccccc@{}}
\toprule
 & & & \multicolumn{2}{c}{Forward speedup} & Training \\
\cmidrule(lr){4-5}
Program (run through interpreter) & Params & max$|\textrm{seq}{-}\textrm{bat}|$ & BS=64 & BS=256 & speedup ($N{=}64$) \\
\midrule
Beer--Lambert (arithmetic)       & 2  & $0$ & $67.1\times$ & $270.2\times$ & $56.3\times$ \\
Decay chain (recursive ODE)      & 2  & $0$ & $64.0\times$ & $255.1\times$ & $63.3\times$ \\
GPP (multi-driver composite)     & 12 & $0$ & $64.0\times$ & $214.0\times$ & $62.9\times$ \\
\bottomrule
\end{tabular}
\end{table}

\subsection{Experimental Design}

Eight sub-experiments measure the practical impact of batching, all run on a single NVIDIA A100 (40\,GB) GPU:

\begin{enumerate}[leftmargin=2em,topsep=2pt,itemsep=1pt]
\item \textbf{Part~A (Forward throughput).} For each of 12 directly-compiled models from Experiment~B, measure forward-pass time at batch sizes 1--4096. Compare batched evaluation against a sequential loop baseline.
\item \textbf{Part~B (Training speedup).} Run 500-epoch training with Adam on 64 data points, comparing sequential (64 evaluations per epoch) against batched (1 evaluation per epoch). Both start from identical initial parameters.
\item \textbf{Part~C (Correctness).} Verify that batched and sequential evaluation produce identical predictions (within floating-point tolerance) and that gradients flow correctly through the batched path.
\item \textbf{Part~D (Population batching).} Evaluate $M$ random parameter initializations across $N$ data points simultaneously, with shape $(M, N)$, through a single forward pass. This tests DMCI as an engine for population-based optimization.
\item \textbf{Part~E (\texttt{torch.compile}).} Apply \texttt{torch.compile} (PyTorch~2.x graph compilation with Inductor backend) to the batched evaluator, measuring whether kernel fusion and Python overhead elimination provide additional speedup beyond batching alone.
\item \textbf{Part~F (Large-model stress test).} Benchmark a coupled 9-module Earth System Model (DiffESM-S) with 95 inputs, 70 learnable parameters, 20 state variables, and 100-step recursive time integration. This tests batching on a program an order of magnitude larger than the Experiment~B models, with complex nonlinear feedback loops and deeply nested arithmetic.
\item \textbf{Part~G (Convergence speedup).} Fit 15 DiffESM-S climate parameters to synthetic observational data ($N{=}32$ data points). Compare sequential, batched, and population-batched ($M$ random restarts $\times$ $N$ data points) training over 300 epochs, each restart initialized by perturbing the 15 parameters with $15\%$ Gaussian noise. This tests whether throughput gains translate to faster optimization convergence.
\item \textbf{Part~H (Second large-model stress test).} Benchmark \textsc{DiffSoc-S}, a 10-module Urban Political Economy Simulator with 204 inputs, 87 learnable parameters, and 206 compiled graph nodes. This second large model tests whether the batching results from Part~F generalize across model domains and at $2.1\times$ the graph complexity of DiffESM-S.
\end{enumerate}

\subsection{Correctness Verification (Part~C)}

\begin{table}[ht]
\centering
\caption{Experiment~H Part~C: Correctness verification on A100 GPU.}
\label{tab:exp_h_correctness}
\small
\begin{tabular}{@{}lrrr@{}}
\toprule
Model & Pred.\ diff & Loss at $\theta^*$ & Grads OK \\
\midrule
M01 Coulomb & $0.00$ & $2.8 \times 10^{-14}$ & \checkmark \\
M02 Beer--Lambert & $0.00$ & $7.2 \times 10^{-14}$ & \checkmark \\
M03 Michaelis--Menten & $0.00$ & $9.7 \times 10^{-14}$ & \checkmark \\
M04 Arrhenius & $4.8 \times 10^{-7}$ & $1.9 \times 10^{-14}$ & \checkmark \\
M05 Damped oscillator & $2.4 \times 10^{-7}$ & $4.1 \times 10^{-14}$ & \checkmark \\
M06 Logistic growth & $4.8 \times 10^{-7}$ & $7.6 \times 10^{-14}$ & \checkmark \\
M07 Power law & $9.5 \times 10^{-7}$ & $1.9 \times 10^{-13}$ & \checkmark \\
M08 Euler ODE & $0.00$ & $6.6 \times 10^{-15}$ & \checkmark \\
M10 SiLU & $1.2 \times 10^{-7}$ & $4.9 \times 10^{-15}$ & \checkmark \\
M11 Recursive filter & $0.00$ & $1.1 \times 10^{-14}$ & \checkmark \\
M12 Newton sqrt & $0.00$ & $7.8 \times 10^{-15}$ & \checkmark \\
M14 Anomaly scorer & $0.00$ & $7.9 \times 10^{-15}$ & \checkmark \\
\bottomrule
\end{tabular}
\end{table}

All 12 models pass: maximum prediction difference between sequential and batched is below $10^{-6}$ (Table~\ref{tab:exp_h_correctness}), loss at true parameters is below $10^{-13}$, and all parameter gradients are finite. Batching introduces no numerical artifacts.

\subsection{Forward Throughput (Part~A)}

Table~\ref{tab:exp_h_throughput} reports batched forward-throughput speedups over sequential DMCI across batch sizes.

\begin{table}[ht]
\centering
\caption{Experiment~H Part~A: Forward throughput speedup (batched vs.\ sequential) on A100 GPU.}
\label{tab:exp_h_throughput}
\small
\begin{tabular}{@{}lrrrrr@{}}
\toprule
Model & BS=1 & BS=32 & BS=128 & BS=512 \\
\midrule
M01 Coulomb & 2.9$\times$ & 83$\times$ & 302$\times$ & 1{,}309$\times$ \\
M02 Beer--Lambert & 3.5$\times$ & 110$\times$ & 374$\times$ & 1{,}710$\times$ \\
M03 Michaelis--Menten & 2.7$\times$ & 77$\times$ & 312$\times$ & 1{,}226$\times$ \\
M04 Arrhenius & 2.0$\times$ & 56$\times$ & 240$\times$ & 950$\times$ \\
M05 Damped oscillator & 1.6$\times$ & 48$\times$ & 205$\times$ & 821$\times$ \\
M06 Logistic growth & 1.3$\times$ & 39$\times$ & 164$\times$ & 657$\times$ \\
M07 Power law & 3.2$\times$ & 84$\times$ & 385$\times$ & 1{,}531$\times$ \\
M08 Euler ODE & 0.6$\times$ & 17$\times$ & 66$\times$ & 279$\times$ \\
M10 SiLU & 1.5$\times$ & 39$\times$ & 158$\times$ & 662$\times$ \\
M11 Recursive filter & 0.6$\times$ & 18$\times$ & 74$\times$ & 291$\times$ \\
M12 Newton sqrt & 0.7$\times$ & 23$\times$ & 90$\times$ & 357$\times$ \\
M14 Anomaly scorer & 2.8$\times$ & 82$\times$ & 355$\times$ & 1{,}409$\times$ \\
\bottomrule
\end{tabular}
\end{table}

At batch size 1, batched evaluation is comparable to sequential (1--3$\times$). At batch size 512, speedups range from $279\times$ (M08 Euler ODE, 10 nested function calls) to $1{,}710\times$ (M02 Beer--Lambert, simple multiplication). The pattern is consistent: the per-evaluation overhead of Python dispatch and graph walking is paid once per batch, not once per input. Batched forward-pass time is nearly constant across batch sizes (0.03--0.7\,ms), confirming that the overhead is dominated by graph traversal rather than tensor computation.

\subsection{Training Speedup (Part~B)}

Table~\ref{tab:exp_h_training} reports batched training speedups over sequential DMCI.

\begin{table}[ht]
\centering
\caption{Experiment~H Part~B: Training time (500 epochs, 64 data points, Adam, A100 GPU).}
\label{tab:exp_h_training}
\small
\begin{tabular}{@{}lrrrr@{}}
\toprule
Model & Sequential & Batched & Speedup & Loss \\
\midrule
M01 Coulomb & 5.81\,s & 0.34\,s & 17.0$\times$ & 14.22 \\
M02 Beer--Lambert & 3.81\,s & 0.31\,s & 12.2$\times$ & 0.00 \\
M03 Michaelis--Menten & 5.95\,s & 0.35\,s & 16.9$\times$ & 4.60 \\
M04 Arrhenius & 6.82\,s & 0.37\,s & 18.5$\times$ & 0.14 \\
M05 Damped oscillator & 10.86\,s & 0.44\,s & 24.7$\times$ & $<10^{-5}$ \\
M06 Logistic growth & 11.06\,s & 0.44\,s & 25.0$\times$ & 0.32 \\
M07 Power law & 6.55\,s & 0.37\,s & 17.8$\times$ & 0.69 \\
M08 Euler ODE & 48.21\,s & 1.09\,s & 44.3$\times$ & 0.00 \\
M10 SiLU & 11.03\,s & 0.44\,s & 25.2$\times$ & $<10^{-3}$ \\
M11 Recursive filter & 42.02\,s & 0.94\,s & 44.8$\times$ & 0.00 \\
M12 Newton sqrt & 21.97\,s & 0.61\,s & 36.3$\times$ & 0.00 \\
M14 Anomaly scorer & 7.07\,s & 0.37\,s & 19.0$\times$ & 0.00 \\
\bottomrule
\end{tabular}
\end{table}

Batched training is $12$--$45\times$ faster than sequential, with losses matching to all reported digits. Program-tier models with nested function calls (M08, M11) show the largest speedups ($44$--$45\times$) because the per-call overhead of graph walking is amortized most effectively. A training run that took 48\,s sequentially completes in 1.1\,s batched.

\subsection{Population Batching (Part~D)}

Table~\ref{tab:exp_h_population} reports population-batching speedups.

\begin{table}[ht]
\centering
\caption{Experiment~H Part~D: Population batching speedup on A100 GPU ($N{=}64$ data points).}
\label{tab:exp_h_population}
\small
\begin{tabular}{@{}lrrrr@{}}
\toprule
Model & Pop=1 & Pop=10 & Pop=100 \\
\midrule
M01 Coulomb & 147$\times$ & 1{,}401$\times$ & 14{,}528$\times$ \\
M02 Beer--Lambert & 169$\times$ & 1{,}685$\times$ & 16{,}873$\times$ \\
M03 Michaelis--Menten & 144$\times$ & 1{,}402$\times$ & 14{,}395$\times$ \\
M04 Arrhenius & 109$\times$ & 1{,}111$\times$ & 11{,}145$\times$ \\
M05 Damped oscillator & 90$\times$ & 916$\times$ & 9{,}156$\times$ \\
\bottomrule
\end{tabular}
\end{table}

When $M$ independent parameter initializations are evaluated across $N$ data points simultaneously, the total evaluation count is $M \times N$ but the graph is walked only once. At population size 100, speedups range from $9{,}156\times$ to $16{,}873\times$: 6{,}400 evaluations complete in ${\sim}0.1$\,ms. Gradients flow correctly through all population members, enabling batched multi-start optimization in a single backward pass over a tensor of shape $(M, N)$. Standard vectorization (\texttt{jax.vmap}/\texttt{torch.func.vmap}) can vectorize a directly-compiled model over the same population axis with one graph and one optimizer, so this is not unique to DMCI; the distinction is that batched DMCI provides it \emph{automatically for any program supplied as data}, with no per-program vectorization, because the batch dimension broadcasts through the fixed interpreter graph. All speedups in this section are measured against DMCI's own sequential evaluation, not against a \texttt{vmap} baseline.

\subsection{Comparison to Standard Vectorization (\texttt{jax.vmap})}

To place the batched-throughput numbers against standard practice rather than only against DMCI's own sequential loop, we emit each heap-free model to a functional JAX module (\texttt{apply(params, **inputs)}, fully \texttt{jax.grad}/\texttt{jax.vmap}-able) and time \texttt{jax.jit(jax.vmap(\textbf{$\cdot$}))} of the \emph{directly-compiled} graph against DMCI's \texttt{evaluate\_batched} on the identical graph, data ($N{=}64$), and single CPU core (jax CPU backend). A numerical-equivalence guard confirms both paths compute the same function (maximum absolute difference $\le 9.5\times10^{-7}$ across all covered models). Table~\ref{tab:exp_h_vmap} reports DMCI throughput as a fraction of the \texttt{jax.vmap} baseline.

\begin{table}[ht]
\centering
\caption{DMCI batched throughput as a fraction of a \texttt{jax.vmap} baseline (directly-compiled model, same graph/data/CPU). Values ${<}1$ mean DMCI is slower than vmap. ``forward'' is a batched forward pass; ``train'' is one value-and-grad step. jit/trace (XLA compile) time is excluded from the steady-state timing.}
\label{tab:exp_h_vmap}
\begin{tabular}{lcc}
\toprule
Model & forward & train step \\
\midrule
M02 Beer--Lambert        & $0.93\times$ & $0.11\times$ \\
M01 Coulomb              & $0.63\times$ & $0.10\times$ \\
M07 power law            & $0.63\times$ & $0.10\times$ \\
M03 Michaelis--Menten    & $0.54\times$ & $0.10\times$ \\
M04 Arrhenius            & $0.46\times$ & $0.11\times$ \\
M14 anomaly scorer       & $0.45\times$ & $0.10\times$ \\
M05 Damped oscillator    & $0.36\times$ & $0.08\times$ \\
M06 logistic growth      & $0.19\times$ & $0.08\times$ \\
M12 Newton $\sqrt{\cdot}$ (5-iter) & $0.10\times$ & $0.03\times$ \\
M11 recursive filter (8-step)      & $0.05\times$ & $0.06\times$ \\
M08 Euler ODE (10-step)            & $0.04\times$ & $0.03\times$ \\
\midrule
\textbf{geometric mean} ($n{=}11$) & $\mathbf{0.27\times}$ & $\mathbf{0.076\times}$ \\
\bottomrule
\end{tabular}
\end{table}

The pattern is clear: for closed-form equations DMCI's broadcasting graph-walk is within a small constant factor of XLA-compiled \texttt{vmap} ($0.36$--$0.93\times$), while loop-heavy programs (Euler ODE, recursive filter, Newton iteration) are ${\sim}10$--$25\times$ slower because each loop iteration is a Python-level graph traversal rather than a fused XLA kernel; training steps are uniformly ${\sim}0.03$--$0.11\times$, where XLA's compilation advantage is largest. The comparison is therefore honest about speed: \texttt{vmap} wins wherever it applies. The two DMCI advantages it does not capture are (i) \emph{zero per-program vectorization}, the same interpreter batches every program supplied as data with no per-program code, and (ii) \emph{coverage}: of the models DMCI batches natively, M10 (a looping activation) and both production integrators (\textsc{DiffESM-S}, \textsc{DiffSoc-S}) cannot be \texttt{jit}/\texttt{vmap}-ed by the JAX backend at all (the recursive integrator raises a concretization error on the traced loop bound), whereas DMCI batches them bit-for-bit. The within-system speedups reported above ($\le 875\times$, $3{,}848\times$) are thus recoveries of interpreter overhead relative to sequential interpretation, not gains over standard vectorization.

\subsection{\texttt{torch.compile} (Part~E)}

Table~\ref{tab:exp_h_compile} reports the effect of \texttt{torch.compile} on the batched evaluator.

\begin{table}[ht]
\centering
\caption{Experiment~H Part~E: \texttt{torch.compile} speedup over plain batched evaluation. Values $<1.0\times$ indicate compiled is slower.}
\label{tab:exp_h_compile}
\small
\begin{tabular}{@{}lrrrrrr@{}}
\toprule
 & \multicolumn{3}{c}{Forward only} & \multicolumn{3}{c}{Forward + backward} \\
\cmidrule(lr){2-4}\cmidrule(lr){5-7}
Model & BS=64 & BS=512 & BS=4096 & BS=64 & BS=512 & BS=4096 \\
\midrule
M01 Coulomb & 0.30 & 0.46 & 0.39 & 0.40 & 0.40 & 0.41 \\
M02 Beer--Lambert & 0.29 & 0.24 & 0.24 & 0.32 & 0.31 & 0.33 \\
M03 Michaelis--Menten & 0.37 & 0.34 & 0.34 & 0.49 & 0.51 & 0.49 \\
M04 Arrhenius & 0.58 & 0.49 & 0.48 & 0.54 & 0.56 & 0.55 \\
M05 Damped oscillator & 0.80 & 0.60 & 0.63 & 0.89 & 0.85 & 0.86 \\
M06 Logistic growth & 0.64 & 0.63 & 0.64 & 0.85 & 0.85 & 0.85 \\
M07 Power law & 0.39 & 0.37 & 0.39 & 0.80 & 0.79 & 0.80 \\
M08 Euler ODE & 0.90 & 0.86 & 0.86 & 0.94 & 0.93 & 0.94 \\
M10 SiLU & 0.69 & 0.64 & 0.65 & 0.68 & 0.84 & 0.84 \\
M11 Recursive filter & 0.87 & 0.87 & 0.86 & 0.94 & 0.94 & 0.98 \\
M12 Newton sqrt & 0.81 & 0.81 & 0.80 & 0.89 & 0.93 & 0.91 \\
M14 Anomaly scorer & 0.56 & 0.59 & 0.60 & 0.81 & 0.82 & 0.82 \\
\bottomrule
\end{tabular}
\end{table}

\texttt{torch.compile} (PyTorch~2.12, Inductor backend, \texttt{mode="reduce-overhead"}) provides \emph{no benefit} for any model at any batch size, with compiled execution $1.02$--$4.2\times$ \emph{slower} than the plain batched evaluator.

The reason is architectural: the compiled programs are small computation graphs (3--30 tensor operations). The plain batched evaluator already issues a short sequence of individually optimized CUDA kernels, and the per-kernel launch overhead is negligible on GPU hardware. \texttt{torch.compile}'s tracing, guard checks, and Triton code generation impose a fixed overhead that exceeds any savings from kernel fusion. The deepest graphs (M08 Euler ODE: 10 nested calls; M11 recursive filter: 8 iterations) approach parity at $0.86$--$0.98\times$ because they have more operations to fuse, but never break even.

This negative result is informative: it confirms that the batched evaluator is already near-optimal for GPU throughput on these program sizes. The remaining $14\times$ per-evaluation latency (Table~\ref{tab:overhead_decomp}) is dominated by tagged-value wrapping and Python dispatch costs that are amortized by batching but not eliminated by graph compilation. Reducing single-evaluation latency requires lowering to a compiled representation (e.g., MLIR) rather than applying \texttt{torch.compile} on top of the existing Python evaluator.

\subsection{DiffESM-S Large-Model Benchmark (Part~F)}

Table~\ref{tab:diffesm_spec} specifies the DiffESM-S model and Table~\ref{tab:exp_h_diffesm} reports its batched-throughput benchmark.

\begin{table}[ht]
\centering
\caption{DiffESM-S model specification.}
\label{tab:diffesm_spec}
\small
\begin{tabular}{@{}lr@{}}
\toprule
Property & Value \\
\midrule
Source lines (Scheme) & 317 \\
Compiled graph nodes  & 97 \\
Total inputs          & 95 \\
\quad Initial state variables & 20 \\
\quad Learnable parameters    & 70 \\
\quad Emission drivers        & 4 \\
\quad Control (timestep count)& 1 \\
Coupled modules       & 9 \\
Recursive timesteps   & 100 \\
Effective tensor ops  & ${\sim}9{,}700$ \\
\bottomrule
\end{tabular}
\end{table}

\begin{table}[ht]
\centering
\caption{Experiment~H Part~F: DiffESM-S benchmark (100 timesteps) on A100 node at batch size~64.}
\label{tab:exp_h_diffesm}
\small
\begin{tabular}{@{}llrrrrr@{}}
\toprule
 & & \multicolumn{2}{c}{Forward} & \multicolumn{2}{c}{Fwd+Bwd} & GPU \\
\cmidrule(lr){3-4}\cmidrule(lr){5-6}
Device & Method & Time (s) & ms/eval & Time (s) & ms/eval & MB \\
\midrule
\multirow{2}{*}{CPU} & Sequential & 11.25 & 175.8 & 30.85 & 482.0 & n/a \\
 & Batched    & 0.18  & 2.85  & 0.45  & 7.09  & n/a \\
\addlinespace
\multirow{2}{*}{GPU} & Sequential & 23.65 & 369.5 & 62.59 & 978.0 & n/a \\
 & Batched    & 0.38  & 5.92  & 0.94  & 14.73 & 5.5 \\
\bottomrule
\multicolumn{7}{@{}l}{\footnotesize Batch size = 64 for all rows. Full data: BS $\in \{1,4,16,64,256,1024\}$, steps $\in \{10,50,100\}$.}
\end{tabular}
\end{table}

At batch size~1, batched and sequential performance are equivalent (${\sim}176$~ms on CPU, ${\sim}370$~ms on GPU per evaluation).
As batch size increases, the batched evaluator amortizes the Python-level graph walk:
at batch size~64, CPU batched achieves a $60.5\times$ forward throughput gain over the single-evaluation baseline
(2.85~ms/eval vs.\ 172.7~ms/eval);
at batch size~1024, amortized cost drops to 0.20~ms/eval, an $875\times$ reduction.
(The table's $175.8$~ms/eval sequential figure is the sustained cost over its 64-evaluation run; throughput speedups use the batch-size-1 baseline, $172.7$~ms/eval, differing only by run-to-run timing jitter.)
Forward-plus-backward speedups track closely: $68.0\times$ at batch size~64.

A striking result is that \emph{CPU batched evaluation outperforms GPU batched} across all batch sizes.
At batch size~64, CPU batched completes in 0.18\,s vs.\ GPU batched at 0.38\,s ($2.1\times$ faster);
at batch size~1024, CPU finishes in 0.21\,s vs.\ GPU at 0.37\,s.
This inversion occurs because DiffESM-S's 97 graph nodes are individually small scalar operations
(addition, multiplication, log, exp, sigmoid).
The GPU must launch a separate CUDA kernel for each of the ${\sim}9{,}700$ effective operations,
and kernel launch overhead (${\sim}5$--$10\,\mu$s each) dominates the negligible arithmetic cost.
CPU execution avoids this overhead entirely, making PyTorch's vectorized CPU kernels faster for this workload.
GPU peak memory remains modest: 43\,MB at batch size~1024 with 100 timesteps, confirming that
memory is not the bottleneck.

This result has important implications for deployment: DMCI-compiled scientific models with
many small coupled operations should default to CPU batched evaluation rather than GPU,
unless individual operations are large enough (e.g., matrix multiplications) to amortize kernel launch costs.
The crossover point depends on per-node computation intensity, which for DiffESM-S is below the GPU advantage threshold.

% Figure: Exp H - Batched throughput scaling (log-log)
% DiffESM-S and DiffSoc-S, 100 timesteps, CPU forward pass
% Usage: \input{figures/fig_h_throughput_scaling.tex}
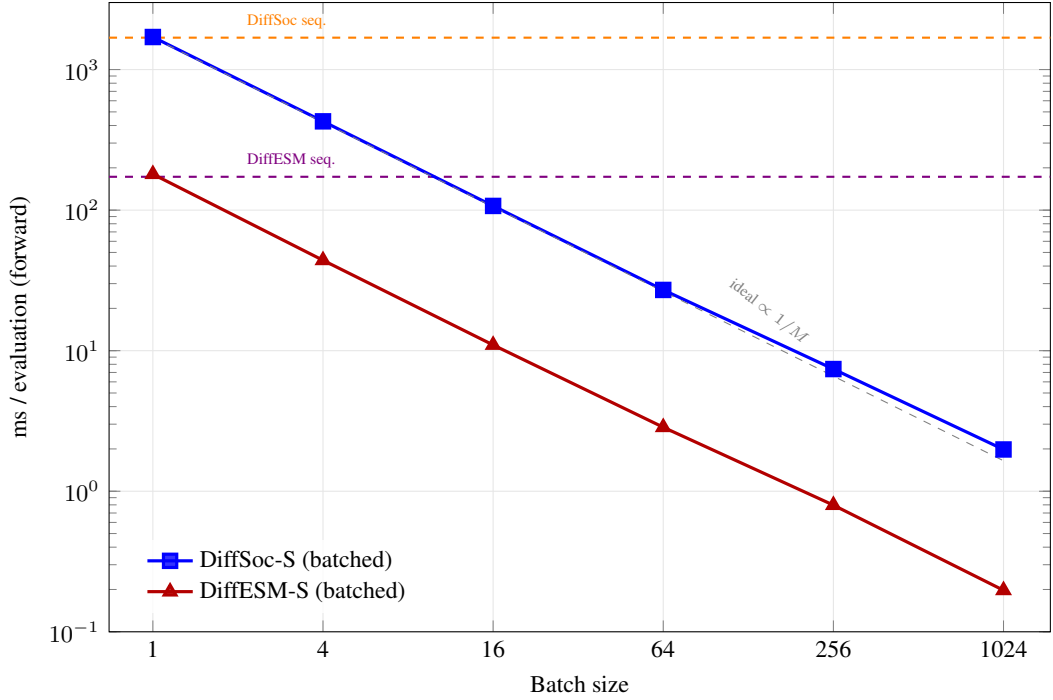
\begin{figure}[t]
\centering
\begin{tikzpicture}
\begin{loglogaxis}[
    width=0.85\columnwidth,
    height=0.60\columnwidth,
    xlabel={Batch size},
    ylabel={ms / evaluation (forward)},
    xmin=0.7, xmax=1500,
    ymin=0.1, ymax=3000,
    xtick={1,4,16,64,256,1024},
    xticklabels={1,4,16,64,256,1024},
    grid=major,
    grid style={gray!20},
    tick label style={font=\small},
    label style={font=\small},
    legend style={font=\small, at={(0.03,0.03)}, anchor=south west, draw=none, fill=white, fill opacity=0.8, text opacity=1},
    legend cell align={left},
]

% Sequential baselines (horizontal dashed lines)
\addplot[orange, dashed, thick, domain=0.7:1500, samples=2, forget plot] {1686.9506};
\node[font=\tiny, orange, anchor=south west] at (axis cs:2,1686.9506) {DiffSoc seq.};

\addplot[violet, dashed, thick, domain=0.7:1500, samples=2, forget plot] {172.7051};
\node[font=\tiny, violet, anchor=south west] at (axis cs:2,172.7051) {DiffESM seq.};

% DiffSoc-S batched
\addplot[blue, very thick, mark=square*, mark size=2.5pt]
    table[x=batch_size, y=ms_per_eval] {figures/exp_h_throughput_soc.dat};
\addlegendentry{DiffSoc-S (batched)}

% DiffESM-S batched
\addplot[red!70!black, very thick, mark=triangle*, mark size=2.5pt]
    table[x=batch_size, y=ms_per_eval] {figures/exp_h_throughput_esm.dat};
\addlegendentry{DiffESM-S (batched)}

% Ideal linear scaling reference (slope = -1 on log-log)
\addplot[gray, thin, dashed, domain=1:1024, samples=2] {1686.9506/x};
\node[font=\tiny, gray, anchor=south east, rotate=-38] at (axis cs:200,8.5) {ideal $\propto 1/M$};

\end{loglogaxis}
\end{tikzpicture}
\caption{Batched throughput scaling for DMCI-compiled models (100 timesteps, CPU forward pass). Both DiffSoc-S and DiffESM-S achieve near-linear scaling: ms/eval decreases proportionally to batch size, closely tracking the ideal $1/M$ line. Horizontal dashed lines show the sequential baseline (constant per-evaluation cost regardless of batch size).}
\label{fig:exp_h_throughput}
\end{figure}

\subsection{Convergence Speedup (Part~G)}

Table~\ref{tab:convergence} reports the convergence-speedup comparison.

\begin{table}[ht]
\centering
\caption{Experiment~H Part~G: Convergence speedup on DiffESM-S (20 timesteps, 15 params, $N{=}32$, 300 epochs, CPU).}
\label{tab:convergence}
\small
\begin{tabular}{@{}lrrrr@{}}
\toprule
Condition & Time (s) & Final loss & Conv.\ & vs.\ Seq.\ \\
\midrule
Sequential        & 731.4 & ${<}10^{-8}$ & 1/1   & $1.0\times$ \\
Batched           &  23.9 & ${<}10^{-8}$ & 1/1   & $30.6\times$ \\
\addlinespace
Population $M{=}1$   &  23.8 & $2.0{\times}10^{-8}$ & 1/1   & n/a \\
Population $M{=}10$  &  26.8 & ${<}10^{-8}$ & 10/10 & n/a \\
Population $M{=}50$  &  29.9 & ${<}10^{-8}$ & 50/50 & n/a \\
Population $M{=}200$ &  38.0 & ${<}10^{-8}$ & 200/200 & n/a \\
\bottomrule
\multicolumn{5}{@{}l}{\footnotesize Conv.\ = restarts reaching loss ${<}10^{-3}$ by epoch 300.} \\
\multicolumn{5}{@{}l}{\footnotesize $200$ sequential restarts: $146{,}279$\,s (40.6\,h) vs.\ population: $38.0$\,s ($3{,}848\times$ speedup).}
\end{tabular}
\end{table}

Sequential and batched training produce \emph{identical} optimization trajectories (the same loss at each epoch to floating-point precision) because both compute the same mean-squared-error gradient over the same $N{=}32$ samples.
The only difference is wall-clock time: batched completes $300$ epochs in $23.9$\,s versus $731.4$\,s for sequential, a $30.6\times$ speedup.
When batched training finishes all $300$ epochs, sequential has completed only $9$ epochs and has $10\times$ higher loss ($1.07 \times 10^{-3}$ vs.\ ${<}10^{-8}$).
This is not a statistical improvement but a direct consequence of replacing $N$ Python-level graph walks with one batched evaluation.

Population batching demonstrates the most dramatic advantage.
By packing $M$ random parameter initializations $\times$ $N$ data points into a single batch of size $M \cdot N$,
DMCI simultaneously optimizes $M$ independent restarts in one forward--backward pass.
Scaling from $M{=}1$ to $M{=}200$ increases wall-clock time by only $1.6\times$ (from $23.8$\,s to $38.0$\,s),
while exploring $200\times$ more of the parameter landscape.
All $200$ restarts converge (loss~${<}10^{-3}$) by epoch~$75$.
The equivalent sequential computation ($200$ independent training runs of $300$ epochs each) would
require $200 \times 731.4 = 146{,}279$\,s ($40.6$~hours).
Population batching completes the same exploration in $38.0$\,s,
a \textbf{3{,}848$\boldsymbol{\times}$} speedup.

This sub-linear scaling arises because the Python-level graph walk (the dominant cost at small batch sizes)
is executed once regardless of batch size.
Increasing $M$ from $1$ to $200$ multiplies the tensor dimension from $32$ to $6{,}400$,
but PyTorch's vectorized CPU kernels scale the arithmetic cost far less than linearly
for the small per-node operations in DiffESM-S.
The practical consequence is that multi-restart global optimization, typically
the most expensive phase of scientific model calibration, becomes nearly free
relative to a single training run.

% Figure: Exp H - Convergence wall-clock time vs population size
% Shows sub-linear scaling: nearly flat line from M=1 to M=200
% Usage: \input{figures/fig_h_convergence_walltime.tex}
\begin{figure}[t]
\centering
\begin{tikzpicture}
\begin{axis}[
    width=0.85\columnwidth,
    height=0.55\columnwidth,
    xlabel={Population size $M$},
    ylabel={Total wall-clock time (s)},
    xmin=0, xmax=220,
    ymin=0, ymax=800,
    xtick={1,10,50,100,200},
    grid=major,
    grid style={gray!20},
    tick label style={font=\small},
    label style={font=\small},
    legend style={font=\small, at={(0.97,0.97)}, anchor=north east, draw=none, fill=white, fill opacity=0.8, text opacity=1},
    legend cell align={left},
]

% Sequential baseline (horizontal dashed line at 731.4s)
\addplot[orange, dashed, thick, domain=0:220, samples=2] {731.4};
\addlegendentry{Sequential ($M{=}1$)}

% Population batched data
\addplot[blue, very thick, mark=square*, mark size=2.5pt]
    table[x=pop_size, y=wall_time] {figures/exp_h_convergence_walltime.dat};
\addlegendentry{Batched}

% Annotations
\draw[<->, thick, gray] (axis cs:150,38.0) -- node[right, font=\scriptsize, gray, align=left] {$19\times$\\faster} (axis cs:150,731.4);

% Label the nearly flat batched region
\node[font=\scriptsize, blue!70!black, anchor=south] at (axis cs:100,42) {23.8s $\to$ 38.0s};

\end{axis}
\end{tikzpicture}
\caption{Wall-clock time for 300 training epochs on DiffESM-S as population size $M$ increases. Batched evaluation (blue) scales sub-linearly: increasing the population from $M{=}1$ to $M{=}200$ adds only 60\% wall time (23.8\,s $\to$ 38.0\,s), compared to the sequential baseline at 731.4\,s for $M{=}1$ alone. This enables population-based training at near-zero marginal cost.}
\label{fig:exp_h_convergence}
\end{figure}
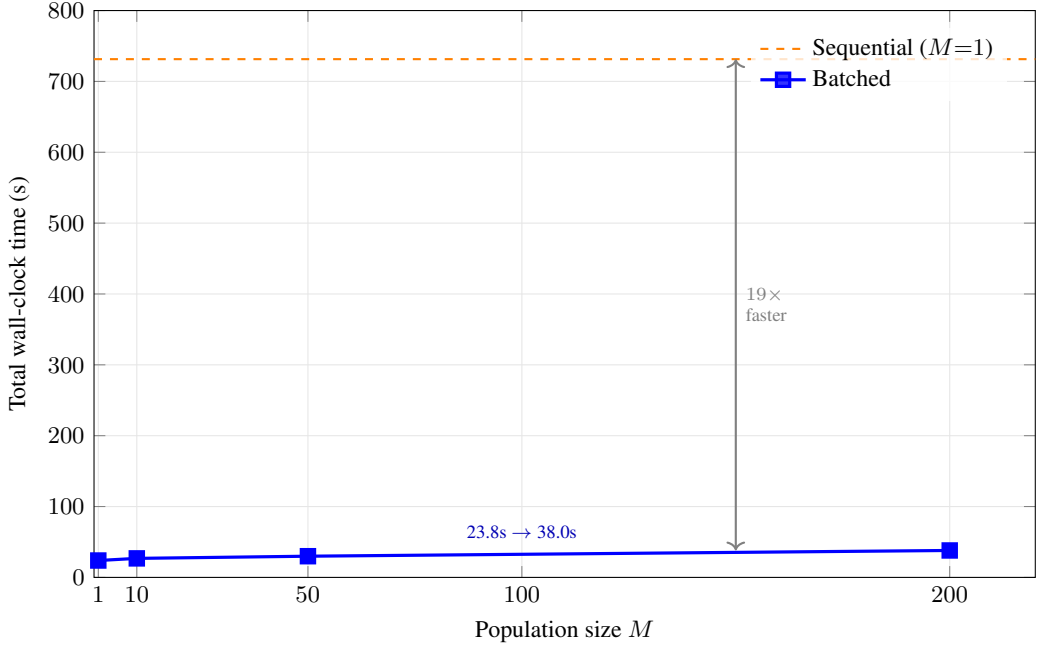

\subsection{DiffSoc-S Large-Model Benchmark (Part~H)}

Table~\ref{tab:diffsoc_spec} specifies the DiffSoc-S model, and Tables~\ref{tab:exp_h_diffsoc} and~\ref{tab:exp_h_diffsoc_scaling} report its batched benchmark and batch-size scaling.

\begin{table}[ht]
\centering
\caption{DiffSoc-S model specification.}
\label{tab:diffsoc_spec}
\small
\begin{tabular}{@{}lr@{}}
\toprule
Property & Value \\
\midrule
Source lines (Scheme) & 703 \\
Compiled graph nodes  & 206 \\
Total inputs          & 204 \\
\quad Initial state variables & 96 \\
\quad Learnable parameters    & 87 \\
\quad Exogenous drivers       & 5 \\
\quad Structural constants    & 15 \\
\quad Control (timestep count)& 1 \\
Coupled modules       & 10 \\
\bottomrule
\end{tabular}
\end{table}

\begin{table}[ht]
\centering
\caption{Experiment~H Part~H: DiffSoc-S benchmark on A100 node at batch size~64.}
\label{tab:exp_h_diffsoc}
\small
\begin{tabular}{@{}llrrrrrr@{}}
\toprule
 & & \multicolumn{3}{c}{Forward} & \multicolumn{3}{c}{Fwd+Bwd} \\
\cmidrule(lr){3-5}\cmidrule(lr){6-8}
Device & Method & 10-step & 50-step & 100-step & 10-step & 50-step & 100-step \\
\midrule
\multirow{2}{*}{CPU} & Sequential (s) & 12.05 & 56.37 & 111.91 & 31.01 & 155.21 & 318.12 \\
 & Batched (s)    & 0.19  & 0.88  & 1.73   & 0.45  & 2.26   & 4.58 \\
\addlinespace
\multirow{2}{*}{GPU} & Sequential (s) & 24.93 & 116.61 & 230.74 & 74.56 & 384.79 & 630.28 \\
 & Batched (s)    & 0.39  & 1.84  & 3.66   & 1.14  & 5.81   & 9.45 \\
\bottomrule
\end{tabular}
\end{table}

\begin{table}[ht]
\centering
\caption{Experiment~H Part~H: DiffSoc-S throughput scaling (100 timesteps).}
\label{tab:exp_h_diffsoc_scaling}
\small
\begin{tabular}{@{}rrrrrr@{}}
\toprule
 & \multicolumn{2}{c}{CPU} & \multicolumn{3}{c}{GPU} \\
\cmidrule(lr){2-3}\cmidrule(lr){4-6}
Batch size & ms/eval & Speedup & ms/eval & Speedup & MB \\
\midrule
1 (seq.)  & 1687.0 & $1.0\times$    & 3577.9 & $1.0\times$    & n/a \\
1         & 1705.1 & $1.0\times$    & 3624.0 & $1.0\times$    & 43 \\
4         &  428.3 & $3.9\times$    &  915.8 & $3.9\times$    & 43 \\
16        &  107.1 & $15.7\times$   &  228.5 & $15.7\times$   & 43 \\
64        &   27.1 & $62.3\times$   &   57.1 & $62.6\times$   & 43 \\
256       &    7.4 & $\mathbf{228\times}$  &   14.4 & $\mathbf{249\times}$  & 86 \\
1024      &    2.0 & $\mathbf{851\times}$  &    3.6 & $\mathbf{995\times}$  & 344 \\
\bottomrule
\multicolumn{6}{@{}l}{\footnotesize Forward-only ms/eval (batched). Fwd+bwd speedups are similar ($70\times$ at BS=64, $880$--$1{,}256\times$ at BS=1024).}
\end{tabular}
\end{table}

The batching pattern matches DiffESM-S precisely:
at batch size~64, CPU batched is $62$--$65\times$ faster than the same-batch sequential run across all timestep counts,
and forward-plus-backward speedup is $68$--$70\times$.
GPU speedup is similar ($63\times$ forward, $66$--$67\times$ fwd+bwd).
The CPU-over-GPU inversion observed in Part~F persists: CPU batched is $1.8$--$2.1\times$ faster than GPU batched across all configurations, confirming that models composed of many small operations favor CPU execution.

At larger batch sizes, DiffSoc-S achieves throughput speedups comparable to DiffESM-S.
At batch size~1024 on CPU, amortized forward latency drops to $1.98$~ms/eval for 100-step integration, an $851\times$ throughput improvement over sequential evaluation ($1{,}687$~ms/eval).
On GPU at batch size~1024, 100-step fwd+bwd drops to $3.59$~ms/eval with $344$~MB peak memory,
representing a $1{,}256\times$ throughput speedup over sequential GPU evaluation.
Scaling is nearly linear with batch size: $4\times$ at BS=4, $16\times$ at BS=16, $63\times$ at BS=64, with sub-linear gains continuing to BS=1024.

% Figure: Exp H - CPU vs GPU comparison (DiffSoc-S, 100 timesteps)
% Shows CPU-over-GPU inversion: CPU batched is faster than GPU batched
% Usage: \input{figures/fig_h_cpu_vs_gpu.tex}
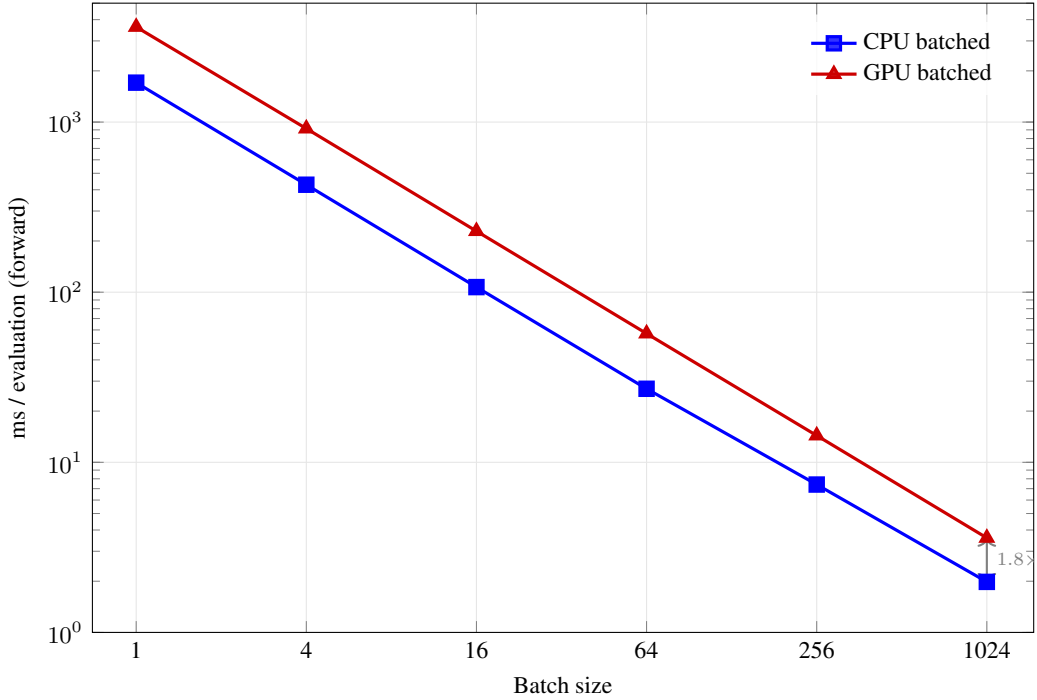
\begin{figure}[t]
\centering
\begin{tikzpicture}
\begin{loglogaxis}[
    width=0.85\columnwidth,
    height=0.60\columnwidth,
    xlabel={Batch size},
    ylabel={ms / evaluation (forward)},
    xmin=0.7, xmax=1500,
    ymin=1, ymax=5000,
    xtick={1,4,16,64,256,1024},
    xticklabels={1,4,16,64,256,1024},
    grid=major,
    grid style={gray!20},
    tick label style={font=\small},
    label style={font=\small},
    legend style={font=\small, at={(0.97,0.97)}, anchor=north east, draw=none, fill=white, fill opacity=0.8, text opacity=1},
    legend cell align={left},
]

% CPU batched
\addplot[blue, very thick, mark=square*, mark size=2.5pt]
    table[x=batch_size, y=ms_per_eval] {figures/exp_h_cpugpu_soc_cpu.dat};
\addlegendentry{CPU batched}

% GPU batched
\addplot[red!80!black, very thick, mark=triangle*, mark size=2.5pt]
    table[x=batch_size, y=ms_per_eval] {figures/exp_h_cpugpu_soc_gpu.dat};
\addlegendentry{GPU batched}

% Annotations
\draw[<->, thick, gray] (axis cs:1024,1.9825) -- node[right, font=\scriptsize, gray] {$1.8\times$} (axis cs:1024,3.5948);

\end{loglogaxis}
\end{tikzpicture}
\caption{CPU vs.\ GPU batched evaluation for DiffSoc-S (100 timesteps, forward pass). CPU consistently outperforms GPU across all batch sizes, with the GPU incurring ${\sim}2\times$ overhead from per-kernel launch latency (the interpreter issues many small operations rather than a few large ones). The interpreter's control-flow-heavy execution pattern does not benefit from GPU parallelism; all speedup comes from tensor-level batching within the CPU evaluator.}
\label{fig:exp_h_cpu_gpu}
\end{figure}

The large-model stress tests (Parts~F and~H) confirm these findings across two production-scale scientific models from different domains: DiffESM-S (97 nodes, climate) and DiffSoc-S (206 nodes, urban economics). Both achieve $60$--$62\times$ forward speedup at batch size~64 and $851$--$875\times$ at batch size~1024, with consistent CPU-over-GPU inversion at all tested batch sizes.

The convergence experiment (Part~G) makes this concrete: batched training completes $300$ epochs of DiffESM-S parameter fitting in $23.9$\,s, while sequential training reaches only $9$ epochs in the same wall-clock time. Population batching pushes this further: $200$ simultaneous random restarts complete in $38.0$\,s, a $3{,}848\times$ speedup over $200$ sequential runs. The overhead that appears prohibitive in single-evaluation benchmarks vanishes entirely in the optimization workloads where DMCI is actually used.

\subsection{Implications for the Overhead Narrative}

The per-evaluation overhead documented in Table~\ref{tab:overhead_decomp} (14$\times$ for sequential scalar evaluation) is real but misleading as a characterization of DMCI's practical performance. Like neural networks, DMCI's computation graph is a fixed structure through which batched data flows via tensor parallelism. The 14$\times$ overhead is a \emph{latency} cost (the time for a single scalar evaluation), not a \emph{throughput} cost.

At batch size 512 on a single GPU, the amortized per-evaluation overhead drops to $<0.002\times$ for equation-tier models: 512 evaluations complete in the same time as ${\sim}1$ sequential evaluation. For population-based workloads (Part~D), the amortized per-evaluation throughput cost falls to $<0.0001\times$ of a single sequential evaluation: DMCI evaluates 6{,}400 program instances in less time than sequential evaluation handles one. (This is an amortized throughput ratio, not a sub-unity per-evaluation latency; single-evaluation latency remains $14\times$ higher, as noted below.)

This does not eliminate the overhead; it redistributes it. The Python-level graph walking cost is paid once per batch, making it negligible at scale, but single-evaluation latency remains 14$\times$ higher than direct compilation. Workloads that are inherently sequential (interactive evaluation, debugging, small-batch inference) still pay the full per-evaluation cost.

\paragraph{Why near-linear scaling on a single CPU core.}
The ${\sim}850\times$ throughput gain at batch size~1024 on a single CPU core may appear paradoxical: the same core performs $1024\times$ more arithmetic, yet total wall time increases by only ${\sim}19\%$.
The explanation lies in where time is actually spent.
The overhead decomposition (Table~\ref{tab:overhead_decomp}) shows that raw arithmetic and PyTorch autograd account for $<0.2\%$ of sequential evaluation time; the remaining $>99.8\%$ is Python-level bookkeeping: tagged-value wrapping (41.4\%), Python interpreter dispatch (32.1\%), and graph walking (16.9\%).
In sequential mode, this bookkeeping is paid once per evaluation.
With $1{,}024$ evaluations, it is paid $1{,}024$ times.

With batching, the graph is walked \emph{once}.
Each of the 206 nodes (for DiffSoc-S) operates on a contiguous tensor of shape $(1024,)$ rather than a scalar.
The DiffSoc-S numbers make this concrete: a single 100-step forward pass takes $1{,}705$~ms;
$1{,}024$ batched evaluations take $2{,}030$~ms.
The additional $325$~ms is the actual arithmetic for $1{,}023$ extra evaluations.
The other $1{,}705$~ms is pure overhead, paid once regardless of batch size.
Even on a single core, the batched arithmetic is cheap because PyTorch's CPU kernels use SIMD vectorization (AVX-256 or AVX-512), processing 8--16 float32 values per instruction, and the contiguous memory layout of batched tensors yields favorable cache behavior compared to $1{,}024$ separately allocated scalars.

\paragraph{When GPU begins to pay off.}
At small batch sizes, GPU is \emph{slower} than CPU for these models: DiffSoc-S on GPU takes $3{,}624$~ms per evaluation at batch size~1 versus $1{,}705$~ms on CPU.
Each of the ${\sim}20{,}600$ effective tensor operations (206 nodes $\times$ 100 timesteps) launches a separate CUDA kernel, and kernel launch overhead (${\sim}5$--$10\,\mu$s each) dominates the negligible scalar arithmetic.
As batch size grows, the per-kernel arithmetic increases while launch overhead stays constant.
At batch size~1024, GPU forward-plus-backward time per evaluation ($9.26$~ms) is $1.8\times$ slower than CPU ($5.11$~ms) for DiffSoc-S, indicating that the crossover has not yet occurred for this model's operation granularity.
However, the GPU throughput \emph{speedup} relative to its own sequential baseline reaches $995\times$ at batch size~1024 (vs.\ $851\times$ for CPU), and the gap narrows further for backward passes at large batch sizes ($1{,}256\times$ GPU vs.\ $880\times$ CPU for 100-step fwd+bwd).
This reflects the GPU's massive arithmetic parallelism becoming relevant once per-kernel payloads are large enough: at batch size~1024, each kernel processes $1{,}024$ elements, which is sufficient to partially saturate the A100's compute units.
For models with larger per-node operations (matrix multiplications, convolutions, or longer vectors), the GPU crossover would occur at smaller batch sizes.
The practical guidance is clear: for scientific models composed of many small coupled scalar operations, CPU batching on a single core is both simpler and faster; GPU becomes advantageous when either the batch size or the per-operation tensor size is large enough to amortize kernel launch overhead.

\subsection{Why Gradient Descent Works for Nonlinear Programs}

The experimental results above demonstrate \emph{that} gradient-based optimization works for continuous parameters in compiled programs, but it is worth stating explicitly \emph{why} this is the natural approach, and where it breaks down.

\paragraph{Dimensionality.}
The fundamental challenge of parameter fitting in coupled scientific models is dimensionality.
DiffESM-S has 70 learnable parameters; DiffSoc-S has 87; even the small Experiment~A programs have up to 4 interacting constants.
Black-box methods (grid search, Bayesian optimization, evolutionary strategies) scale poorly with parameter count because they must explore the space without directional information.
Experiment~A quantifies this directly: the evolution strategy converges in only 45\% of runs (27/60) and fails catastrophically on recursive and composed programs (P3 loss~$>10^6$), while gradient descent converges in 100\% of runs across all program types.
Finite differences achieve 83\% convergence but require $12.5$--$534\times$ more wall time (mean ${\sim}195\times$), and their cost scales linearly with the number of parameters (one additional forward pass per parameter per gradient step).
Automatic differentiation computes the exact gradient of all parameters in a single backward pass regardless of parameter count, making it the only approach that scales to the 70--87 parameter models in Experiments~F and~G.

\paragraph{Compositionality of the chain rule.}
The programs compiled by DMCI are deeply nonlinear: DiffESM-S chains 9{,}700 tensor operations through recursive time integration with logarithmic, exponential, and sigmoidal nonlinearities.
The chain rule decomposes this end-to-end derivative into a product of local Jacobians, each of which is trivial (the derivative of addition, multiplication, $\exp$, $\log$, etc.).
PyTorch's autograd assembles these local derivatives automatically during the backward pass, regardless of how complex the overall composition is.
This is the same mechanism that enables training of deep neural networks with millions of parameters; DMCI simply applies it to compiled programs rather than learned architectures.
The gradient path analysis (Section~\ref{sec:exp_a}) confirms that the compiled interpreter adds 5--18 extra nodes to the autograd graph relative to direct compilation (roughly doubling the count for the simplest programs), yet gradients still match direct compilation to numerical precision, preserving gradient quality through arbitrarily deep compositions.

\paragraph{Landscape geometry.}
Gradient descent finds a local minimum, not necessarily a global one.
Whether this matters depends on the geometry of the loss landscape, which in turn depends on the nature of the search space.
TerpreT~\citep{gaunt2017terpret} demonstrated that gradient descent over \emph{discrete} program structure (choosing instructions, operands, or control flow) converges in only ${\sim}2\%$ of random restarts because the relaxed landscape is riddled with spurious local minima.
DMCI avoids this failure mode by fixing the program structure and optimizing only continuous parameters within it.
The program provides strong inductive bias: the functional form constrains the space of possible input--output mappings, and gradient descent need only find the constants that make the mapping fit the data.
The population batching experiment (Part~G) provides direct evidence that the resulting landscape is benign: 200 random restarts from $15\%$-perturbed initial conditions \emph{all} converge to loss~$<10^{-3}$ by epoch~75.
If the landscape contained significant local minima, a substantial fraction of restarts would become trapped.
The 100\% convergence rate across 200 restarts, combined with 100\% convergence across all 171 autograd program-seed pairs in Experiments~A--C, indicates that continuous parameter optimization in the tested fixed program structures produces well-behaved loss landscapes. Whether this holds for substantially larger or deeper programs remains an open question.

\paragraph{Where gradients fail.}
Two cases defeat gradient-based optimization in DMCI.
First, \emph{discrete parameters}: when the quantity to be learned is inherently discrete (which operator to apply, which branch to take), the gradient is either zero or undefined.
Experiment~E addresses this via Gumbel-Softmax relaxation, achieving $10.8\%$ recovery vs.\ $2.1\%$ for random search, better than chance but far from reliable, confirming TerpreT's finding that discrete program structure is fundamentally harder than continuous parameter fitting.
Second, \emph{branch-condition parameters}: when a learnable constant appears in a conditional test (e.g., \texttt{(< x $\alpha$)}), the comparison returns a discrete 0 or~1 with zero gradient in $\alpha$.
The S3 experiment confirms this produces a piecewise-constant loss landscape that gradient descent cannot navigate.
These limitations are inherent to the semantics of the source program, not artifacts of DMCI's compilation strategy.

% ============================================================================
\section{Experiment I: Differentiable Calibration of a Composite Ecosystem Model}
\label{app:exp_i}
% ============================================================================

\paragraph{Motivation.} Experiment~I stress-tests DMCI as a calibration engine on a realistic, multi-parameter \emph{composite} scientific model (a gross-primary-productivity / ecosystem model whose free-parameter count grows with the number of plant functional types), and uses it to ask how exact-gradient calibration through the compiled model \emph{scales} against gradient-free search as the parameter dimension grows. It probes the boundary of the paper's central claim (continuous parameter optimization within a fixed program structure) at parameter counts more than an order of magnitude beyond Experiments~A--C. An initial 12-parameter pilot was inconclusive (a short, non-batched single-start Adam run stalled while a gradient-free baseline reached a better optimum); here, with batched evaluation and an equal wall-clock budget, we sweep the dimension by more than an order of magnitude to characterize the scaling honestly.

\paragraph{Design.} The composite model is instantiated at $d\in\{6,12,18,24,36,48,66,96,126\}$ learnable parameters (1--21 plant functional types, six parameters each); we recover known parameters from generated data, scoring on a \emph{held-out} driver split to measure generalization rather than training fit, with three seeds per size. Three optimizers receive an \emph{equal} wall-clock budget ($300$\,s each per (size, seed)): \texttt{dmci} (single-start Adam through the batched compiled model), \texttt{dmci\_lbfgs\_ms} (the same exact gradients under a log-reparameterization with multi-start L-BFGS), and \texttt{diffevo} (gradient-free differential evolution). Batched evaluation (one interpreter walk over all data points per step) makes the high-dimensional sweep tractable, since per-fit cost is otherwise dominated by interpreter traversal rather than arithmetic.

\paragraph{Results.} Table~\ref{tab:exp_i} and Figure~\ref{fig:exp_i_heldout} report mean held-out MSE, and three patterns emerge. (i)~\emph{At small dimension} ($d\le12$) the problem is easy for black-box search: differential evolution and L-BFGS both reach near machine-zero at $d{=}6$, and differential evolution is best at $d{=}12$. (ii)~\emph{Exact-gradient calibration degrades the most gracefully with dimension}: single-start Adam through DMCI stays within a narrow $1.7\times10^{-2}$--$1.5\times10^{-1}$ held-out band across the \emph{entire} range and is the best method at every $d\ge18$, whereas differential evolution erodes and becomes erratic as $d$ grows (up to $4.9\times10^{-1}$ at $d{=}36$) and loses to Adam for all $d\ge36$. (iii)~the curvature-aware multi-start L-BFGS fitter is \emph{not} robust at scale: at a fixed $300$\,s budget it completes too few restarts and diverges in the log-reparameterized coordinates, with held-out error blowing up to $2.8$ ($d{=}48$) and $2.0$ ($d{=}96$), worse than every other method. We report this as a finding rather than tuning it away: the simplest exact-gradient optimizer is the most reliable here, while the more aggressive curvature-aware multistart does not scale at a fixed budget. The comparison is \emph{relative}: no method reaches the $10^{-3}$ convergence threshold beyond $d{=}6$, and with three seeds the curves are noisy and non-monotone.

\begin{table}[ht]
\centering
\caption{Experiment~I: mean held-out MSE by parameter count $d$ (3 seeds per cell), under an equal $300$\,s wall-clock budget per method. Single-start Adam through DMCI is the most robust to dimension (best at every $d\ge18$); differential evolution wins only at $d\le12$ and erodes thereafter; the multi-start L-BFGS fitter is erratic at high $d$ (it diverges at a fixed budget). Bold marks the best held-out error per row. No method reaches the $10^{-3}$ convergence threshold beyond $d{=}6$, so these are \emph{relative} robustness comparisons.}
\label{tab:exp_i}
\begin{tabular}{rccc}
\toprule
$d$ (\# params) & \texttt{dmci} (Adam) & \texttt{dmci\_lbfgs\_ms} & \texttt{diffevo} \\
\midrule
6   & $2.5\times10^{-2}$ & $1.8\times10^{-12}$ & $\mathbf{4.5\times10^{-14}}$ \\
12  & $4.6\times10^{-2}$ & $5.5\times10^{-2}$ & $\mathbf{4.6\times10^{-3}}$ \\
18  & $\mathbf{1.7\times10^{-2}}$ & $4.2\times10^{-2}$ & $4.0\times10^{-2}$ \\
24  & $\mathbf{4.8\times10^{-2}}$ & $2.7\times10^{-1}$ & $5.7\times10^{-2}$ \\
36  & $\mathbf{1.5\times10^{-1}}$ & $2.7\times10^{-1}$ & $4.9\times10^{-1}$ \\
48  & $\mathbf{1.5\times10^{-1}}$ & $2.8\times10^{0}$ & $2.9\times10^{-1}$ \\
66  & $\mathbf{8.4\times10^{-2}}$ & $5.2\times10^{-1}$ & $3.5\times10^{-1}$ \\
96  & $\mathbf{7.4\times10^{-2}}$ & $2.0\times10^{0}$ & $1.2\times10^{-1}$ \\
126 & $1.4\times10^{-1}$ & $\mathbf{1.4\times10^{-1}}$ & $2.0\times10^{-1}$ \\
\bottomrule
\end{tabular}
\end{table}

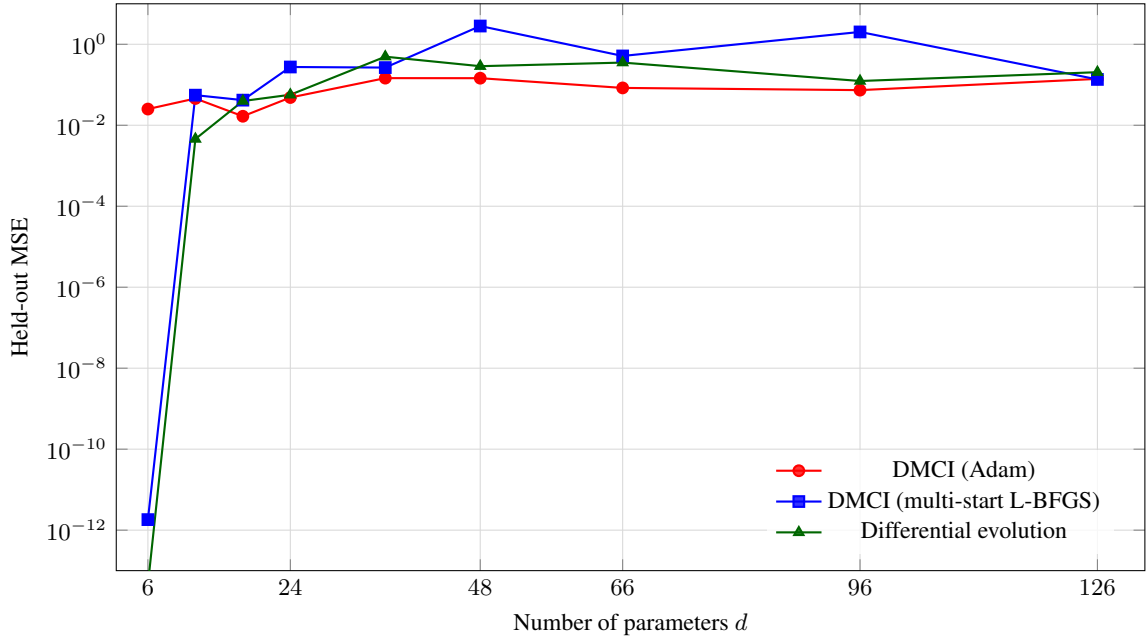
\begin{figure}[ht]
\centering
\begin{tikzpicture}
\begin{semilogyaxis}[
    width=0.92\columnwidth, height=0.55\columnwidth,
    xlabel={Number of parameters $d$}, ylabel={Held-out MSE},
    xtick={6,24,48,66,96,126}, xmin=2, xmax=132, ymin=1e-13, ymax=1e1,
    legend pos=south east,
    legend style={font=\small, draw=none, fill=white, fill opacity=0.8, text opacity=1},
    grid=major, grid style={gray!30},
    tick label style={font=\small}, label style={font=\small},
]
\addplot[red, thick, mark=*] table[x=nparams, y=dmci] {figures/exp_i_heldout.dat};
\addlegendentry{DMCI (Adam)}
\addplot[blue, thick, mark=square*] table[x=nparams, y=lbfgs] {figures/exp_i_heldout.dat};
\addlegendentry{DMCI (multi-start L-BFGS)}
\addplot[black!60!green, thick, mark=triangle*] table[x=nparams, y=diffevo] {figures/exp_i_heldout.dat};
\addlegendentry{Differential evolution}
\end{semilogyaxis}
\end{tikzpicture}
\caption{Experiment~I: held-out MSE versus parameter count $d$ (equal $300$\,s budget, 3 seeds, log scale). Single-start Adam through DMCI (red) is the most robust to dimension and best for all $d\ge18$; differential evolution (green) wins only at small $d$ and erodes as dimension grows; multi-start L-BFGS (blue) is erratic at high $d$, diverging at a fixed budget. No method reaches $10^{-3}$ beyond $d{=}6$, so the comparison is relative robustness, not absolute recovery.}
\label{fig:exp_i_heldout}
\end{figure}

\paragraph{Conclusion.} Across more than an order of magnitude in parameter dimension, exact-gradient calibration through the compiled model degrades the most gracefully: even single-start Adam holds a narrow held-out band and is the most reliable method at every $d\ge18$, while gradient-free differential evolution (competitive only at the smallest sizes) erodes as dimension grows. This is a qualified positive for DMCI as a calibration engine, with three honest caveats: (i)~black-box search still wins at low dimension ($d\le12$); (ii)~the curvature-aware multi-start L-BFGS fitter is \emph{not} robust at a fixed budget (it diverges at high $d$), so the simplest optimizer is preferable here; and (iii)~no method achieves true convergence beyond $d{=}6$, so this is a \emph{relative} robustness result, not a demonstration that high-dimensional composite calibration is solved. Practically: select on held-out (not training) error, and prefer a multi-optimizer portfolio (Experiment~F), since no single optimizer dominates across dimensionalities. This experiment is the quantitative core of the correctness-versus-optimization-success limitation (Section~\ref{sec:tradeoffs}).

% ============================================================================
\section{Experiment J: Program-Space Calibration, DMCI versus Compile-Each-Program}
\label{app:exp_j}
% ============================================================================

\paragraph{Motivation.} The preceding experiments optimize one program at a time. Experiment~J asks what happens when the object of optimization is a \emph{growing space of distinct, runtime-generated programs}, the regime LLM-driven discovery actually produces. It compares one compiled interpreter (DMCI) against the two compile-each-program workflows a practitioner would otherwise use: \textbf{B1}, automatic translation of each program to JAX via \texttt{sympy.lambdify} followed by \texttt{jax.jit(jax.grad)}; and \textbf{B2}, a hand-port of each program to a jitted JAX function. The experiment quantifies the \emph{workflow economics} behind the ``programs as data'' claim along four axes (amortized compile cost, coverage, per-program implementation burden, and matched recovery accuracy) and tests external validity on genuine LLM output.

\paragraph{Design.} A reproducible corpus of $N$ \emph{structurally distinct} programs (random closed-form expression trees, and damped-Euler recursive relaxations) is generated deterministically, with the recursive fraction set to $0\%$ or $100\%$ and $N\in\{1,100,10^4\}$. All three arms run on the identical corpus. We measure: cumulative compile/trace time; coverage (the fraction each arm can evaluate with no human intervention); per-structure engineering (a lines-of-code proxy); and \emph{matched} parameter recovery; every arm uses the same scipy L-BFGS-B multi-start driver under an identical budget, so the comparison is about cost and coverage, \emph{not} accuracy. A separate LLM-validation subset runs the same arms on 260 programs authored by MindRouter (Qwen3.6-27B with reasoning) and cached for reproducibility.

\paragraph{Results (synthetic).} Table~\ref{tab:exp_j} and Figure~\ref{fig:exp_j_compile} summarize the four axes. DMCI compiles the interpreter \emph{once} ($0.02$\,s, a $283$\,KB portable artifact) and then runs every program; the compile-each-program arms pay per program, reaching $592$\,s (B1) and $618$\,s (B2) at $N=10^4$ closed-form programs, a flat-versus-linear crossover that is the core result. Coverage is uniform for DMCI ($100\%$ on both families) and the hand-port ($100\%$), but \texttt{lambdify} \emph{cannot ingest recursion}: B1 covers $0\%$ of the recursive family. Per-structure engineering is zero for DMCI versus $56{,}361$ (closed-form) and $78{,}005$ (recursive) lines for B2 at $N=10^4$. Under the matched optimizer, recovery accuracy is at \emph{parity} across arms (closed-form mean relative parameter error $\approx0.44$--$0.47$; recursive $\approx0.36$--$0.38$), confirming the comparison isolates cost and coverage, not accuracy.

\begin{table}[ht]
\centering
\caption{Experiment~J (synthetic): workflow economics at $N=10^4$ distinct programs. DMCI compiles the interpreter once and then evaluates each program as data, so cumulative compile cost remains constant. The compile-each-program baselines pay per program: B1 (lambdify-to-JAX) covers closed-form expressions but cannot ingest recursive programs, while B2 (hand-port-to-JAX) covers both families only with substantial per-program implementation effort. Because matched-recovery results are held at parity in the accompanying experiment, this table isolates workflow cost and coverage rather than optimization accuracy. Dashes indicate not applicable because the method has zero coverage for that program family.}
\label{tab:exp_j}
\setlength{\tabcolsep}{4pt}
\begin{tabular}{llccc}
\toprule
Family & Metric ($N{=}10^4$) & DMCI & B1 (lambdify-to-JAX) & B2 (hand-port-to-JAX) \\
\midrule
\multirow{3}{*}{closed-form}
 & cumulative compile (s) & \textbf{0.02} & 591.6 & 617.7 \\
 & coverage               & 100\% & 100\% & 100\% \\
 & per-program implementation burden (LOC)      & \textbf{0} & 0 & 56{,}361 \\
\midrule
\multirow{3}{*}{recursive}
 & cumulative compile (s) & \textbf{0.02} & n/a & 1287.1 \\
 & coverage               & \textbf{100\%} & \textbf{0\%} & 100\% \\
 & per-program implementation burden (LOC)      & \textbf{0} & n/a & 78{,}005 \\
\bottomrule
\end{tabular}
\end{table}

\begin{figure}[ht]
\centering
\begin{tikzpicture}
\begin{loglogaxis}[
    width=0.92\columnwidth, height=0.55\columnwidth,
    xlabel={Number of distinct programs $N$}, ylabel={Cumulative compile time (s)},
    xmin=0.8, xmax=2e4, ymin=1e-2, ymax=1e3,
    legend pos=north west,
    legend style={font=\small, draw=none, fill=white, fill opacity=0.8, text opacity=1},
    grid=major, grid style={gray!30},
    tick label style={font=\small}, label style={font=\small},
]
\addplot[blue, thick, mark=*] table[x=N, y=dmci] {figures/exp_j_compile.dat};
\addlegendentry{DMCI (one interpreter compile)}
\addplot[red, thick, mark=square*] table[x=N, y=b1] {figures/exp_j_compile.dat};
\addlegendentry{B1: lambdify-to-JAX}
\addplot[black!60!green, thick, mark=triangle*] table[x=N, y=b2] {figures/exp_j_compile.dat};
\addlegendentry{B2: hand-port to JAX}
\end{loglogaxis}
\end{tikzpicture}
\caption{Experiment~J (synthetic, closed-form family): cumulative compile time versus the number of distinct programs. DMCI's single interpreter compile is flat ($0.02$\,s) while the compile-each-program arms grow linearly, reaching ${\sim}600$\,s at $N=10^4$. The crossover is immediate and the gap widens without bound.}
\label{fig:exp_j_compile}
\end{figure}
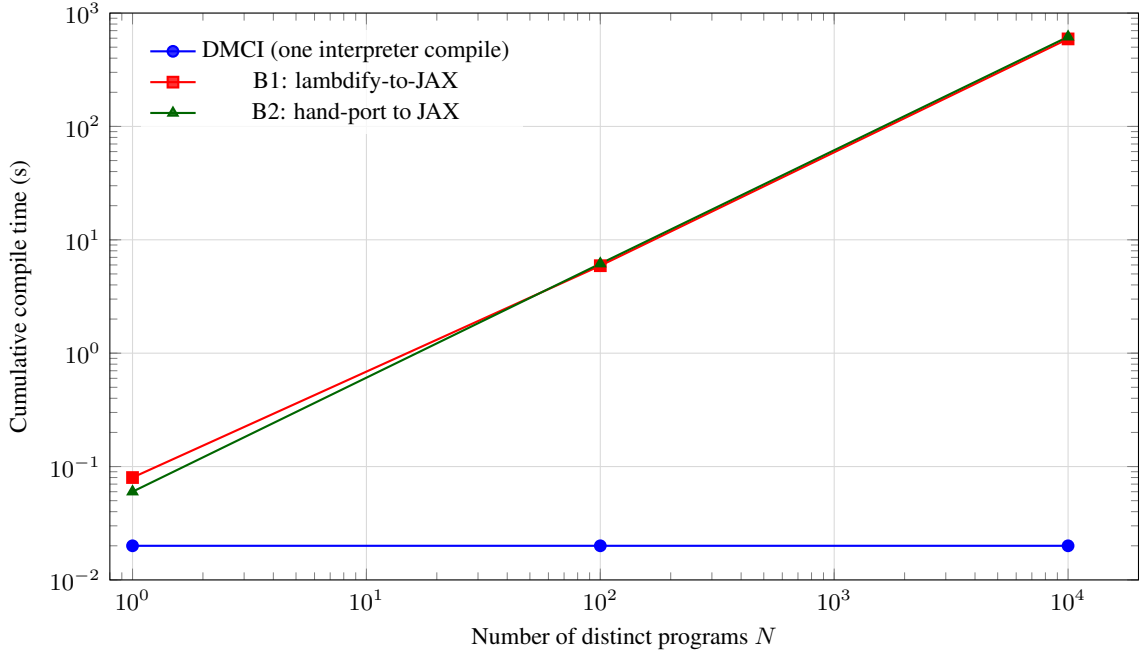

\paragraph{Results (LLM-validation).} On 260 genuine LLM-authored programs the synthetic story reproduces. Coverage: DMCI $100\%$ on both families; B1 $99\%$ on closed-form (one real program \texttt{lambdify} could not ingest) and $0\%$ on recursive; B2 $100\%$. Cumulative compile over all 260: DMCI $0.02$\,s (flat) versus B1 $23.2$\,s and B2 $23.9$\,s, with B2 engineering $5{,}034$ lines. Matched recovery on the closed-form family is again at parity (mean relative error $\approx0.088$ for all three arms). The honest negative is recursive recovery: within the matched budget, DMCI's optimizer fails to recover the constants of the LLM-authored recursive programs (effectively infinite error); these programs present the same hard, often ill-identified landscapes seen in Experiments~F (F3) and~I, so recursive \emph{recovery} remains an open optimization problem even though recursive \emph{coverage} is solved.

\paragraph{Conclusion.} When the object of optimization is a space of distinct runtime-generated programs, one compiled interpreter wins on the axes that scale: amortized compile (one compile versus $N$), uniform $100\%$ coverage including recursion that automatic translation cannot reach, and zero per-structure engineering; and these advantages reproduce on real LLM output. Two honest boundaries frame the result: per-evaluation latency favors compiled JAX (Section~\ref{sec:exp_h}, and the \texttt{jax.vmap} comparison therein), so DMCI's advantage is amortization, coverage, and zero engineering rather than raw speed; and parameter \emph{recovery} on hard recursive programs remains limited by optimization, not by the engine.
% ============================================================================
\section{Experiment L (Battery): Co-Search, Synthetic Recovery, and Real-Data Forecast}
\label{app:exp_battery_ext}
% ============================================================================
% All numbers FINAL: synthetic battery_rescore.json; real battery_rescore_real_ks*.json + knee + ablate_inner.
% results/battery_rescore_real.json after the three real islands + iters=300 re-score.

This appendix details the battery capacity-fade co-search of \S\ref{sec:exp_battery}: the differentiable interpreter and held-out scoring, the OpenEvolve outer loop, the synthetic ground-truth recovery, the real Severson data pipeline and the grid choice it forces, and the rigorous re-scoring protocol that anchors both the battery claim and the FluZoo stress test.

\subsection{Differentiable interpreter and held-out scoring}
Each candidate is a two-form Scheme artifact: a \texttt{(params ...)} schema and a per-cycle \texttt{(loop ... (recur ...))} rollout that carries the predicted capacity in a \texttt{yhat} loop variable and accumulates a floored-variance Gaussian negative log-likelihood of the observed capacity into \texttt{L}. The \emph{same} frozen interpreter folds the NLL for fitting and, via a single base-case token swap (\texttt{L}$\to$\texttt{yhat}), rolls the fitted autonomous model forward to read held-out capacity, the identical FIT$\to$PREDICT mechanism used in the LIM-ENSO and FluZoo studies. All cells are stacked into one batched interpreter walk with per-cell parameter vectors, so wall-clock is ${\sim}$constant in the cell count; this batching is mandatory, as a single forward+backward over a $60$-cycle rollout costs ${\sim}1.7$\,s on one CPU core. The held-out score is the mean RMSE between predicted and observed state-of-health over a strided set of held-out cycles: structures are fit on cycles $[0,k_{\text{split}})$ and scored on $[k_{\text{split}},T)$ with $T{=}100$.

\subsection{Outer loop (OpenEvolve)}
The structure search runs through OpenEvolve~\citep{openevolve}, an open implementation of AlphaEvolve~\citep{alphaevolve}: an LLM ensemble (Qwen3.6-35B/27B) proposes diff edits to the Scheme rollout, a MAP-Elites archive maintains quality-diversity over (structural complexity $\times$ knee-capability), and three island populations evolve in parallel ($200$ iterations each, periodic migration). The evaluator gates each candidate through a validity funnel before scoring: parse, whitelisted-operator prescan (against the silent-zero footgun where an unsupported head evaluates to $0$ and corrupts the model), compile, finite and nonzero gradient, stable rollout, and forecastable; an invalid program returns a targeted error that the next prompt repairs. Search-time fitness uses a cheap $60$-step Adam calibration; island winners are re-scored at $300$ steps.

\subsection{Synthetic recovery (ground truth)}
The synthetic target is $12$ mechanism-labeled cells generated from the two-reservoir bottleneck $Q=q_0\min(1-a\sqrt{k},\,1-ck)$ (knee near cycle $55$, coulometry-grade noise), and the search is seeded with the smooth $Q=q_0-B\sqrt{k}$. Two of the three islands converge to $\min(q_0-B_1\sqrt{k},\,q_0-B_2 k)$, algebraically the generating family and structurally knee-capable, from a seed that is not; the third finds a more elaborate power-law two-reservoir. A canonical-AST funnel over the retained archive (positional symbol renaming, numeric-literal bucketing on the loop body; \texttt{openevolve/funnel.py}) quantifies the discrete search: the three islands evaluated $600$ programs spanning $119$ structurally distinct rollouts ($48/42/65$ per island), of which exactly one (the smooth $\sqrt{t}$ seed) coincides with a reference family while all three converged winners are structurally \emph{novel} ASTs. The recovery claim is therefore anchored to \emph{algebraic identity}, not an AST hash or a token flag: islands~0 and~2 converged to $\min(q_0-B_1\sqrt{k},\,q_0-B_2 k)$, which is the generating $q_0\min(1-a\sqrt{k},\,1-ck)$ family up to algebra and parameter renaming yet a different syntax tree: the search wrote its own encoding of the mechanism rather than copying the reference. (The coarse MAP-Elites \emph{knee-capable} descriptor, a token-level \texttt{min}/\texttt{max}/\texttt{exp} test, only drives archive diversity and is not the recovery evidence.) Anchoring to algebraic identity is why the recovery is robust to the fitness-fidelity effect below. At the rigorous $300$-step budget the recovered structure forecasts the held-out tail at RMSE $0.0111$ versus $0.0301$ for the smooth seed, $0.0143$/$0.0149$ for the hand two-reservoir/sigmoidal references (Table~\ref{tab:battery}, top); the same structure scores $0.0067$ at the cheap $60$-step budget, an optimistic figure whose interpretation is the subject of the protocol below.

\subsection{Real Severson pipeline and the grid choice}
The real target is the \citet{severson2019} commercial A123 LFP/graphite fast-charging dataset ($124$-cell modeling subset, DOI \texttt{10.1038/s41560-019-0356-8}); we take per-cycle discharge capacity and normalize each cell to state of health $Q/Q_0$ ($Q_0$ the cycle-$2$ capacity), retaining $117$ cells after dropping channels with non-finite or anomalous capacity (the dataset's known noisy cells). Severson cells fade negligibly over their first ${\sim}100$ cycles and accelerate only near end of life, so the grid choice is critical: an absolute first-$T$-cycle window is essentially flat (mean held-out fade ${<}0.002$ SOH, no family separable), whereas resampling each cell's trajectory onto a common $T{=}100$-point grid spanning its full observed cycle life exposes the flat-then-soft-rollover shape, with the held-out tail fading $3.7\times$ (late split, $k_{\text{split}}{=}70$) to $11.8\times$ (early split, $k_{\text{split}}{=}45$) more than the fit window. We report both splits; the early split is the honest extrapolation case (the fit window ends before the rollover). Baselines are the smooth structure families (square-root SEI and power-law) and the knee-capable hand models (a stretched-exponential soft knee and a sigmoidal knee; \texttt{experiments/exp\_battery/structures.py} places stretched-exponential on the knee side of its smooth-to-knee axis), all calibrated through the \emph{same} interpreter, plus naive persistence and linear-extrapolation floors. Real-data held-out RMSE (late/early): persistence $0.084$/$0.089$, linear extrapolation $0.080$/$0.078$, best smooth family $0.083$/$0.086$, stretched-exponential soft knee $0.069$/$0.065$, best hand knee (sigmoidal) $0.051$/$0.057$, and the co-search-selected program (island2, a discovered three-reservoir bottleneck) $0.051$/$0.034$ (Table~\ref{tab:battery}, bottom); the selected program beats the best hand model by $1.7\times$ on the early-extrapolation split and ties on the late split, and the search explored $160$ structurally distinct rollouts (canonical-AST count over $601$ programs across three islands). \emph{Inner-fitter ablation} (\texttt{openevolve/ablate\_inner.py}): holding each discovered structure, the batched forward, the per-cell parameterization, and a matched wall-clock budget fixed and replacing exact-gradient Adam with gradient-free differential evolution gives held-out RMSE (DE vs.\ DMCI, late/early) island0 $1.98$/$1.76$ vs.\ $0.072$/$0.065$, island1 $0.25$/$0.21$ vs.\ $0.027$/$0.167$, and island2 $0.90$/$0.92$ vs.\ $0.051$/$0.034$ over $468$--$702$ free parameters; the gradient-free fit is $9$--$27\times$ worse on the identical structures, so DMCI's exact-gradient calibration, not the structure search alone, realizes the held-out skill. Following the original study, ``knee'' is downstream terminology (Severson reports a ``nonlinear degradation process'' that ``accelerates near the end of life''~\citep{severson2019}), and the soft rollover's physical driver (lithium plating after negative-electrode active-material loss~\citep{attia2022knees}) is latent in a capacity-only curve, so on real cells we claim forecast skill, not mechanism recovery.

\subsection{The fitness-fidelity protocol and the FluZoo stress test}
The cheap search-time fitness ($60$ Adam steps) is an under-converged proxy that acts as an implicit regularizer, flattering more flexible structures; the outer search can therefore optimize the proxy rather than true held-out skill. We guard against this by re-scoring all island winners and baselines at a converged $300$-step budget. On battery the structural conclusion is unchanged (the synthetic claim rests on the knee-capable label against ground truth, and the recovered structure still leads at $300$ steps). On influenza (FluZoo, Appendix~\ref{app:exp_fluzoo_ext}) the same re-scoring \emph{erases} the evolved program's apparent advantage, leaving it tied with a hand-written regional SEIR on validation ($0.0148$ vs.\ $0.0149$) and test ($0.0181$ vs.\ $0.0181$). This is the cautionary half of the co-search story: program-and-parameter co-search recovers correct structure when the data discriminate it and degrades to the baseline when they do not, and only a converged inner fitness reveals which case obtains.

\subsection{Reproducibility}
The harness is in \texttt{experiments/exp\_battery/} (synthetic generation \texttt{synth}/\texttt{gen\_target}, reference structures, batched DMCI scorer \texttt{oe\_score}, the OpenEvolve driver \texttt{openevolve/run\_battery}, and the rigorous re-scores \texttt{rescore}/\texttt{rescore\_real}); the real Severson loader (\texttt{openevolve/gen\_target\_real}) ingests per-cell capacity-vs-cycle and emits the $[N,T,1]$ state-of-health target. The dataset is publicly available from the source of~\citep{severson2019}.

% ============================================================================
\section{Design-Choice Ablations and Failure-Mode Summary}
\label{app:ablation}
% ============================================================================

For quick reference, this appendix consolidates why each significant design choice is necessary
(Table~\ref{tab:ablation}) and where the method breaks (Table~\ref{tab:failure_modes}). Both summarize
material argued in the main text and the per-experiment appendices; the final column of each gives the
supporting evidence.

\begin{table}[h]
\centering
\caption{Design-choice ablations. Removing or replacing any one choice forfeits a specific
capability; each row names the choice, the consequence of dropping it, and where the supporting evidence
appears.}
\label{tab:ablation}
\small
\begin{tabular}{@{}p{0.29\linewidth}p{0.47\linewidth}l@{}}
\toprule
Design choice removed & Consequence & Evidence \\
\midrule
Compile the interpreter (vs.\ direct compilation only) & No reuse across runtime-supplied programs: every new program must be recompiled or re-implemented & \S\ref{sec:exp_a}, \S\ref{sec:discussion}, App.~\ref{app:exp_j} \\
\addlinespace
Compile the evaluator from source (vs.\ a hand-coded PyTorch interpreter) & Loses the self-hosting / inherited-semantics claim; each language feature must be re-implemented and re-verified by hand & App.~\ref{app:handcoded} \\
\addlinespace
Batched evaluation & Sequential interpreter overhead (${\sim}14\times$ per evaluation) is not amortized, and the throughput gains vanish & \S\ref{sec:exp_h}, App.~\ref{app:exp_h_ext} \\
\addlinespace
Real-data case study & Validation would be synthetic-only, recovering self-generated targets & \S\ref{sec:exp_c}, App.~\ref{app:exp_lim_enso_ext} \\
\addlinespace
Trace-constant scoping of the theory & Gradient correctness would be overclaimed across branch boundaries, where the source program is itself non-differentiable & \S\ref{sec:theory}, App.~\ref{app:proofs} \\
\bottomrule
\end{tabular}
\end{table}

\begin{table}[h]
\centering
\caption{Failure-mode summary: where DMCI does not work, the consequence, and where it is documented.
These bound the central claim of continuous-parameter optimization within runtime-supplied program
structures.}
\label{tab:failure_modes}
\small
\begin{tabular}{@{}p{0.33\linewidth}p{0.43\linewidth}l@{}}
\toprule
Failure mode & Consequence & Where \\
\midrule
Parameter appears only in a branch condition & Zero gradient almost everywhere; inherits the source program's non-differentiability (not a compiler artifact) & \S\ref{sec:theory}, App.~\ref{app:exp_a_ext} (S3) \\
\addlinespace
Discrete structure search & Low success ($10.8\%$ vs.\ $2\%$ random) via Gumbel-Softmax; exploratory only & App.~\ref{app:exp_e} \\
\addlinespace
Single-evaluation latency & ${\sim}14\times$ slower than direct compilation; worse for one-at-a-time use & \S\ref{sec:exp_h}, App.~\ref{app:exp_h_ext} \\
\addlinespace
Interpreter on the JAX backend & Not yet supported: needs define-by-run autograd, with a fixed-length \texttt{lax.scan} rewrite as future work & \S\ref{sec:method}, App.~\ref{app:backends} \\
\addlinespace
\texttt{call/cc}, \texttt{set!}, exact integers, mutation & Unsupported by design (autograd purity, fixed compiled subset) & \S\ref{sec:discussion} \\
\addlinespace
Misspecification / hard recursive recovery & Recovery limited by optimization, not the engine; not fully characterized & \S\ref{sec:discussion}, App.~\ref{app:exp_j} \\
\bottomrule
\end{tabular}
\end{table}

\end{document}